\documentclass[preprint,aos]{imsart}

\usepackage{amsmath,amsfonts,amsthm,dsfont,mathtools,accents,tikz,float}
\usepackage[numbers]{natbib}
\RequirePackage[colorlinks,citecolor=purple,urlcolor=blue]{hyperref}

\numberwithin{equation}{section}
\makeatletter

\newcommand{\Rmnum}[1]{\expandafter\@slowromancap\romannumeral #1@}
\newcommand{\indep}{\rotatebox[origin=c]{90}{$\models$}}
\def\widebar{\accentset{{\cc@style\underline{\mskip10mu}}}}
\makeatother

\newcommand{\E}{\mathbb{E}}
\newcommand{\F}{\mathcal{F}}
\newcommand{\T}{\mathrm{T}}
\newcommand{\ds}{\,\mathrm{d}}
\newcommand{\e}{\varepsilon}
\newcommand{\rr}{\gamma}
\newcommand{\R}{\mathbb{R}}
\newcommand{\s}{\mathcal{S}}
\newcommand{\I}{\mathbf{I}}
\newcommand{\1}[1]{\mathds{1}_{\{#1\}}}

\theoremstyle{plain}
\newtheorem{assumption}{Assumption}
\newenvironment{assump}[1]
	{\renewcommand\theassumption{#1}\assumption}
	{\endassumption}

\newtheorem{thm}{Theorem}
\newtheorem{prop}{Proposition}
\newtheorem{corol}{Corollary}
\newtheorem{lem}{Lemma}
\newtheorem{defn}{Definition}

\theoremstyle{definition}
\newtheorem{remk}{Remark}
\newtheorem{eg}{Example}

\graphicspath{{pics/}}

\begin{document}

\begin{frontmatter}
\title{``Sound and Fury'': Nonlinear Functionals of Volatility Matrix in the Presence of Jump and Noise}
\runtitle{Nonlinear Functionals of Volatility Matrix: Jump and Noise}

\begin{aug}
	\author[A]{\fnms{Richard Y.} \snm{Chen}\ead[label=e1]{richard.chen.ryc@gmail.com}}
	\address[A]{\printead{e1}}
\end{aug}

\begin{abstract}
	This paper resolves a pivotal open problem on nonparametric inference for nonlinear functionals of volatility matrix. Multiple prominent statistical tasks can be formulated as functionals of volatility matrix, yet a unified statistical theory of general nonlinear functionals based on noisy data remains challenging and elusive. Nonetheless, this paper shows it can be achieved by combining the strengths of pre-averaging, jump truncation and nonlinearity bias correction. In light of general nonlinearity, bias correction beyond linear approximation becomes necessary. Resultant estimators are nonparametric and robust over a wide spectrum of stochastic models. Moreover, the estimators can be rate-optimal and stable central limit theorems are obtained. The proposed framework lends itself conveniently to uncertainty quantification and permits fully feasible inference. With strong theoretical guarantees, this paper provides an inferential foundation for a wealth of statistical methods for noisy high-frequency data, such as realized principal component analysis, continuous-time linear regression, realized Laplace transform, generalized method of integrated moments and specification tests, hence extends current application scopes to noisy data which is more prevalent in practice.
\end{abstract}

\begin{keyword}
	\kwd{nonlinear functionals}
	\kwd{nonlinearity bias}
	\kwd{stochastic volatility matrix}
	\kwd{microstructure noise}
	\kwd{jumps}
	\kwd{stable central limit theorems}
	\kwd{principal component analysis}.\\
	** This manuscript supersedes \hyperlink{https://arxiv.org/abs/1810.04725}{an earlier draft} named ``\textit{Inference for Volatility Functionals of Multivariate It\^o Semimartingales Observed with Jump and Noise}''. The overlap with the previous draft consists of stochastic model in Section \ref{sec:model}, rate-optimal estimator in Section \ref{sec:method}, Theorem \ref{CLT} in Section \ref{sec:asymp}. What're new in this manuscript include statistical applications of the proposed methodology, a new family of estimators with a new set of hyper-parameters that facilitate applications, additional theorem and propositions that show consistency, asymptotic Gaussianity, finite-sample positive definiteness. Section \ref{sec:taq} is new and we demonstrate a concrete financial application with large high-frequency dataset and discuss its empirical implication on portfolio construction
\end{keyword}\end{frontmatter}

\section{Introduction}\label{sec:intro}
Functions of covariance matrix are the formulations in various eminent statistical applications. The inverse function leads to precision matrix, while a more prosaic example is the identity function (covariance matrix). Their applications include, for example, graphical models and discriminant analysis. Model selection and testing are oftentimes accomplished by functions of covariance matrix, as statistical models imply linear and nonlinear relationships among its elements and blocks. As for factor models and principal component analysis, eigenvalues or associated spectral functions are differentiable nonlinear functions of the corresponding covariance matrix, so are eigenvectors.

Accordingly, functionals of stochastic volatility matrix are critical in applications using time-dependent high-frequency data. High-frequency data often arises from financial transactions, neurological measurements, satellite weather recordings, gravitational wave detectors etc. To demonstrate a concrete application, we will study principal component analysis from the perspective of nonlinear functionals in Section \ref{sec:pca}; and we will present an empirical analysis of transaction data from TAQ (Trade and Quote) database in Section \ref{sec:taq}.

Extant literature on functionals of volatility matrix demands the assumption that there is no noise in the data. This no-noise assumption implies, as the sampling frequency becomes higher, consistent estimators converge to parameters of interest. However, this assumption imposes a restricting limitation on the type of data that we can use, as the empirical fact of high-frequency data tells us an opposing story of divergence at higher frequencies. Figure \ref{fig.signature} presents how estimators behave as sampling interval $\Delta_n$ shrinks. The subscript $n$ in $\Delta_n$ denotes sample size. Larger sample sizes come with smaller sampling intervals, i.e., $\Delta_n\to0$ as $n\to\infty$. Strikingly the estimator for noiseless data does not converge for large sample sizes.

\begin{figure}[H]
	\centering
	\includegraphics[width=.9\textwidth]{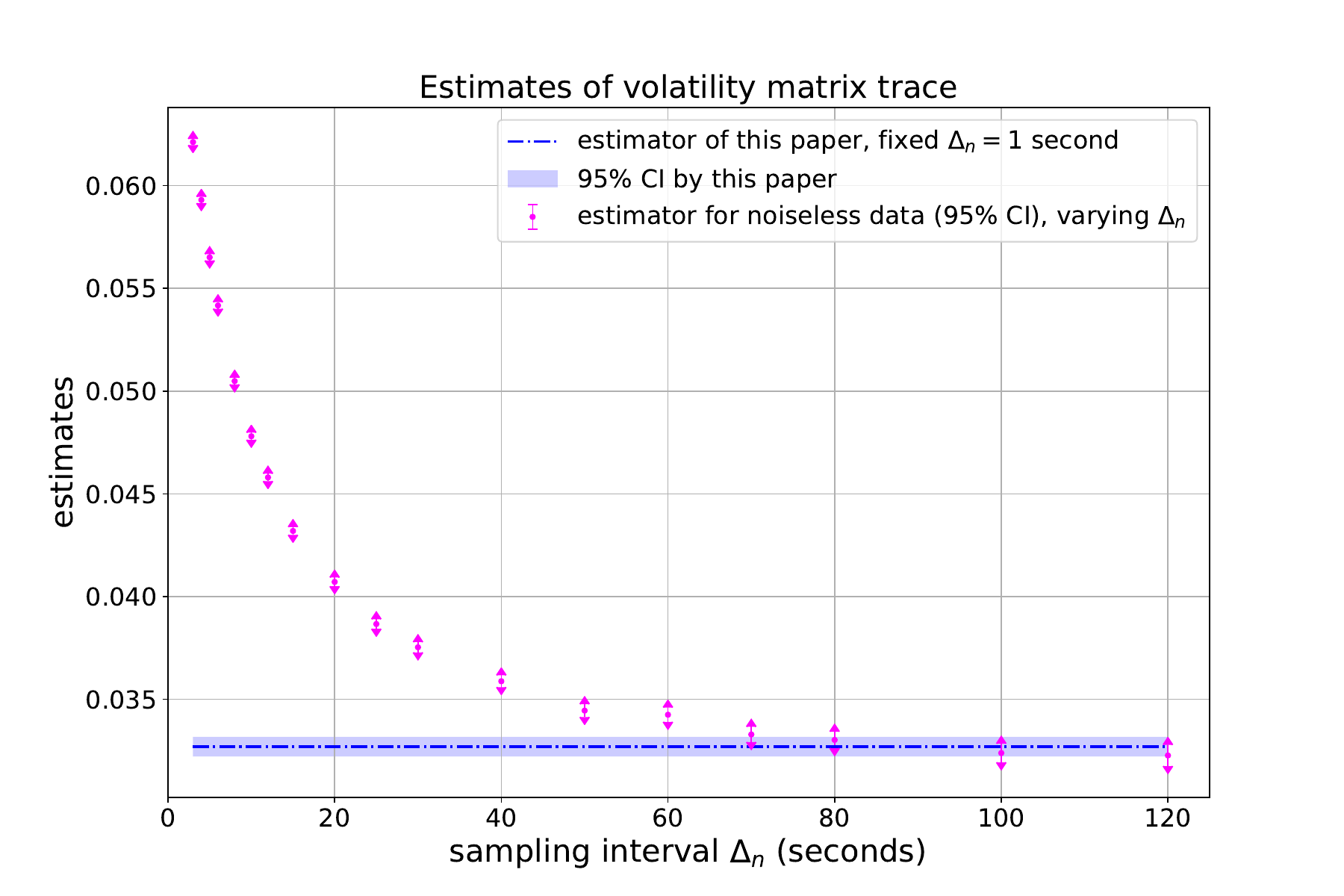}
	\caption{Estimators applied to S\&P 100 transaction data during 12/16/2019 - 12/20/2019}\label{fig.signature}
\end{figure}

Noise affects statistical significance and severely biases estimation. The presence of microstructure noise is commonly acknowledged. Solving the noise problem is vital to statistics of high-frequency data. The literature has called for a solution to this challenging open problem for years, e.g. \cite{tt12a,lx16,ltt17,ax19a,llx19,y21}.\footnote{For reference, see p. 1108 in \cite{tt12a}; p. 1617 in \cite{lx16}; p. 38 in \cite{ltt17}; Section 6 in \cite{ax19a}; p. 159 in \cite{llx19}; Remark 1 in \cite{y21}.} In this paper we provide a viable approach to this much sought-after solution. 

Formally, functionals of volatility matrix are defined as
\begin{equation}\label{def.S(g)}
	S(g)_t = \int_0^t g(c_s)\ds s,
\end{equation}
here $t$ is finite; $g:\R^{d\times d}\mapsto \R^r$ is a twice continuously differentiable function on some convex cone of positive-definite matrices; $c_s$ is stochastic volatility matrix, viz., instantaneous covariation of the continuous part of It\^o semimartingale (e.g. jump-diffusion model). (\ref{def.S(g)}) is the inferential objective and it is of paramount statistical importance. By suitable choices of $g$, the quantity $S(g)_t$ can aid researchers and practitioners in uncovering valuable yet hidden information about cross sections and temporal evolutions from datasets of interest.

Indeed, a large amount of research efforts have been devoted to this kind of questions in the last twenty years. Early literature largely concentrated in stochastic volatility matrix and linear functionals; recent research advanced on general nonlinear functionals. It may come as no surprise that general functionals of stochastic volatility matrix provide a powerful unifying framework. 

For linear functionals of stochastic volatility, the problem of microstructure noise has been resolved. Some influential approaches are two-scales realized volatility \cite{zma05,z11}, realized kernel method \cite{bhls08,bhls11}, pre-averaging method \cite{pv09b,j09,jpv10,ckp10}, quasi-maximum likelihood method \cite{x10,afx10,sx17}, spectral local method of moments \cite{bhmr14,bhmr19}.

For general nonlinear functionals, existing literature has absolved the inferential problem in the absence of microstructure noise. The pioneering work \cite{mz09} put forward a statistical framework based on contiguity and local likelihood for nonlinear transformations of volatility. The seminal paper \cite{jr13} proposed a plug-in inferential framework for nonlinear functionals, laid out the statistical theory, and spurred a long line of follow-up work on statistical methodologies and empirical applications. A more recent work \cite{c19} proposed a spectral method to circumvent asynchronous observations across dimensions when conducting statistical inference for nonlinear functionals of volatility matrices.

Detailed study of nonlinear functionals of statistical importance and resultant empirical results using noise-free data blossomed in recent years, some examples are
\begin{itemize}
	\item realized Laplace transform \cite{tt12a,tt11,tt12b};
	\item generalized method of integrated moments \cite{lx16};
		\item continuous-time linear regressions \cite{ltt17};
	\item realized principal component analysis \cite{ax19a,cmz20};
	\item estimation of asymptotic variances, e.g. the quantity called quarticity \cite{jr13};
	\item model specification tests \cite{ltt16,y20}.
\end{itemize}

These aforementioned applications utilized high-frequency transaction data to answer statistical questions in economics and finance, such as model calibration \cite{tt12a,tt11,tt12b}, testing economic models \cite{lx16,ltt16,y20}, measuring financial risks \cite{ltt17,ax19a,cmz20}, and uncertainty quantification \cite{jr13}.  Behind this powerful formulation and the wide range of applicability, however, there lurk difficult challenges.

A substantial challenge is nonlinearity. General nonlinearity induces high-order derivatives which contributes to the estimation error. Local linear approximation of nonlinear functionals only guarantees consistency and a suboptimal rate of convergence. In order to improve estimation accuracy and establish central limit theorems, bias correction of the high-order nonlinearity terms is indispensable, cf. \cite{jr13,jr15}. On one hand, direction correction of these nonlinearity biases requires computing gradients and Hessian matrices. On the other hand, nonlinearity bias correction requires careful tuning in nonparametric estimation. Creative and groundbreaking advancements have been made, under the no-noise assumption, to relax tuning restriction, facilitate bias correction and variance estimation by multiscale jackknife \cite{llx19} and matrix calculus \cite{y20}. In situations where noise must be considered, notwithstanding one more layer of complexity, we find novel nonlinearity bias terms and a new tuning range that allows successful bias correction.

The second significant challenge is jumps. Discontinuity is frequently encountered in functional data. It is also a stylized fact in financial time series. In consequence, nonparametric jump-diffusion models are widely accepted by the literature. When estimating functionals of volatility matrix, the technique to handle discontinuity is jump truncation, which has been used in \cite{llx19,y20,jr13,jr15}. However, jump and noise exhibit subtle interplay and delicate intricacy as documented in \cite{l13,mz16}. Jump truncation for the purpose of volatility matrix functionals in the noisy setting is far from being straightforward. Despite this complexity, we establish a theoretical base that takes account of noise and jump simultaneously. It is one of the contributions of this paper.

Nevertheless, the central challenge is to find an effective statistical method and theory robust to microstructure noise. Figure \ref{fig.signature} has demonstrated the curse of noise, especially in the regime where the sampling interval $\Delta_n$ is small. As made manifest in Figure \ref{fig.freq}, $\Delta_n$ being a few seconds or smaller is typical for high-frequency financial data. Yet, because of the no-noise assumption in the extant literature on nonlinear functionals, the current available statistical methods require researchers to subsample high-frequency data so that $\Delta_n$ is between 1 minute to 5 minutes, e.g. \cite{ax19a,ax19b,p20}. The goal of subsampling is to mitigate the degree of noise accumulation, so that estimators designed for noise-free data are able to yield sound results, cf. Figure \ref{fig.signature}.

\begin{figure}[H]
	\centering
	\includegraphics[width=1.\textwidth]{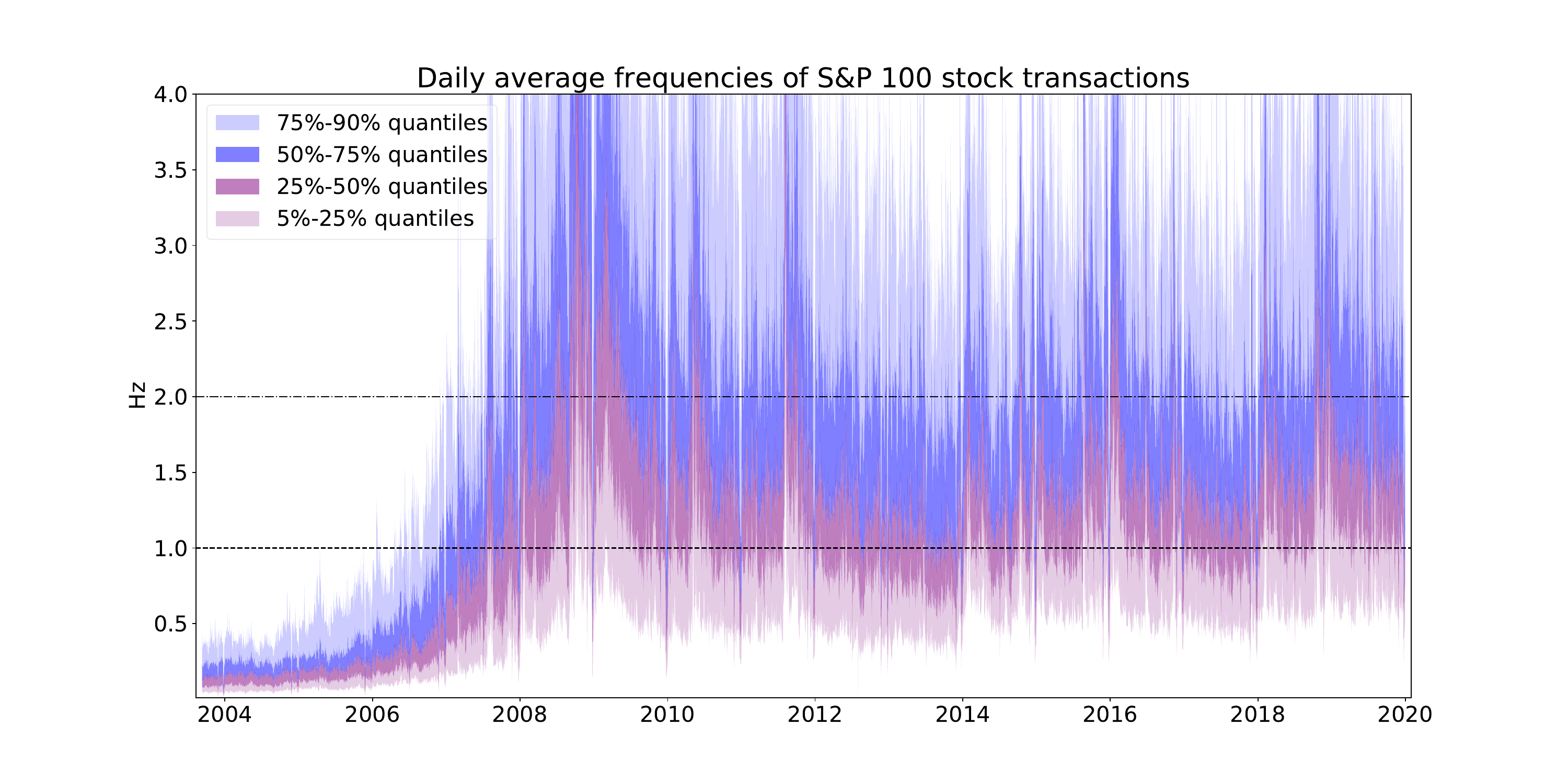}
	\caption{Daily average frequencies of S\&P 100 stock transactions}\label{fig.freq}
\end{figure}

In contrast to mitigation, we adopt a ``bottom-up'' approach, i.e., to start with a more realistic stochastic model that incorporates microstructure noise, then design noise-robust statistical methods, and then provide relevant statistical theory. We do not dispute the validity of subsampling as a mitigation, but we would like to emphasize that utilizing data of higher frequencies in a noise-robust fashion indeed gains statistical benefits, namely better estimation accuracy and sharper confidence intervals. It is an exemplar that data gives feedback to the model, and the attendant change delivers a better statistical performance.

The limitation imposed by the no-noise assumption on the amount of usable data is a pressing issue, and will be more so as we are observing explosions of data in the data-revolution age. The TAQ database of financial time series is no exception. The size of TAQ transaction data is depicted in Figure \ref{fig.size}. TAQ is timestamped at millisecond latency and is subject to noise contamination. If equipped with a noise-robust statistical methodology, we can improve $\Delta_n$ from 5 minutes to, for example, 1 second. In this case there would immediately be 300 times more data at our disposal. Consequently this paper enables applications of volatility functionals to harvest data at much higher frequencies than the state-of-the-art methods allow.

\begin{figure}[H]
	\centering
	\includegraphics[width=.75 \textwidth]{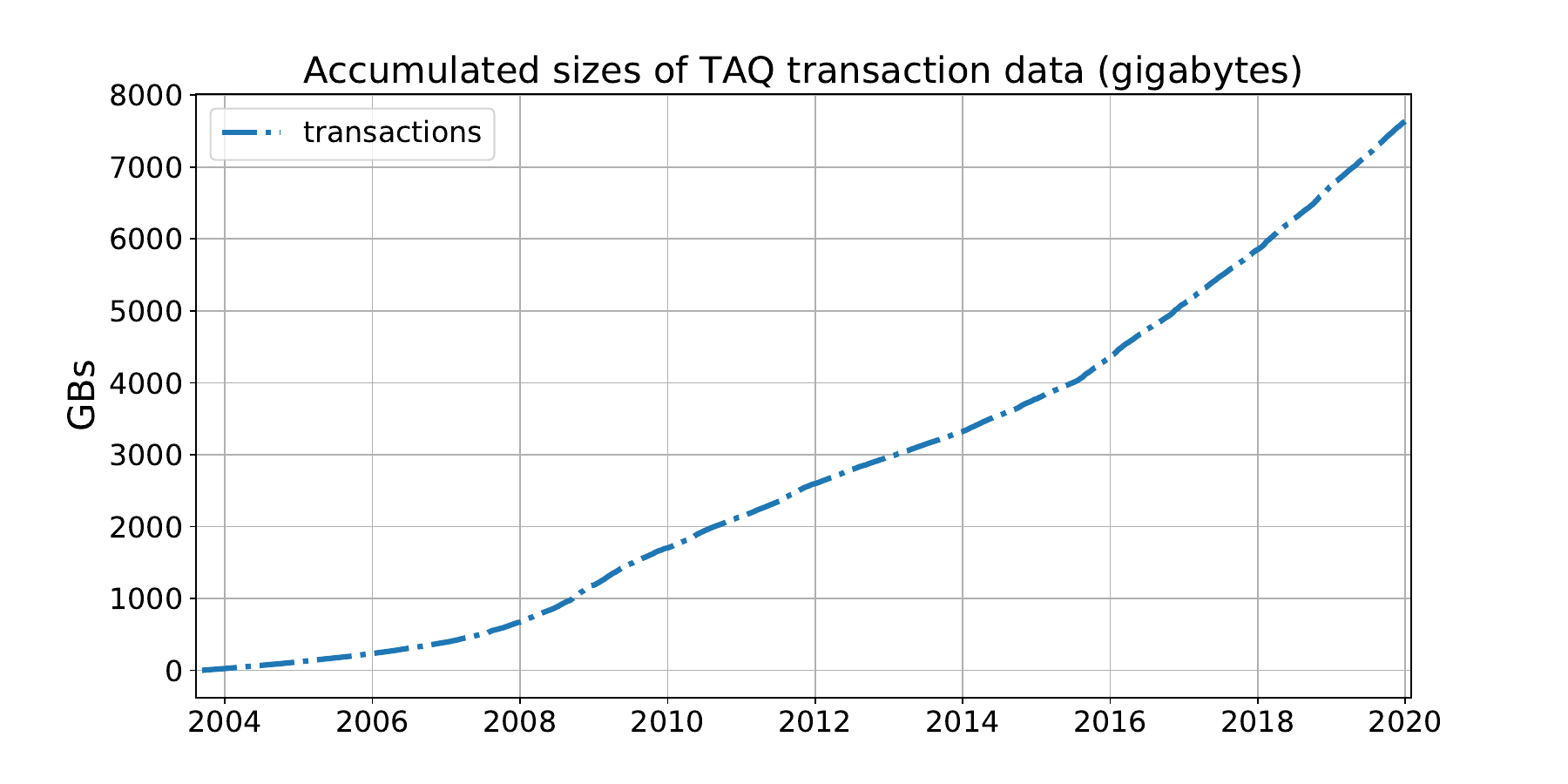}
	\caption{Accumulated sizes of TAQ transaction data}\label{fig.size}
\end{figure}

To cope with noise, we conduct local time-domain smoothing by the pre-averaging method \cite{j09,jpv10,ckp10}. To estimate general nonlinear functionals, we embed the pre-averaging method into the plug-in framework of \cite{jr13}. In this direction, the technical advancement of this work is two-fold. One is to generalize the pre-averaging method from volatility matrix to its general functionals. The other is to extend the framework of functional inference to accommodate noisy data.

Besides nonlinearity bias, jumps, and noise discussed earlier, we need to tackle the following technicalities on the journey to asymptotic theory valuable to statistics.
\begin{itemize}
	\item \textit{Stochastic volatility}. Our model is nonparametric and can reproduce the salient characteristics of empirical financial data. An primary feature is stochastic volatility which models the volatility matrix as a stochastic process. However, to develop a robust and model-free methodology with satisfactory statistical guarantees, it becomes crucial to balance statistical error (dictated by the extent of local smoothing) and target error (attributable to time-varying parameters).
	\item \textit{Dependence}. The pre-averaging method uses overlapping windows, hence the local moving averages are highly correlated for which standard central limit theorem is not valid. We use the ``big block - small block'' technique of \cite{j09}. It is instrumental to derive central limit theorems (hereafter CLT) for highly correlated high-frequency data, but in proof it induces technicality and involved notation.
	\item \textit{Stable convergence}.\footnote{Readers can read section 2.2.1 and section 2.2.2 in \cite{jp12} for an exposition of stable convergence.} For statistical purpose, we need to establish stable convergence which is stronger than weak convergence (convergence in distribution). The reason is that sampling distributions of the estimators depend on the sample path of stochastic volatility matrix. We need a mode of convergence in which the estimators converge jointly with other variables, so that we can consistently estimate the asymptotic variances for the purpose of statistical inference.
	\item \textit{Unbounded derivatives}. Some statistical applications, such as precision matrix estimation, specification tests, and linear regression, correspond to functionals with singularities in derivatives near the origin. The original framework \cite{jr13} is not applicable in this situation. \cite{ltt17} proposed a spatial localization procedure that accommodates this type of functionals under the no-noise assumption. In the interest of applications using noisy data, we extend the spatial localization by establishing uniform convergence of the pre-averaging method  plus jump truncation in the presence of discontinuity.
\end{itemize}
We state the stochastic model and our assumptions in Section \ref{sec:model}. After presenting the appropriate range of tuning parameters, we show consistency, convergence rates, and associated CLT. There are two types of general results. The first satisfies the rate-optimal CLT, hence it is asymptotically superior; the second ensures positive semi-definiteness in finite samples, and in a sense requires less computation for statistical inference. We will give details about them in Section \ref{sec:method} and Section \ref{sec:asymp}.

We focus on specific functionals which are the basis of principal component analysis (hereafter PCA) using noisy high-frequency data in Section \ref{sec:pca}. It has generated a lot of statistical interest and has been studies in, for example, \cite{ax19a,cmz20}. This paper significantly improves the accuracy of PCA either by being able to harness more data at higher frequencies, or by increasing the convergence rate. We also present recent examples among other interesting applications. Based on PCA of more than 16 years of noisy high-frequency data, a large-scale empirical analysis is presented in Section \ref{sec:taq}. Finally we conclude in Section \ref{sec:conclude}.  Proofs and simulations are in the appendices.

\section{Setting}\label{sec:model}
We describe our data generating mechanism and set up notation in this section.

The model of high-frequency data with microstructure noise consists of three components: signal process, noise process, sampling scheme.
\begin{center}
	\usetikzlibrary{positioning}
	\usetikzlibrary{shapes,snakes}
	\begin{tikzpicture}[xscale=12,yscale=6,>=stealth]
	\tikzstyle{e}=[rectangle,minimum size=5mm,draw,thick]
	\tikzstyle{v}=[ellipse,  minimum size=5mm,draw,thick]
	\node[e] (X)	[draw=red!60] {It\^o Semimartingale $X$};
	\node[e] (Y)	[draw=blue!70,right=of X] {Noisy Process $Y$};
	\node[e] (n)	[draw=black!0,right=of Y] {};
	\node[v] (D)	[draw=black!60,right=of n] {Noisy Data $Y^n_i$};
	\draw[thick,->,snake=snake] (X) to node[anchor=north]{$(Q_t)$} (Y);
	\draw[thick,->] (Y) to node[anchor=north]{sampling} (D);
	\end{tikzpicture}
\end{center}

\subsection{Signal}
The underlying signal process is modeled by an It\^o semimartingale $X$, i.e., $X$ is a solution to the stochastic differential equation
\begin{equation}\label{def.X}
	X_t = X_0 + \int_0^tb_s\ds s + \int_0^t\sigma_s\ds W_s + J_t
\end{equation}
where $b_s\in\R^d$, $\sigma_s\in\R^{d\times d'}$ with $d\le d'$, $b$ and $c$ are adapted and c\`adl\`ag\footnote{``\textit{continue \`a droite, limite \`a gauche}''. It means ``right continuous with left limits'', i.e., the sample paths lie in the Skorokhod space.}; $W$ is a $d'$-dimensional standard Brownian motion; $J$ is a purely discontinuous process.
 
The mathematical description of $J$ is relegated to (\ref{J}), which can accommodate compound Poisson process, L\'evy process etc. The stochastic volatility matrix is defined as $c_s=\sigma_s\sigma_s^\T$, which is stochastic process specified in (\ref{c}).

The full signal model assumption is given in Appendix \ref{sec:assump}, below is a summary of its features:
\begin{itemize}
	\item the drift $b$ has a smooth trajectory in the mean-square sense;
	\item the volatility $c$ is a spatially constrained It\^o semimartingale\footnote{However, it is important to accommodate long-memory volatility model. The volatility functional inference in long-memory and noisy setting is an open question under investigation.} such that $c$ is locally bounded and well-conditioned;
	\item the celebrated leverage effect (correlation between price and volatility) is allowed, and is defined through the tensor product of $\sigma$ and $\sigma^{(c)}$; 
	\item $J$ may exhibit finite and infinite activities, but has finite variation (i.e., finite-length trajectory) on compact sets;
\end{itemize}
These assumptions are necessary for the existence of CLT. The assumption that $c$ is locally well-conditioned is for the applicability of some functionals of statistical interest, such as some specification tests and linear regression.

This signal model has no parametric restriction and allows time-varying parameters, e.g. the drift $b_s$ and diffusion coefficient $\sigma_s$ could be stochastic processes. Moreover it permits unrestricted non-stationarity and dependence among the state variables.

\subsection{Noise}
We adopt the probabilistic model in \cite{j09,jpv10} to describe noise. The advantage of this formalism is that it incorporates additive white noise, rounding error, the combination thereof as special cases. The gist is to suppose for $t>0$ there is a probability transition kernel $Q_t$ linking the value of a noisy process $Y$ with that of $X$ at time $t$. The abstract formalism is given in Appendix \ref{sec:assump}.

The noise assumption is that under the conditional probability measure $Q_t$, the process $Y-X$ has zero mean, finite eighth moment, no serial correlation (though cross-sectional correlation is allowed). An example of this noisy process is given below.
\begin{eg}
\begin{equation}\label{def.Y}
	Y_t = f(X_t, e_t)
\end{equation} 
where $f:\R^d\times\R^d\mapsto\R^d$ and $e$ is a $\R^d$-valued perturbation process such that $\e_t = f(X_t,e_t) - X_t$ is centered white noise process conditional on $X$.
\end{eg}

\subsection{Sampling}
This work treats regularly sampled observations and considers in-fill asymptotics.\footnote{aka fixed-domain asymptotics, high-frequency asymptotics, small-interval asymptotics} Specifically, the samples are observed every $\Delta_n$ time units on a finite time interval $[0,t]$ where $n=\lfloor t/\Delta_n\rfloor$ is the sample size up to time $t$. As $\Delta_n\to0$, $n\to\infty$ while $t$ is fixed. We establish functional convergence and our statistical theory holds for any $t>0$. 

Throughout this paper, $U^n_i$ is written for $U_{i\Delta_n}$ where $U$ can be a process or filtration. For example, $c^n_i$ denotes the stochastic volatility matrix at time $i\Delta_n$; let be $(\F_t)$ a filtration, $\F^n_i$ is the information up to time $i\Delta_n$. For any process $U$, $\Delta^n_iU$ represents the increment $U^n_i-U^n_{i-1}$.

\subsection{Notations} For $r\in\mathbb{N}^+$, $\mathcal{C}^r(\mathcal{S})$ denotes the space of $r$-time continuously differentiable functions on the domain $\mathcal{S}$; $\mathcal{S}^+_d$ is the convex cone of $d\times d$ positive semidefinite matrices; $\I_r$ is the $r\times r$ identity matrix; $\|\cdot\|$ denotes a norm on vectors, matrices or tensors; given $a\in\R$, $\lfloor a\rfloor$ denotes the largest integer no more than $a$; $a\vee b=\max\{a,b\}$, $a\wedge b=\min\{a,b\}$; $a_n=O_p(b_n)$ means $\forall\epsilon>0$, $\exists M>0$ such that $\sup_n\mathbb{P}(a_n/b_n>M)<\epsilon$; $\mathbf{A}^\mathrm{T}$ is the transpose of the vector or matrix $\mathbf{A}$; for a multidimensional array, the entry index is written in the superscript, e.g., $X_t=(X^1_t,\cdots,X^d_t)^\T$, $\mathbf{A}^{jk}$ denotes the $(j,k)$-th entry in the matrix $\mathbf{A}$; $\otimes$ is the outer product, e.g., for two vectors $\mathbf{u}$ and $\mathbf{v}$, $\mathbf{u}\otimes\mathbf{v}=\mathbf{u}\mathbf{v}^\T$; $\partial_{jk}g$ and $\partial^2_{jk,lm}g$ denote the gradient and Hessian of $g$ with respect to the $(j,k)$-th and $(l,m)$-th entries in its matrix argument; for a scalar function $f$ define on a subset of $\R$, $f'$ denotes its gradient.
$\overset{\mathcal{L}-s(f)}{\longrightarrow}$ (resp. $\overset{\mathcal{L}-s}{\longrightarrow}$) denotes stable convergence in law of processes (resp. variables); $\overset{\mathcal{L}}{\longrightarrow}$ denotes convergence in law (weak convergence); $\overset{u.c.p.}{\longrightarrow}$ denotes uniform convergence in probability on compact sets; $\mathcal{MN}(0,V)$ is a mixed Gaussian distribution with zero mean and conditional covariance matrix $V$.

\section{Methodology}\label{sec:method}
The estimation methodology is comprised of five steps:
\begin{enumerate}
\item[i.] local moving averages of noisy data by a smoothing kernel;
\item[ii.] jump truncation operated on local moving averages;
\item[iii.] spot volatility estimator $\widehat{c}^n_i$'s for estimating $c^n_i$'s;
\item[iv.] Riemann sum of $g(\widehat{c}^n_i)$'s for approximating the integral functional $\int g(c_s)\ds s$;
\item[v.] nonlinearity bias correction.
\end{enumerate}

The idea of local moving average is due to the pre-averaging method \cite{j09,jpv10}; the truncation is modified from (16.4.4) in \cite{jp12}; the Riemann sum and bias correction are inspired by \cite{jr13}. 

Though nonlinear bias does not impact on statistical consistency, we stress that we must carefully take into account the nonlinear terms since they lead to bias in the second order. To improve convergence rate and to establish CLT, nonlinearity bias correction is crucial.

The specific recipe is given next.

\subsection{Estimator of stochastic volatility matrix}
For the local moving averages, we choose a smoothing kernel $\varphi$ such that
\begin{equation}\label{phi.cond}
\begin{array}{l}
	\text{supp}(\varphi)\subset(0,1),\, \int_0^1\varphi^2(s)\ds s>0  \\
	\varphi\in\mathcal{C} \text{ is piecewise } \mathcal{C}^1;\, \varphi' \text{ is piecewise Lipschitz}
\end{array}
\end{equation}

We choose an integer $l_n$ as the number of observations in each smoothing window, define $\varphi^n_h=\varphi(h/l_n)$ and $\psi_n=\sum_{h=1}^{l_n-1}(\varphi^n_h)^2$. We associate the following quantities with the noisy process $Y$,
\begin{equation}\label{def.Ybar.Yhat}
\begin{array}{lcl}
	\widebar{Y}^n_i &=& (\psi_n)^{-1/2}\sum_{h=1}^{l_n-1}\varphi^n_h\, \Delta^n_{i+h-1}Y \\
	\widehat{Y}^n_i &=& (2\psi_n)^{-1}\sum_{h=0}^{l_n-1}(\varphi^n_{h+1}-\varphi^n_h)^2\, \Delta^n_{i+h}Y \otimes \Delta^n_{i+h}Y
\end{array}
\end{equation}
\begin{itemize}
	\item $\widebar{Y}^n_i$ is a local moving average of the noisy data $Y^n_i$'s, and it serves as a proxy for $\Delta^n_iX$;
	\item $\widehat{Y}^n_i$ is an offset to the remaining noise in the outer product $\widebar{Y}^n_i\otimes\widebar{Y}^n_i$. 
\end{itemize}

If there is jump(s) in the local window $(i\Delta_n,(i+l_n-1)\Delta_n]$, $\widebar{Y}^n_i$ deviates from the target. We deem the moving average to be unreliable and veto it if $\|\widebar{Y}^n_i\|>\nu_n$. The threshold of jump truncation is $\nu_n=\alpha\,\Delta_n^\rho$. We let $\alpha$ mirror the dimensionality and volatility level, and set the value of $\rho$ according to the scaling property of Brownian motion (the choice of $\rho$ is stated in (\ref{tuning})). There are some variants to this jump truncation. An alternative is elementwise truncation $|\widebar{Y}^{r,n}_i|>\alpha^r_i\Delta_n^\rho$ where $\alpha^r_i$ is related to the local volatility of the $r$-th component. We recommend jump truncation tuned to optimize finite-sample performance, and it is used in our simulation and data analysis. We should emphasize that these variants share the same asymptotic behavior and satisfy the same asymptotic theory.

Based on these ingredients, we choose $k_n$ such that $k_n>l_n$, and design the following estimator of stochastic volatility matrix:
\begin{equation}\label{def.chat}
	\widehat{c}^n_i \equiv \frac{1}{(k_n-l_n)\Delta_n}\sum_{h=1}^{k_n-l_n+1}\Big(\widebar{Y}^n_{i+h}\otimes\widebar{Y}^n_{i+h}\1{\|\widebar{Y}^n_{i+h}\|\le\nu_n} - \widehat{Y}^n_{i+h}\Big)
\end{equation}

This offset term $\widehat{Y}^n_i$ is needed to achieve the optimal convergence rate in CLT. Yet a disadvantage is that $\widehat{c}^n_i$ can not be guaranteed to be positive semidefinite in finite samples. A fix is to project $\widehat{c}^n_i$ onto the convex cone $\mathcal{S}^+_d$ in Frobenius norm. Suppose $\widehat{c}^n_i = Q\Lambda Q^\T$ is the eigenvalue factorization, the positive semidefinite projection is $\widehat{c}'^n_i=Q\Lambda_+Q^\T$ where $\Lambda_+^{ij} = (\Lambda^{jj}\vee 0)\1{i=j}$.

Stimulated by \cite{ckp10}, we design another estimator of stochastic volatility matrix to fulfill the positive semidefinite requirement:
\begin{equation}\label{def.ctilde}
	\widetilde{c}^n_i \equiv \frac{1}{(k_n-l_n)\Delta_n}\sum_{h=1}^{k_n-l_n+1} \widebar{Y}^n_{i+h}\otimes\widebar{Y}^n_{i+h}\1{\|\widebar{Y}^n_{i+h}\|\le\nu_n}
\end{equation}
In addition to being positive semidefinite which some readers may prioritize, (\ref{def.ctilde}) leads to easier and more practical bias correction and uncertainty quantification for general functionals of volatility matrix (our end goal). This kind of advantage of (\ref{def.ctilde}) is pronounced in PCA. Detailed comparison between these two estimators is given in Remark \ref{remk.compare}.

The estimators of stochastic volatility matrix use local moving averages computed on overlapping windows. Local moving averages on overlapping windows fully utilize data and can lead to the optimal convergence rate. The computation can be done via convolution of the noisy data with a sequence induced by the smoothing kernel $\varphi$. Hence (\ref{def.chat}) and (\ref{def.ctilde}) can be implemented by a fast Fourier transform (FFT) algorithm and they are scalable for very large datasets.
\begin{figure}[H]
	\centering
	\includegraphics[width=.8\textwidth]{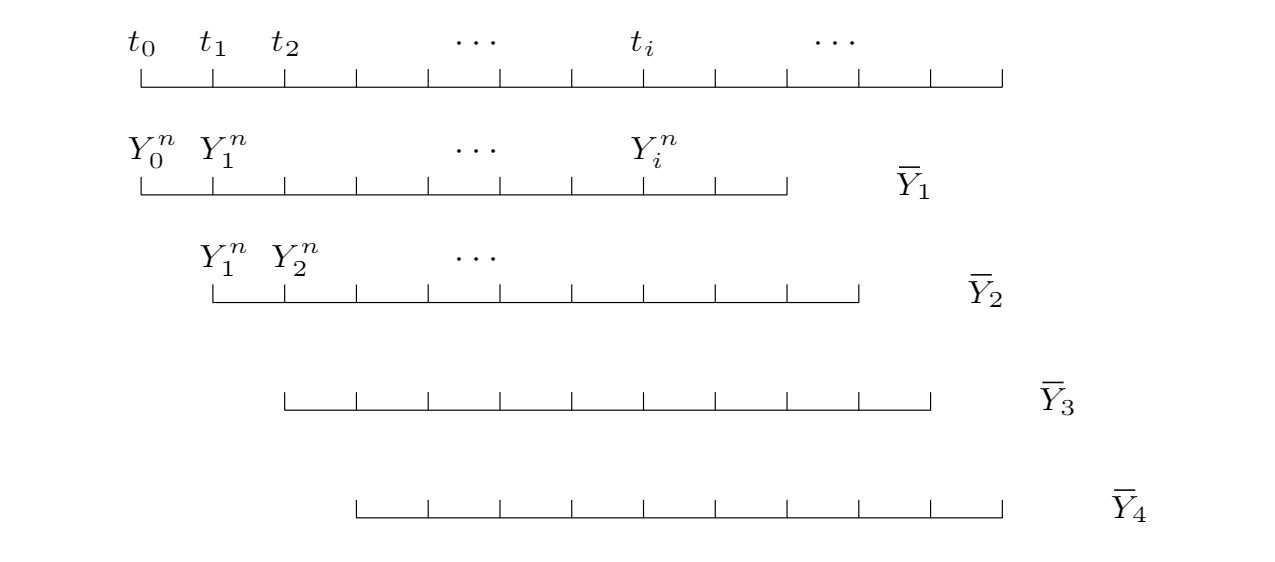}
\end{figure}

\subsection{Correction of nonlinearity bias}
As an exposition of the necessity of nonlinear bias correction, let first consider the case of $d=1$ and constant volatility. By Taylor expansion, the estimation error of the plug-in estimator $g(\widehat{c})$ can be decomposed as 
\begin{equation*}
	g(\widehat{c}) - g(c) = \underbrace{\partial g(c)(\widehat{c} - c)}_{\text{limit}} + \underbrace{\frac{1}{2}\partial^2 g(c)(\widehat{c} - c)^2}_{\text{nonlinearity bias}} + \underbrace{O_p(|\widehat{c} - c|^3)}_{\text{asymptotically negligible}}
\end{equation*}
As it turns out (shown in Appendix \ref{apdx:thm1}, \ref{apdx:thm2}), the first-order term converges to the limit in CLT. The bias arises inevitably, as long as $g$ is nonlinear (with non-zero Hessian), from the second-order derivative and the quadratic form of estimation error.

In the multivariate case, the bias term is more involved. Let $\widehat{B}(g)^n_i$ be the nonlinear bias correction term to $g(\widehat{c}^n_i)$, we find that $\widehat{B}(g)^n_i$ need to be
\begin{equation}\label{def.Bhat}
	\widehat{B}(g)^n_i = \frac{1}{2k_n\Delta_n^{1/2}}\sum^d_{j,k,l,m=1}\partial^2_{jk,lm}g(\widehat{c}^n_i)\times\Xi\big(\widehat{c}^n_i,\widehat{\rr}^n_i\big)^{jk,lm}
\end{equation}
where $\Xi$ is a tensor-valued function defined in (\ref{def.Xi}), and $\widehat{\rr}^n_i$ is an estimator of the instantaneous covariance matrix of noise at time $i\Delta_n$:
\begin{equation}\label{def.rhat}
	\widehat{\gamma}^n_i \equiv \frac{1}{2m_n}\sum_{h=1}^{m_n}\Delta^n_{i+h}Y\otimes\Delta^n_{i+h}Y
\end{equation}
with $m_n=\lfloor\theta'\Delta_n^{-1/2}\rfloor$ and $\theta'$ being positive finite.

The nonlinearity bias correction term $\widetilde{B}(g)^n_i$ to $g(\widetilde{c}^n_i)$ is
\begin{equation}\label{def.Btilde}
	\widetilde{B}(g)^n_i = \frac{1}{2k_n\Delta_n^{1/2+\delta}}\sum^d_{j,k,l,m=1}\partial^2_{jk,lm}g(\widetilde{c}^n_i)\times\Sigma\big(\widetilde{c}^n_i\big)^{jk,lm}
\end{equation}
where $\delta$ is introduced in (\ref{tuning.psd}), $\Sigma$ is defined in (\ref{tensors}).

\subsection{Estimators of functionals}
\begin{defn}\label{def.S(g)hat}
	Let $N^n_t=\lfloor t/(k_n\Delta_n)\rfloor$, a rate-optimal estimator of $(\ref{def.S(g)})$ is defined as
	\begin{equation*}
		\widehat{S}(g)^n_t \equiv k_n\Delta_n\sum_{i=0}^{N^n_t-1}\left[g(\widehat{c}^n_{ik_n}) - \widehat{B}(g)^n_{ik_n}\right] \times a^n_t 
	\end{equation*}
	$\widehat{c}^n_h$ is defined in (\ref{def.chat}), $\widehat{B}(g)^n_h$ is defined in (\ref{def.Bhat}), $a^n_t = t/(N^n_tk_n\Delta_n)$ is a finite-sample adjustment.
\end{defn}

We choose to compute $\widehat{c}^n_h$'s over disjoint windows for estimating $S(g)_t$. One may choose to compute $\widehat{c}^n_h$'s on overlapping windows instead and define an alternative estimator as $\widehat{S}(g)'^n_t\equiv\Delta_n\sum_{i=0}^{\lfloor t/\Delta_n\rfloor-1}[g(\widehat{c}^n_i) - B(g)^n_i]$. But we advice against this alternative for two reasons. The first one is that $\widehat{S}(g)^n_t$ and $\widehat{S}(g)'^n_t$ has the same asymptotic properties such as convergence rate and limiting distributions, cf. \cite{ax19a,jr13}. The second one, more essential for some applications, is that computing $\widehat{S}(g)'^n_t$ is unfortunately impractical in some applications, such as the data analysis using PCA in Section \ref{sec:taq}. In fact, should $\widehat{S}(g)'^n_t$ be employed, we need to compute more than 100,000 eigenvalue factorizations of $90\times90$ matrices for data of every week, with more than 850 weeks our data analysis would take too long to finish.

\begin{defn}\label{def.S(g)tilde}
	Let $N^n_t=\lfloor t/(k_n\Delta_n)\rfloor$, an estimator of $(\ref{def.S(g)})$ with positive semidefinite plug-ins is defined as
	\begin{equation*}
	\widetilde{S}(g)^n_t \equiv k_n\Delta_n\sum_{i=0}^{N^n_t-1}\left[g(\widetilde{c}^n_{ik_n}) - \widetilde{B}(g)^n_{ik_n}\right] \times a^n_t
	\end{equation*}
	where $\widetilde{c}^n_h$ is defined in (\ref{def.ctilde}), $\widetilde{B}(g)^n_h$ is defined in (\ref{def.Btilde}), and $a^n_t = t/(N^n_tk_n\Delta_n)$ is an adjustment for the right-edge effect.
\end{defn}

Our estimators, given the smoothing kernel $\varphi$, entail three tuning parameters.
\begin{table}[H]
	\centering
	\caption{Tuning parameters}
	\begin{tabular}{c|l}
			& description \\
		\hline
		$l_n$   & length of overlapping windows for local moving averages \\
		$k_n$   & length of disjoint windows for stochastic volatility matrix estimation \\
		$\nu_n$ & threshold of jump truncation
	\end{tabular}
\end{table}

A judicially tuning choice is an indispensable if not the most important part of sound estimation and inference. As for $k_n$, there are three levels of considerations:
\begin{enumerate}
	\item[level 1:] $k_n$ should be chosen to result in consistency and a satisfactory convergence rate;
	\item[level 2:] $k_n$ should be chosen so that we can establish CLT for inference;
	\item[level 3:] $k_n$ should be chosen in a way that allows convenient correction of nonlinear bias.
\end{enumerate}
A severe pitfall looms when level 3 is dismissed. If falling prey to this pitfall, researchers have to estimate the volatility of volatility and volatility jumps (both are notoriously hard to estimate) in order to establish CLT \cite{jr15}. Under the no-noise assumption, this pitfall is documented and avoided by \cite{jr13}. Recently \cite{llx19} proposed the multiscale jackknife method and \cite{y20} used matrix calculus for practical nonlinear bias correction, and they investigated how to choose $k_n$ more robustly under the no-noise assumption. Their results greatly facilitate statistical applications using sparsely subsampled data. 

When we must confront the noise problem to improve statistical quality, there is one more tuning parameter, namely $l_n$, to consider. However, when noise is taken into account, the three levels of considerations for $k_n$ are missing so far, and given a legitimate $l_n$ how to adjust $v_n$ is unknown.

We provide an answer to this tuning problem. We finding the range of tuning parameters when noise and jump are both present. The benefits include more usable data and better estimation accuracy. Importantly, practical bias correction and uncertainty quantification are possible under our tuning configuration. We will show, with suitable choices of the combination $(l_n,\,k_n,\,\nu_n)$, our estimators are widely applicable to any $g:\mathcal{S}^+_d\mapsto\R^r$ that satisfies
\begin{equation}\label{g.cond}
	g\in\mathcal{C}^3(\s)
\end{equation}
where $\s \subset \s^+_d$ is the subspace containing an enlargement of the sample path of $c$. Precisely, $\s\supset\cup_m\s^\epsilon_m$ for some $\epsilon>0$, $\s^\epsilon_m=\big\{A\in\s^+_d: \inf_{M\in\s_m}\|A-M\|\le\epsilon\big\}$ is an enlargement of $\s_m$ which is identified in Assumption \ref{A-v}.

\subsection{Tuning parameters for the optimal convergence rate}
A proper combination of the tuning parameters is the requisite for consistency, the existence of CLT, and the optimal convergence rate. To fulfill these objectives for $\widehat{S}(g)^n_t$, we need
\begin{equation}\label{tuning}
\left\{\begin{array}{rcll}
	l_n &=& \lfloor \theta\Delta_n^{-1/2} \rfloor & \\ 
	k_n &=& \lfloor \varrho\Delta_n^{-\kappa} \rfloor & \hspace{5mm} \kappa\in\left(\frac{2}{3}\vee\frac{2+\nu}{4},\frac{3}{4}\right)\\
	\nu_n &=& \alpha\Delta_n^{\rho} & \hspace{5mm} \rho\in \big[\frac{1}{4}+\frac{1-\kappa}{2-\nu},\frac{1}{2}\big)
\end{array}\right.
\end{equation}
where $\theta,\varrho,\alpha>0$ are positive finite, and $\nu\in[0,1)$ is introduced in Assumption \ref{A-v} which dictates the jump intensity. 

Here we offer an intuition for (\ref{tuning}). The reader can skip this part without affecting understanding of the main results in Section \ref{sec:asymp}.
\begin{itemize}
\item \textit{$l_n$ influences the convergence rate}\\
	Under the noise model (\ref{def.Y}), according to the definition (\ref{def.Ybar.Yhat}),
	\[\widebar{Y}^n_i = \widebar{X}^n_i + \widebar{\e}^n_i\]
	where $\widebar{X}^n_i$ and $\widebar{\e}^n_i$ are defined analogously to $\widebar{Y}^n_i$ (see (\ref{def.Ybar.Yhat})). We can write $\widebar{\e}^n_i = -\psi_n^{-1/2}\sum_{h=0}^{l_n-1}(\varphi^n_{h+1}-\varphi^n_h)\,\e^n_{i+h}$.
	By the conditional independence of $\e^n_i$'s, $\widebar{\e}^n_i=O_p(l_n^{-1})$; $\widebar{X}^n_i=O_p(\Delta_n^{-1/2})$ by (\ref{est.Cbar}) in Appendix \ref{apdx:contin}. By taking $l_n\asymp\Delta_n^{-1/2}$ the orders of $\widebar{X}^n_i$ and $\widebar{\e}^n_i$ are equal. It turns out that this choice of local smoothing window will deliver the optimal rate of convergence.
\item \textit{$k_n$ dictates bias-correction and the CLT form}\\
	Here for exposition let's suppose $d=1$ and $X$ is an continuous It\^o semimartingale, then 
	\[\widehat{c}^n_i - c^n_i = d^n_i + s^n_i\]
	where
	\begin{eqnarray*}
		d^n_i &=& \frac{1}{(k_n-l_n)\Delta_n}\int_{i\Delta_n}^{(i+k_n-l_n+1)\Delta_n}(c_s-c^n_i)\ds s \\
		s^n_i &=& \frac{1}{(k_n-l_n)\Delta_n} \bigg[\widehat{c}^n_i(k_n-l_n)\Delta_n - \int_{i\Delta_n}^{(i+k_n-l_n+1)\Delta_n}c_s\ds s\bigg]
	\end{eqnarray*}
	\begin{itemize}
	\item We call $d^n_i$ ``target error''; $d^n_i=O_p((k_n\Delta_n)^{1/2})$ by (\ref{classic}) in Appendix \ref{apdx:prelim}; its value depends on the temporal evolution of the underlying stochastic volatility.
	\item We call $s^n_i$ ``statistical error''; $s^n_i \approx \frac{1}{(k_n-l_n)\Delta_n}\Delta_n^{1/4}(\chi^n_{i+k_n-l_n+1}-\chi^n_i)$,  where $\chi$ is another continuous It\^o semimartingale; this result is due to (3.8) in \cite{j09}; so $s^n_i=O_p((k_n\Delta_n^{1/2})^{-1/2})$.
	\end{itemize}
	
	Balancing the orders of $d^n_i$ and $s^n_i$ by setting $\kappa=3/4$ brings the error $\widehat{c}^n_i - c^n_i$ down to the minimum order. However, for the purpose of general nonlinear functionals, choosing $\kappa\ge3/4$ causes the bias to be substantially influenced by $d^n_i$ which depends on volatility of volatility and volatility jump (difficult to estimate and de-bias in applications). Therefore, it is advisable to choose $\kappa<3/4$, in which case $s^n_i$ dominates and $d^n_i$ vanishes asymptotically, thereby the thorny terms are circumvented. Besides, to achieve successful de-biasing of statistical error and negligibility of higher-order Taylor-expansion terms, we need $\kappa>2/3$. 
	
	\begin{figure}
	\begin{tikzpicture}[scale=1.3]
	\draw[->] (0,0) -- (8.0,0) node[anchor=west] {$\kappa$ for $k_n$};
	\draw	(2.0,0)  node[anchor=north] {$1/2$}
	(3.66,0) node[anchor=north] {$2/3$}
	(4.5,0)  node[anchor=north] {$3/4$}
	(7.0,0)  node[anchor=north] {$1$};
	\draw[<->] (3.66,-0.5) -- (4.5,-0.5);
	\fill[red!10]  (7.0,3.0) -- (4.5,1.5) -- (4.5,0) -- (7.0,0) -- cycle;
	\fill[blue!10] (2.0,3.0) -- (4.5,1.5) -- (4.5,0) -- (2.0,0) -- cycle;
	\draw[->] (1,0) -- (1,3.6) 
	node[anchor=south,align=center] {order of $\log(|\widehat{c}^n_i-c^n_i|)$};
	\draw[thick,dotted] (2.0,0)  -- (2.0,3.5);
	\draw[thick,dashed] (3.66,0) -- (3.66,3.5);
	\draw[thick,dashed] (4.5,0)  -- (4.5,3.5);
	\draw[thick,dotted] (7.0,0)  -- (7.0,3.5);
	\draw (7.5,3) node {\textcolor{red}{$d^n_i$}}; 
	\draw[thick,red] (2,0) -- (7,3);
	\draw[thick,<-] (6.5,1.5) -- (7.2,1.3)
	node[right,align=left] {dominated by\\ target error};
	\draw (1.5,3) node {\textcolor{blue}{$s^n_i$}}; 
	\draw[thick,blue] (2,3) -- (7,0);
	\draw[thick,<-] (2.5,1.5) -- (1.8,1.3)
	node[left,align=right] {dominated by\\ statistical error};
	\draw[thick,<-] (4.55,1.6) -- (4.8,2.2)
	node[above,align=center] {error\\ minimizing $\kappa$};
	\end{tikzpicture}
	\caption{A schematic plot of tuning $k_n$} 
	\label{fig.intuition}
	\end{figure}
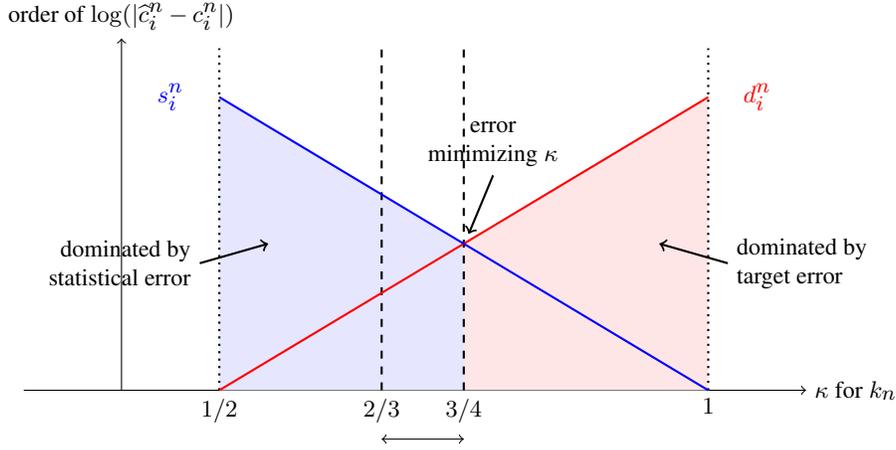
	
	When jump is present and its activity index $\nu>2/3$, the lower-bound requirement on $\kappa$ increase from $2/3$ to $(2+\nu)/4$. This increase is required for uniform convergence in the spatial localization procedure, which in turn is necessary for functionals with a singularity at the origin.
	
	Section 3.1, 3.2 of \cite{jr13} give an analogous discussion and their conclusion is $1/3<\kappa<1/2$ under the no-noise assumption. With noise and local moving averages, our conclusion is markedly different. 
\item \textit{$\nu_n$ disentangles volatility from jump variation}\\
	According to (\ref{est.Y.bars.hats}) in Appendix \ref{apdx:jmp}, $\|\widebar{Y}^n_i\|=O_p(\Delta_n^{1/2})$ if there is no jump over $[i\Delta_n,(i+l_n)\Delta_n]$. By choosing $\rho<1/2$, the truncation level $\nu_n>\Delta_n^{1/2}$, hence jump truncation keeps the diffusion movements largely intact and discards jumps in a certain sense. To effectively filter out the jumps, the truncation level should be bounded from above and the upper bounds depends on the jump activity index $\nu$.
\end{itemize}

\begin{remk}
If the reader is interested to estimate $(\ref{def.S(g)})$ with $g$ satisfying
\begin{equation}\label{g.cond1}
	\|\partial_h g(x)\| \le K_h(1+\|x\|^{r-h}),\,\, h=0,1,2,3,\, r\ge3
\end{equation}
the requirements on $k_n$ and $\nu_n$ can be loosened and become
\begin{equation}\label{tuning1}
\begin{array}{rcll}
	k_n &=& \lfloor\varrho\Delta_n^{-\kappa}\rfloor &  \hspace{5mm} \kappa\in\left(\frac{2}{3},\frac{3}{4}\right)\\
	\nu_n &=&   \alpha\Delta_n^{\rho} & \hspace{5mm} \rho\in \big[\frac{1}{4}+\frac{1}{4(2-\nu)},\frac{1}{2}\big)
\end{array}
\end{equation}
For generality and wider applicability, we choose to accommodate the functional space (\ref{g.cond}) and retain the requirement (\ref{tuning}).
\end{remk}

\subsection{Tuning parameters for p.s.d. plug-ins}
If readers choose to use $\widetilde{S}(g)^n_i$, the tuning parameters should satisfy
\begin{equation}\label{tuning.psd}
\left\{\begin{array}{rcll}
	l_n &=& \lfloor \theta\Delta_n^{-1/2-\delta} \rfloor & \hspace{5mm} \delta\in\big(\frac{1}{10},\frac{1}{2}\big) \\
	k_n &=& \lfloor \varrho\Delta_n^{-\kappa} \rfloor & \hspace{5mm} \kappa\in\big(\big(\frac{2}{3}+\frac{2\delta}{3}\big)\vee\big(\frac{2+\nu}{4}+\frac{(2-\nu)\delta}{2}\big),\frac{3}{4}+\frac{\delta}{2}\big)\\
	\nu_n &=& \alpha\Delta_n^{\rho} & \hspace{5mm} \rho\in \big[\frac{1}{4}+\frac{\delta}{2}+\frac{1-\kappa}{2-\nu},\frac{1}{2}\big)
\end{array}\right.
\end{equation}
The key insight is that we can increase $l_n$ to attenuate noise in $\widebar{Y}^n_i$ and dispense with the noise offset term $\widehat{Y}^n_i$. However, the positive semidefinite estimator does not possess the optimal convergence rate (see the next section).


\section{Asymptotics and Inference}\label{sec:asymp}
\subsection{Elements}
Before stating the asymptotic result, we need to introduce some elements in our asymptotic theory. We associate the following quantities with the smoothing kernel $\varphi$ for $l,m=0,1$:
\begin{equation}\label{phi.vars}
\begin{array}{ll}
	\phi_0(s)=\int_s^1\varphi(u)\varphi(u-s)\ds u, \hspace{5mm} & \phi_1(s)=\int_s^1\varphi'(u)\varphi'(u-s)\ds u \\
	\Phi_{lm}=\int_0^1\phi_l(s)\phi_m(s)\ds s, & \Psi_{lm}=\int_0^1s\,\phi_l(s)\phi_m(s)\ds s
\end{array}
\end{equation}

We define $\Sigma$, $\Theta$, $\Upsilon$ as $\R^{d\times d\times d\times d}$-valued functions, such that for $x,z\in\R^{d\times d}$, $j,k,l,m=1,\cdots,d$,
\begin{eqnarray}\label{tensors}
	\Sigma(x)^{jk,lm}   &=& \frac{2\theta\,\Phi_{00}}{\phi_0(0)^2}\, \big(x^{jl}x^{km} + x^{jm}x^{kl}\big) \nonumber\\
	\Theta(x,z)^{jk,lm} &=& \frac{2\Phi_{01}}{\theta\,\phi_0(0)^2}\, \big(x^{jl}z^{km} + x^{jm}z^{kl} + x^{km}z^{jl} + x^{kl}z^{jm}\big) \\
	\Upsilon(z)^{jk,lm} &=& \frac{2\Phi_{11}}{\theta^3\,\phi_0(0)^2}\, \big(z^{jl}z^{km} + z^{jm}z^{kl}\big) \nonumber
\end{eqnarray}
where $\theta$ is introduced in (\ref{tuning}). We define $\Xi$ as a tensor-valued function
\begin{equation}\label{def.Xi}
	\Xi(x,z) = \Sigma(x) + \Theta(x,z) + \Upsilon(z)
\end{equation}

Now we are ready to describe limiting processes.
\begin{defn}\label{def.Z(g)}
Suppose $g$ satisfies (\ref{g.cond}) or (\ref{g.cond1}), $Z_\Xi(g)$ and $Z_\Sigma(g)$ are two stochastic processes defined on an extension of the probability space  $(\Omega,\F,(\F_t),\mathbb{P})$ specified in (\ref{prob.sp}),
such that conditioning on $\mathcal{F}$, $Z_\Xi(g)$ and $Z_\Sigma(g)$ are centered Gaussian processes with conditional covariance matrices being
\begin{eqnarray*}
	\widebar{E}[Z_\Xi(g)Z_\Xi(g)^\T|\F] &=& V_\Xi(g) \\
	\widebar{E}[Z_\Sigma(g)Z_\Sigma(g)^\T|\F] &=& V_\Sigma(g)
\end{eqnarray*}
where $\widebar{E}$ is the conditional expectation operator on the extended probability space and
\begin{eqnarray}
	V_\Xi(g)_t &=& \int_0^t \sum_{j,k,l,m=1}^{d}\partial_{jk}g(c_s)\,\partial_{lm}g(c_s)^\T\,\Xi(c_s,\gamma_s)^{jk,lm}\ds s \label{AVAR} \\
	V_\Sigma(g)_t &=& \int_0^t \sum_{j,k,l,m=1}^{d}\partial_{jk}g(c_s)\,\partial_{lm}g(c_s)^\T\,\Sigma(c_s)^{jk,lm}\ds s \label{AVAR.psd} 
\end{eqnarray}
where the process $\rr$ is the conditional instantaneous covariance matrix of noise defined in (\ref{def.r}).
\end{defn}

\subsection{Central limit theorems}
\begin{thm}\label{CLT}
Assume assumptions \ref{A-v}, \ref{A-r}, and suppose $g$ satisfies (\ref{g.cond}). We select the tuning parameters according to (\ref{tuning}), then we have the following stable convergence in law of discretized process to a conditional continuous It\^o semimartingale on compact subsets of $\R^+$:
\begin{equation}\label{clt1}
	\Delta_n^{-1/4}\Big[\widehat{S}(g)^n-S(g)\Big] \overset{\mathcal{L}-s(f)}{\longrightarrow} Z_\Xi(g)
\end{equation}
where $S(g)$ is defined in (\ref{def.S(g)}), $\widehat{S}(g)^n$ is from Definition \ref{def.S(g)hat}, $Z(g)$ is identified in Definition \ref{def.Z(g)}.
\end{thm}

Theorem \ref{CLT} is valid over the functional space (\ref{g.cond}), which is as general as the state-of-the-art literature can achieve. If applications require functionals whose derivatives satisfy a stronger condition (\ref{g.cond1}) of polynomial growth, we can put less restrictions on the tuning parameters.
\begin{prop}\label{CLT.restricted}
Assume Assumptions \ref{A-v}, \ref{A-r}. Replace the functional space (\ref{g.cond}) with (\ref{g.cond1}), replace the tuning conditions (\ref{tuning}) on $k_n$, $\nu_n$ with (\ref{tuning1}), then the functional stable convergence (\ref{clt1}) still holds true.
\end{prop}
However, Proposition \ref{CLT.restricted} rules out functionals that involve matrix inversion, it is not applicable to, for instance, estimation of precision matrices and inference of linear regression models. In the rest of this paper, we focus on the results over the general functional space (\ref{g.cond}).

We also establish a limit theorem for $\widetilde{S}(g)^n$ using positive semidefinite plug-ins:
\begin{thm}\label{CLT.psd}
	Assume assumptions \ref{A-v}, \ref{A-r}, and suppose $g$ satisfies (\ref{g.cond}). We select the tuning parameters according to (\ref{tuning.psd}), then we have the following stable convergence in law of discretized process to a conditional continuous It\^o semimartingale on compact subsets of $\R^+$:
	\begin{equation}\label{clt1.psd}
		\Delta_n^{-(1/4-\delta/2)}\Big[\widetilde{S}(g)^n-S(g)\Big] \overset{\mathcal{L}-s(f)}{\longrightarrow} Z_\Sigma(g)
	\end{equation}
	where $\delta$ is specified in (\ref{tuning.psd}), $Z_\Sigma(g)$ is identified in Definition \ref{def.Z(g)}.
\end{thm}

The foregoing asymptotic results are stated in terms of converging stochastic processes, which is necessary to express functional stable convergence in law (the strongest mode of convergence one can obtain in our setting). Below is a point-wise formulation which is more relevant for statistical applications. For finite $t$, and under the conditions of Theorem \ref{CLT},
\begin{equation}\label{clt2}
	n^{1/4}\left[\widehat{S}(g)^n_t-S(g)_t\right]\overset{\mathcal{L}-s}{\longrightarrow}\mathcal{MN}\big(0,\sqrt{t}\,V_\Xi(g)_t\big)
\end{equation}
Moreover, under the conditions of Theorem \ref{CLT.psd},
\begin{equation}\label{clt2.psd}
	n^{1/4-\delta/2}\left[\widetilde{S}(g)^n_t-S(g)_t\right]\overset{\mathcal{L}-s}{\longrightarrow}\mathcal{MN}\big(0,t^{1/2-\delta}\,V_\Sigma(g)_t\big)
\end{equation}

\subsection{Confidence intervals}
The asymptotic variances (\ref{AVAR}) and (\ref{AVAR.psd}) are also functionals of volatility matrix, so we can also estimate the asymptotic variances by the plug-in framework, i.e. plugging in estimates of stochastic volatility matrices (and estimates of instantaneous noise covariance matrices) to the corresponding functional forms.
\begin{defn}
	For $g$ that satisfies (\ref{g.cond}), we define
\begin{equation}
	\widehat{V}_\Xi(g)^n_t \equiv k_n\Delta_n\sum_{i=0}^{N^n_t-1}\sum_{j,k,l,m=1}^{d}\partial_{jk}g(\widehat{c}^n_{ik_n})\,\partial_{lm}g(\widehat{c}^n_{ik_n})^\T\;\Xi(\widehat{c}^n_{ik_n},\widehat{\gamma}^n_{ik_n})^{jk,lm} \label{def.V(g)hat}
\end{equation}
as the estimator for $V_\Xi(g)_t$ for which we abide the tuning (\ref{tuning}).

Furthermore, we define
\begin{equation}
	\widetilde{V}_\Sigma(g)^n_t \equiv k_n\Delta_n\sum_{i=0}^{N^n_t-1}\sum_{j,k,l,m=1}^{d}\partial_{jk}g(\widetilde{c}^n_{ik_n})\,\partial_{lm}g(\widetilde{c}^n_{ik_n})^\T\;\Sigma(\widetilde{c}^n_{ik_n})^{jk,lm} \label{def.V(g)tilde}
\end{equation}
as the estimator for $V_\Sigma(g)_t$ for which we abide the tuning (\ref{tuning.psd}).
\end{defn}

We have that both $\widehat{V}_\Xi(g)^n$ and $\widetilde{V}_\Sigma(g)^n$ are consistent:
\begin{prop}\label{AVAR.est}
Suppose Assumptions \ref{A-v}, \ref{A-r}, $g$ satisfies (\ref{g.cond}) and $t$ is finite. 

With the tuning (\ref{tuning}), we have
\begin{equation*}
	\big\|\widehat{V}_\Xi(g)^n_t - V_\Xi(g)_t\big\| = O_p(\Delta_n^{\kappa-1/2})
\end{equation*}
With the tuning (\ref{tuning.psd}), we have
\begin{equation*}
\big\|\widetilde{V}_\Sigma(g)^n_t - V_\Sigma(g)_t\big\| = O_p(\Delta_n^{\kappa-1/2-\delta})
\end{equation*}
\end{prop}
\begin{proof}
The asymptotic variance $V_\Xi(g)_t$ is a smooth functional of volatility matrix and instantaneous noise covariance matrix. Its consistency follows from the consistence of the  volatility matrix estimator (\ref{def.chat}) and the instantaneous noise covariance matrix estimator (\ref{def.rhat}). According to Lemma \ref{est.beta} in Appendix \ref{apdx:contin} and (\ref{xi.cond}) in Appendix \ref{apdx:asymneg}, the error rate of $\widehat{V}_\Xi(g)^n_t$ is determined by the estimation error of (\ref{def.chat}). Therefore, the error rate of $\widehat{V}_\Xi(g)^n_t$ is the same as the error rate of the volatility functional estimator without bias correction, which is $(k_n\Delta_n^{1/2})^{-1}$, then the proposition follows from (\ref{tuning}).

The consistency of $\widetilde{V}_\Sigma(g)^n_t$ is established by a similar argument.
\end{proof}

Based on Theorem \ref{CLT}, Theorem \ref{CLT.psd}, Proposition \ref{AVAR.est} and the property of stable convergence, we have the following feasible central limit theorems:
\begin{corol}
Suppose $g$ satisfies (\ref{g.cond}) and Assumptions \ref{A-v}, \ref{A-r} hold. With the tuning (\ref{tuning}), we have
\begin{equation}\label{clt3}
	\big[\Delta_n^{1/2}\,\widehat{V}_\Xi(g)^n_t\big]^{-1/2}\left[\widehat{S}(g)^n_t-S(g)_t\right]\overset{\mathcal{L}}{\longrightarrow}\mathcal{N}\big(0,\I_r\big)
\end{equation}
in restriction to the event $\{\omega\in\Omega, \widehat{V}_\Xi(g)^n_t \text{ is positive definite}\}$, where $\Omega$ is defined in (\ref{prob.sp}).

With the tuning (\ref{tuning.psd}), we have
\begin{equation}\label{clt3.psd}
	\big[\Delta_n^{1/2-\delta}\,\widetilde{V}_\Sigma(g)^n_t\big]^{-1/2}\left[\widetilde{S}(g)^n_t-S(g)_t\right]\overset{\mathcal{L}}{\longrightarrow}\mathcal{N}\big(0,\I_r\big)
\end{equation}
in restriction to the event $\{\omega\in\Omega, \widetilde{V}_\Sigma(g)^n_t \text{ is positive definite}\}$.
\end{corol}

\section{Applications}\label{sec:pca}
\subsection{PCA of high-frequency data}
PCA is one of powerful statistical methods that help us obtain information and gain insights from large and complex datasets. It reduces dimensionality, facilitates data visualization, reveals common trends et cetera. For high-frequency data, PCA can be done by inference of spectral structure of stochastic volatility matrix
\begin{equation}\label{eigen}
	c_t = \left[q_t^1,\cdots,q_t^d \right]
	\left[\begin{array}{ccc}
		\lambda_t^1 & & \\
		& \ddots & \\
		& & \lambda_t^d
	\end{array}\right]
\left[\begin{array}{c}
	q_t^{1,\T} \\ \vdots \\ q_t^{d,\T}
\end{array}\right]
\end{equation}

The eigenvalues and eigenvectors can be considered as functionals of the associated stochastic volatility matrix once directions and lengths of the eigenvectors are specified and some regularity conditions hold. As per this functional idea we write $\lambda^k(c_t)=\lambda^k_t,\; q^k(c_t)=q^k_t$. Based on this insight, \cite{ax19a} applied PCA to nonstationary financial data by conducting inference on the realized eigenvalue $\int_0^t\lambda_s^k\ds s$, realized eigenvector $\int_0^t q_s^k\ds s$, realized principal component $\int_0^tq_{s-}^k\ds X_s$. More recently, \cite{cmz20} extends the realized PCA to asynchronously observed high-frequency noisy data. Here we extend \cite{ax19a} to noisy data, and improve the convergence rate of \cite{cmz20} by a considerable margin.

Suppose for $t\in[0,T]$,
\begin{equation}\label{eigen.clusters}
	\lambda_t^{r_0+1} =\cdots= \lambda_t^{r_1} > \lambda_t^{r_1+1} =\cdots= \lambda_t^{r_2} > \cdots\cdots > \lambda_t^{r_{K-1}+1} =\cdots= \lambda_t^{r_K}
\end{equation}
where $r_0=0,\, r_K=d$. Let $\mathcal{K}_k=\{r_{k-1}+1,\cdots,r_k\}$ for $k=1,\cdots,K$. Define the following functions of volatility matrix
	\[ \lambda(c_t)=\Big( \frac{1}{r_1-r_0}\sum_{r\in\mathcal{K}_1}\lambda_t^r, \cdots, \frac{1}{r_K-r_{K-1}}\sum_{r\in\mathcal{K}_K}\lambda_t^r \Big)^\T \]
If $r_k=r_{k-1}+1$, $\lambda^k_t$ is called a simple eigenvalue, otherwise a repeated eigenvalue. If $\lambda^k_t$ is a repeated eigenvalue, it is not differentiable with respect to $c_t$, and establishing its asymptotic distribution would be problematic. Nevertheless, $\lambda$ defined above is a differentiable function regardless of singleton or multicity (see Remark \ref{remk.spectral.differentiability} in Appendix \ref{apdx:pca}), so we can do nonlinear bias correction and derive its CLT.

To do PCA, we need to estimate
\begin{eqnarray*}
	S(\lambda)_t &=& \int_0^t \lambda(c_s)\ds s \\
	S(q^k)_t &=& \int_0^t q^k(c_s)\ds s,\hspace{4mm} k=1\cdots,d
\end{eqnarray*}

Given $\widetilde{c}^n_i$ defined in (\ref{def.ctilde}), we compute its spectral factorization:
\begin{equation*}
	\widetilde{c}^n_i = \left[\widetilde{q}^{1,n}_i,\cdots,\widetilde{q}^{d,n}_i\right]
	\left[\begin{array}{ccc}
		\widetilde{\lambda}^{1,n}_i & & \\
		& \ddots & \\
		& & \widetilde{\lambda}^{d,n}_i
	\end{array}\right]
	\left[\widetilde{q}^{1,n}_i,\cdots,\widetilde{q}^{d,n}_i\right]^\T
\end{equation*}
Define $\widebar{\lambda}^{k,n}_i = (r_k-r_{k-1})^{-1}\sum_{r\in\mathcal{K}_k}\widetilde{\lambda}^{r,n}_i$, our estimator of $S(\lambda)_t$ is
\begin{eqnarray}\label{est.eigenvalue.spectral}
	\widetilde{S}(\lambda)^n_t &=& \big[\widetilde{S}(\lambda)^{1,n}_t, \cdots, \widetilde{S}(\lambda)^{K,n}_t\big]^\T \nonumber \\
	\widetilde{S}(\lambda)^{k,n}_t &\equiv& k_n\Delta_n \sum_{h=0}^{N^n_t-1} \left[ 1 - \frac{2\theta\Phi_{00}}{\phi_0(0)^2\,k_n\,\Delta_n^{1/2+\delta}}\sum_{v\notin \mathcal{K}_k}\frac{\widetilde{\lambda}^{v,n}_{hk_n}}{\widebar{\lambda}^{k,n}_{hk_n}-\widetilde{\lambda}^{v,n}_{hk_n}}\right] \widebar{\lambda}^{k,n}_{hk_n}
\end{eqnarray}
where $\theta$, $\Phi_{00}$, $\phi_0(0)$ are defined in (\ref{tuning}), (\ref{phi.vars}). Note that $\widetilde{S}(\lambda)^{k,n}_t$ is a Riemann sum of simple eigenvalues or averages of repeated eigenvalues, with a multiplicative bias-correction factor. We have an expression for this multiplicative factor in terms of eigenvalue estimates. Since we do not have to evaluate Hessian of the functional, (\ref{est.eigenvalue.spectral}) is very advantageous for applications.

This advantage is also carried over to eigenvectors. We derive a closed-form expression for the nonlinear bias correction to the eigenvectors and define the estimator of $S(q^k)_t$ to be
\begin{equation}\label{est.eigenvector}
	\widetilde{S}(q^k)^n_t \equiv k_n\Delta_n\sum_{h=0}^{N^n_t-1} \left[1 + \frac{\theta\Phi_{00}}{\phi_0(0)^2\,k_n\,\Delta_n^{1/2+\delta}} \sum_{v\ne k}\frac{\widetilde{\lambda}^{k,n}_{hk_n}\widetilde{\lambda}^{v,n}_{hk_n}}{\big(\widetilde{\lambda}^{k,n}_{hk_n}-\widetilde{\lambda}^{v,n}_{hk_n}\big)^2} \right] \widetilde{q}^{k,n}_{hk_n}
\end{equation}

\begin{prop}\label{CLT.eigenvalue}
Assume Assumption \ref{A-v}, \ref{A-r}, and suppose (\ref{eigen.clusters}). We select the bandwidths $l_n$, $k_n$ and the truncation threshold $\nu_n$ according to (\ref{tuning.psd}). We have the following functional stable convergence in law:
\begin{equation*}
	\Delta_n^{-1/4+\delta/2}\left[\widetilde{S}(\lambda)^n-S(\lambda)\right]\overset{\mathcal{L}-s(f)}{\longrightarrow}Z(\lambda)
\end{equation*}
where $Z(\lambda)$ is a process defined on an extension of the space  $\left(\Omega,\F,(\F_t)_{t\ge0},\mathbb{P}\right)$, such that conditioning on $\mathcal{F}$ it is a continuous centered Gaussian martingale with variance
\begin{equation*}
	\widebar{E}[Z(\lambda)Z(\lambda)^\T|\F] = V(\lambda)
\end{equation*}
where $\widebar{E}$ is the conditional expectation operator on the extended probability space and
\begin{equation}\label{AVAR.eigenvalue}
	V(\lambda)_t = \frac{4\theta\Phi_{00}}{\phi_0(0)^2}
	\left[\begin{array}{ccc}
	(r_1-r_{0})^{-1}\int_0^t (\lambda^{r_1}_s)^2\ds s & & \\
	& \ddots & \\
	& & (r_K-r_{K-1})^{-1}\int_0^t (\lambda^{r_K})^2\ds s
	\end{array}\right]
\end{equation}
\end{prop}

If $\lambda^k_t$ is a simple eigenvalue, the corresponding eigenvector $q^k_t$ is differentiable with respect to $c_t$ (see Lemma \ref{lem.derivatives.simple} in Appendix \ref{apdx:pca}). In this case we can derive an inferential theory for the integrated eigenvector estimator (\ref{est.eigenvector}).
\begin{prop}\label{CLT.eigenvector}
Assume Assumption \ref{A-v}, \ref{A-r}, and suppose $\lambda^k_t$ is a simple eigenvalue of $c_t,\, \forall t\in[0,T]$. We control the bandwidths $l_n$, $k_n$ and the truncation threshold $\nu_n$ according to (\ref{tuning.psd}). We have the following functional stable convergence in law:
\begin{equation*}
	\Delta_n^{-1/4+\delta/2}\left[\widetilde{S}(q^k)^n-S(q^k)\right]\overset{\mathcal{L}-s(f)}{\longrightarrow}Z(q^k)
\end{equation*}
where $Z(q^k)$ is a process defined on an extension of the space  $\left(\Omega,\F,(\F_t)_{t\ge0},\mathbb{P}\right)$, such that conditioning on $\mathcal{F}$ it is a continuous centered Gaussian martingale with variance
\begin{equation*}
	\widebar{E}[Z(q^k)Z(q^k)^\T|\F] = V(q^k)
\end{equation*}
where $\widebar{E}$ is the conditional expectation operator on the extended probability space and
\begin{equation}\label{AVAR.eigenvector}
	V(q^k)_t = \frac{2\theta\Phi_{00}}{\phi_0(0)^2}\int_0^t \sum_{v\ne k} \frac{\lambda^k_s\lambda^v_s}{(\lambda^k_s-\lambda^v_s)^2}\, q^v_s\,q^{v,\T}_s \ds s
\end{equation}
\end{prop}

\begin{remk}\label{remk.compare}
Our asymptotic theory for PCA is based on the general result of $\widetilde{S}(g)^n_t$. We compare $\widehat{S}(g)^n_t$ and $\widetilde{S}(g)^n_t$ below in view of the asymptotics and practicality of PCA.
\begin{itemize}
\item Convergence rates: $\widehat{S}(g)^n_t$ attains the optimal convergence rate $n^{1/4}$; the convergence rate of $\widetilde{S}(g)^n_t$ is $n^{1/4-\delta/2}$ which is less than $n^{1/5}$ according to (\ref{tuning.psd}).
\item Dealing with negative eigenvalue estimate(s): If $\widehat{c}^n_i$ is not positive semidefinite, we can project it to $\s^d_+$. This is equivalent to setting negative eigenvalue estimate(s) to zero. $\widetilde{S}(g)^n_t$ is positive semidefinite by definition and does not suffer this problem.
\item Bias correction and uncertainty quantification: For $\widehat{S}(g)^n_t$ we need to evaluate $\partial_{jk}g$ and $\partial^2_{jk,lm}g$ for $j,k,l,m=1,\cdots,d$ in order to compute bias correction $\widehat{B}(g)^n_i$ and asymptotic variance $\widehat{V}_\Xi(g)^n_t$; this can be quite difficult in applications and motivated \cite{llx19} to propose the multiscale jackknife method. For $\widetilde{S}(g)^n_t$, we can use estimates $\widetilde{\lambda}^{k,n}_i$ and $\widetilde{q}^{k,n}_i$ to compute $\widetilde{B}(g)^n_i$ and $\widetilde{V}_\Sigma(g)^n_t$. Evident from (\ref{est.eigenvalue.spectral}), (\ref{est.eigenvector}), (\ref{AVAR.eigenvalue}), (\ref{AVAR.eigenvector}), they are straightforward to compute without much computational or memory overhead.
\end{itemize}
\end{remk}

We apply Proposition \ref{CLT.eigenvalue} and Proposition \ref{CLT.eigenvector} to conduct PCA on TAQ data in Section \ref{sec:taq}. We choose PCA based on $\widetilde{S}(g)^n_t$ because it scales better for the large TAQ data.

\subsection{More examples}
\begin{eg}[Uncertainty quantification]
Asymptotic variances are often functionals of volatility. In the univariate setting, the so-called ``quarticity'' defined as $\int_0^tc_s^2\ds s$ appears in the asymptotic variances of many estimators of integrated volatility. The multivariate counterpart involves $\int_0^t c_s^{jl}c_s^{km}+c_s^{jm}c_s^{kl}\ds s$ in \cite{jr13} and $\Xi(c_s,\rr_s)$ in (\ref{AVAR}). For this reason, the volatility functional estimator effectively facilitates uncertainty quantification for various volatility-related estimators.
\end{eg}

\begin{eg}[Laplace transform]
\cite{tt12a} put forward an estimator of the realized Laplace transform of volatility defined as
	\[\int_0^t e^{iwc_s}\ds s\]
This transform can be viewed as the characteristic function of volatility under the occupation measure. By matching the the moments of realized Laplace transform with those induced by a model, we can estimate model parameter(s) or test the model. An open question noted by \cite{tt12a} is the estimation of realized Laplace transform using noisy data. By the nonparametric estimation of volatility path in the first stage and the Riemann summation of bias-corrected functional plug-ins in the second stage, this paper presents a rate-optimal solution to the open question.
\end{eg}

\begin{eg}[Generalized method of moments (GMM)]
\cite{lx16} proposed the generalized method of integrated moments for financial high-frequency data. In estimating an option pricing model, one observes the process $Z_t = (t, X_t, r_t, d_t)$ where $X_t$ is the price of the underlying observed without any noise, $r_t$ is the short-term interest rate, $d_t$ is the dividend yield. One model of the arbitrage-free option price under the risk-neutral probability measure is
\begin{equation*}
	\beta_t = f(Z_t,c_t;\theta^*)
\end{equation*}
where $f$ is deterministic, $\theta^*$ is the true model parameter. The observed option price is often modeled as
\begin{equation*}
	O_{i\Delta_n} = \beta_{i\Delta_n} + \epsilon_i
\end{equation*}
where $\epsilon_i$ is pricing error and $\E(\epsilon_i)=0$. Let $g(Z_t,c_t;\theta) = \E[O_t - f(Z_t,c_t;\theta)]$, then we have the following integrated moment condition:
\begin{equation*}
	\int_0^t g(Z_s,c_s;\theta^*)\,\ds s = 0
\end{equation*}
Utilizing noisy observations at higher frequencies, this paper provides a means to compute a bias-corrected sample moment function of GMM.
\end{eg}

\begin{eg}[Linear regression]
In the practice of linear factor models and financial hedging, one faces the tasks of computing the factor loadings and the hedge ratios. These tasks can be formulated as the estimation of the coefficient $\beta$ in the time-series linear regression model
\begin{equation*}
	Z_t^c = \beta^\T S_t^c + R_t
\end{equation*}
where
\begin{equation*}
\left\{\begin{array}{cl}
	S_t= & S_0 + \int_0^tb^S_u\ds u + \int_0^t\sigma^S_u\ds W^S_u + J^S_t\\
	Z_t= & Z_0 + \int_0^tb^Z_u\ds u + \beta^\T\int_0^t\sigma^S_u\ds W^S_u + \int_0^t\sigma^R_u\ds W^R_u + J^Z_t
\end{array}\right.
\end{equation*}
$\langle W^S,W^R\rangle=0$, $S_t\in\R^{d-1}$, $Z_t\in\R$, and $S^c$, $Z^c$ are the continuous parts of the It\^o semimartingales.

Let $X = (S^\T,Z)^\T$, we can write $X_t=X_0 + \int_0^tb_u\ds u + \int_0^t\sigma_u\ds W_u + J_t$ where $b=(b^{S,T}, b^Z)^\T$, $W=(W^{S,\T},W^R)^\T$, $J=(J^{S,\T},J^Z)^\T$ and
\begin{equation*}
	\sigma = \left[
	\begin{array}{cc}
		\sigma^S & 0 \\
		\beta^\T\sigma^S & \sigma^R
	\end{array}\right]
\end{equation*}
so
\begin{equation*}
	c = \sigma\sigma^\T = 
	\left[\begin{array}{cc}
		\sigma^S\sigma^{S,\T}         & \sigma^S\sigma^{S,\T}\beta \\
		\beta^\T\sigma^S\sigma^{S,\T} & \beta^\T\sigma^S\sigma^{S,\T}\beta + (\sigma^R)^2
	\end{array}\right]
	\coloneqq
	\left[\begin{array}{cc}
		c^{SS} & c^{SZ} \\
		c^{ZS} & c^{ZZ}
	\end{array}\right]
\end{equation*}
hence by letting $g(c)=c^{SS,-1}c^{SZ}$, we have $\beta=t^{-1}S(g)_t$. \cite{ltt17} proposed a method under the no-noise assumption. When the observations contain noise, the methodology of this paper can extend the method of \cite{ltt17} to wider applicability.
\end{eg}

\section{Empirical analysis of TAQ data}\label{sec:taq}
We analyze the transaction data of S\&P 100 constituents. Each stock is sampled every second during 09:35 - 15:55 EST, so we have 22,800 observations for each stock on each business day. We focus on data from 09/10/2003 to 12/31/2019. During this period there were 4106 business days in 852 weeks. By focusing on the time window 09:35 - 15:55 EST, we remove overnight returns which contain big jumps and returns incurred in the first and last five minutes of trading hours which are relatively volatile. 10 illiquid assets are excluded from our statistical analysis.

Instantaneous volatility matrix and its eigenvalues and eigenvectors were estimated for every business day. The realized eigenvalues and realized eigenvectors were computed weekly by aggregating the instantaneous eigenvalues and instantaneous eigenvectors.

The temporal evolution of four leading eigenvalue estimates along with confidence intervals are plotted in Figure \ref{fig.eigvals}. It shows $\lambda_q/\sum_j\lambda_j$'s which are the proportions of total variation that can be explained by the corresponding principal components. The first eigenvalues are conspicuously separated from the others and indicate the first principal component explains more than 60\% of the cross-section variation for the majority of time. Moreover, it is evident from this dataset that the first four leading eigenvalues were simple rather than repeated. The estimates of eigenvalues tell us that a low-dimensional model can largely explain the systematic comovement of these stocks.

\begin{figure}
	\centering
	\includegraphics[width=1.\textwidth]{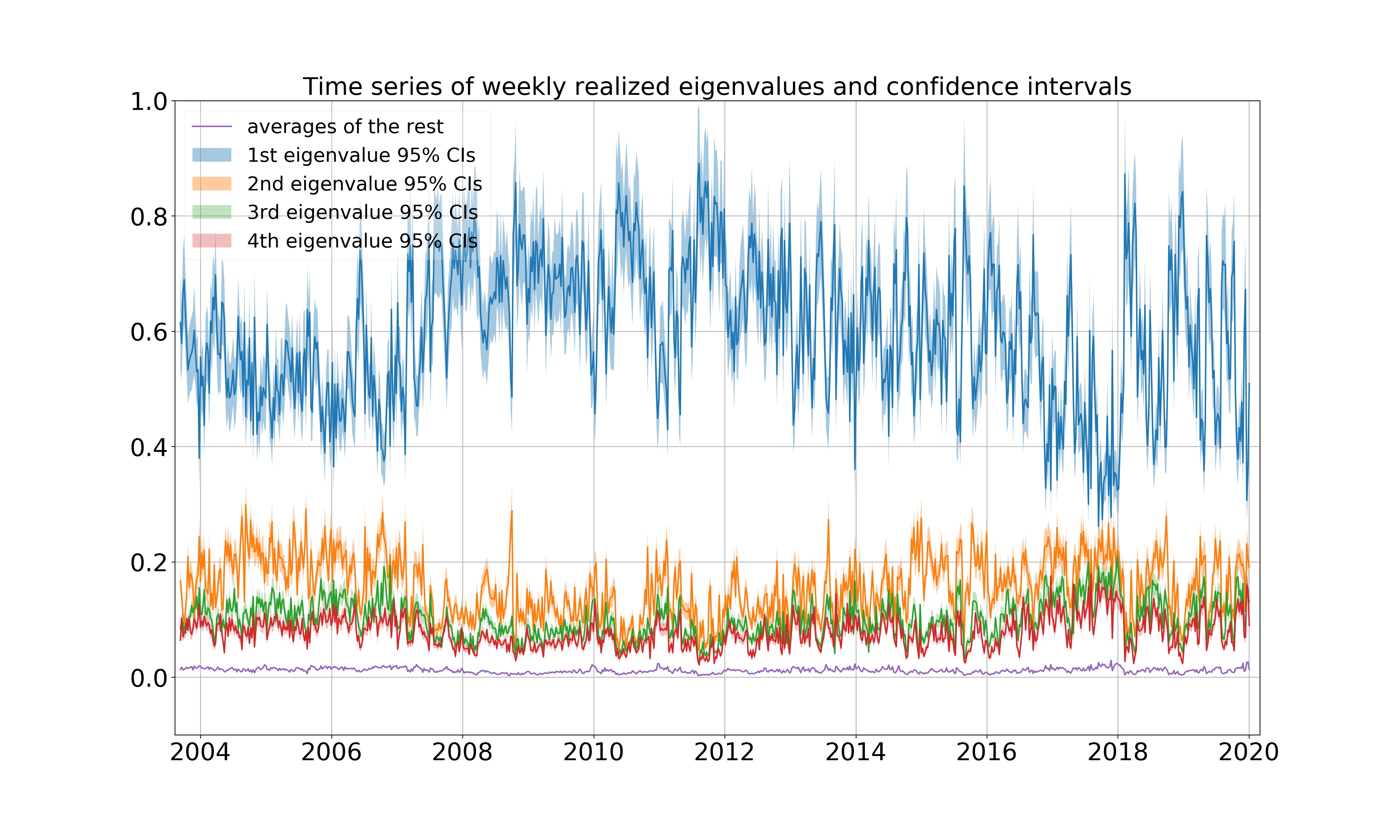}
	\caption{Weekly integrated eigenvalues and the confidence intervals}
	\label{fig.eigvals}
\end{figure}

It is reasonable to interpret the first principal component as a market indicator on how the economy performs in general. As all the major corporations and their stock values are moving with the overall economy, the eigenvector corresponding to the first principal component is expected to comprise only positive entries.

An interesting question is how the first principal component is compared with S\&P 100 index. \cite{cmz20} calculated the accumulated returns from 2007 to 2017 of the portfolio based on the eigenvector $q^1$ associated with the first principal component. \cite{cmz20} normalized the eigenvector so that $\sum_j^dq^{1,j}=1$ and used $q^{1,j}$'s as portfolio weights so that the resultant portfolio is self-financing. A curious finding of \cite{cmz20} is that this portfolio outperformed the market index hence beat the passive investment. Here the author would like to test this portfolio using pre-averaging-based PCA, over a longer time span.

The principal components are $\int_0^t q^r(c_{s-})\ds X_s,\, r=1,2,\cdots$, where $X_t^j = \log(S_t^j)$ and $S_t$ is the vector of asset prices at time $t$. However, one obviously can not transact in logarithmic currency, so a pending question is how to find a portfolio that is analogous to the first principal component and requires trading on the currency's original scale.

\cite{cmz20} proposed the following construction. Let 
$$R_{t,\Delta} = \big[(S^1_{t+\Delta}-S^1_t)/S^1_t,\cdots,(S^d_{t+\Delta}-S^d_t)/S^d_t\big]^\T$$ 
be the vector of returns in percentage. When $\|R_{t,\Delta}\|$ is small,
\[ \log\big(1 + q^{1,\T} R_{t,\Delta}\big) \approx q^{1,\T} R_{t,\Delta}, \]
and
\[ q^{1,j}\frac{S^j_{t+\Delta}-S^j_t}{S^j_t} \approx q^{1,j}\log\Big(1 + \frac{S^j_{t+\Delta}-S^j_t}{S^j_t}\Big) = q^{1,j}\log\Big(\frac{S^j_{t+\Delta}}{S^j_t}\Big), \]
so
\begin{equation*}
q^{1,\T} (X_{t+\Delta}-X_t) \approx \log\big(1 + q^{1,\T} R_{t,\Delta}\big).
\end{equation*}
Therefore, we consider $\sum_t\log(1+q^{1,\T} R_{t,\Delta})$ as an approximation to the first realized principal component. It is also the accumulated log return if there is no transaction cost and $\sum_j^d q_{1,j}=1$. In this case, $\exp\big\{\sum_t\log(1+q^{1,\T} R_{t,\Delta})\big\} - 1 = \prod_t(1+q^{1,\T} R_{t,\Delta}) - 1$
is the accumulated return of this self-financing portfolio in the absence of transaction cost.

Since the volatility matrix and the eigenvector $q^1$ are updated on every business day, the portfolio can be re-weighted daily. One caveat is that the numerical values of eigenvector estimates are relative less stable. Roughly speaking, the eigenvector viewed as a function of the corresponding positive semidefinite matrix is less smooth than the eigenvalue. Due to this reason, heuristically, an error in the covariance matrix estimation tends to result in a larger error in the eigenvector estimate, as compared to the eigenvalue counterpart. To stabilize eigenvector estimates, it is advisable to compute moving averages of the initial eigenvector estimates over a window of six business days. As a result we re-weight the portfolio on all the business days from 09/17/2019 to 12/31/2019.

Before the portfolio re-allocation according to the new weights on the next business day, one can either clear positions at the end of trading hours or let positions stand overnight. In the second scenario, the portfolio bears the overnight risk premium and absorbs overnight returns.

The accumulated log returns of the portfolio that mimics the first principal component are plotted in Figure \ref{fig.pcportfolio}. The portfolio returns were computed using daily prices of S\&P 100 constituents which can be obtained from the Compustat database. As we can see, return morphologies of the portfolio that mimics the first principal component resemble that of S\&P 100 index. If the portfolio was cleared at the end of trading hours, it slightly underperformed the index. Interestingly enough, if the portfolio that stood overnight, it manifestly outperformed the index. It is evident that this portfolio and the factor corresponding to the first principal component earned a significant risk premium during the overnight periods, and implies a difference in the intraday and overnight risk-reward patterns. This finding is consistent with the empirical analysis of \cite{cmz20}, though we use a method that is more accurate and analyze data over a longer time span.

\begin{figure}
	\centering
	\includegraphics[width=.9\textwidth]{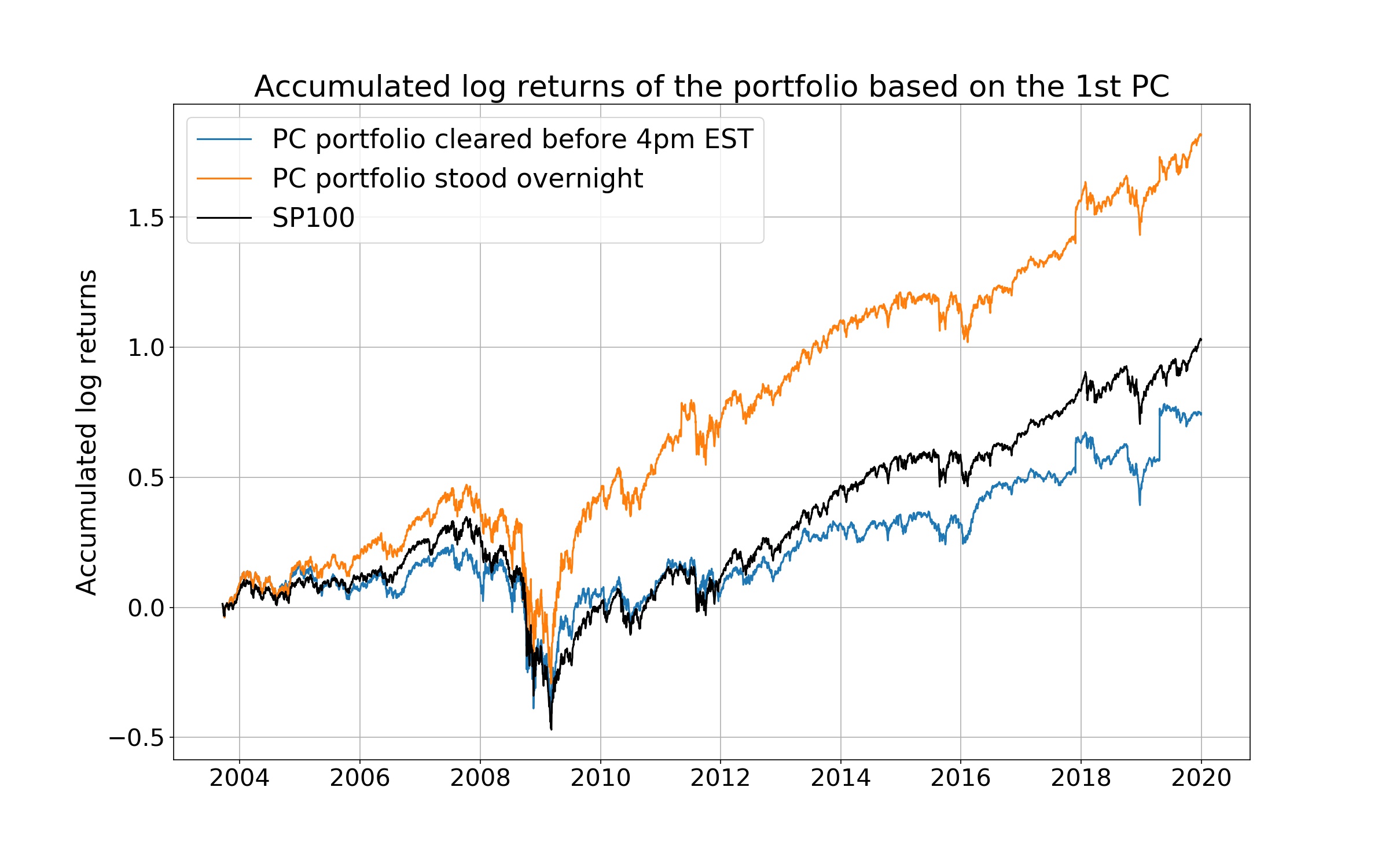}
	\caption{Accumulated log returns of the portfolio based on the first principal component} 
	\label{fig.pcportfolio}
\end{figure}


\section{Conclusion}\label{sec:conclude}
In this paper, we proposed statistical methods for general nonlinear functionals of volatility matrix, and developed statistical theory to guide bias correction, quantify statistical uncertainty and perform statistical inference. Because the methods can utilize data that contains microstructure noise and jumps, and because volatility matrix functionals form a unifying framework that accommodates many applications, this paper provides a viable answer to the open problem of noise and extends numerous statistical methods to much wider applicability. This brings significant statistical gains, such as principled approaches to take advantage of much bigger datasets and improved statistical accuracy.

The tuning choice is an essential component in the inference of volatility matrix functionals, and it has been the main motivation of several recent works under the no-noise assumption. The importance of tuning must not be neglected neither in the noisy setting. We derived tuning ranges for pre-averaging and jump truncation. We would like to stress that the tuning parameters need to be quite different from those under the no-noise assumption.

Asymptotic analysis shows there are non-negligible higher order biases beyond linear terms. We can effectively debias our estimation by quadratic approximation. Given the tuning choice and bias correction, we established consistency, convergence rates, and asymptotic normality by functional stable convergence of processes. We derived stable CLT for both the rate-optimal method and the method that uses positive semidefinite plug-ins. The asymptotic covariance matrices can be consistently estimated, which paves the way for feasible statistical inference.

Based on the general results, we then put forward the method and theory of PCA to analyze noisy high-frequency financial data. Particularly we gave a demonstration on the TAQ dataset and analyzed the temporal structures of eigenvalues and the first principal component of S\&P 100 stocks. Our PCA method is not only statistically sound in the presence of noise and jump, but also computationally practical for large datasets. PCA is one example among many possible applications of this paper. We demonstrate some success in pushing the inferential framework for general functionals into a new frontier of potentials and possibilities, and lending the effort to enable applications to employ noisy high-frequency data where exciting new discoveries await.

\bigskip

\section*{Acknowledgements}
The author is indebted to Jean Jacod and Per A. Mykland for encouragement, detailed discussion and comments. This work also benefits from discussions with Yoann Potiron, Viktor Todorov, Dacheng Xiu.\\
The author was supported by NSF grant DMS 17-13129 and 2019 Stevanovich Fellowship from the University of Chicago.
\bigskip

\appendix
\section{Assumptions}\label{sec:assump}
We present formally the model specification and assumptions in this appendix. 

The pure jump process $J$ in $X$ is modeled by
\begin{equation}\label{J}
J_t = \int_{(0,t]\times E}\delta(s,x)\, \mathfrak{p}(\mathrm{d}s,\ds x)
\end{equation}
where $\delta$ is a $\mathbb{R}^d$-valued predictable function on $\R^+\times E$, $E$ is a Polish space, $\mathfrak{p}$ is a Poisson random measure with compensator $\mathfrak{q}(\ds u, \ds x)=\ds u\otimes\lambda(\ds x)$, $\lambda$ is a $\sigma$-finite measure on $E$ and has no atom.

The stochastic volatility matrix is positive semidefinite almost surely and of the form
\begin{equation}\label{c}
c_t=c_0+\int_0^tb^{(c)}_s\ds s+\int_0^t\sigma^{(c)}_s\ds W_s + J^{(c)}_t
\end{equation}
where $b^{(c)} \in \R^{d\times d}$, $\sigma^{(c)} \in \R^{d \times d \times d'}$, $b^{(c)}$ and $\sigma^{(c)}$ are adapted and c\`adl\`ag; $J^{(c)}$ is the volatility jump, which is modeled by
\begin{multline}\label{Jc}
	J^{(c)}_t = \delta^{(c)}(s,x)\1{\|\delta^{(c)}(s,x)\|\le1}\, (\mathfrak{p}-\mathfrak{q})(\mathrm{d}s,\ds x) \\
	+ \int_{(0,t]\times E}\delta^{(c)}(s,x)\1{\|\delta^{(c)}(s,x)\|>1}\, \mathfrak{p}(\mathrm{d}s,\ds x)
\end{multline}
where $\delta^{(c)}$ is a $\mathbb{R}^{d\times d}$-valued predictable function on $\R^+\times E$.

To guarantee that $c_t$ is positive semidefinite, one needs to impose additional parametric or semiparametric restrictions in concrete applications. The results in this paper hold true for any process $c$ that is positive semidefinite and satisfies (\ref{c}).

The regularization on drift, diffusion coefficient and jump activity is given below.
\begin{assump}{A-$\nu$}[regularity]\label{A-v} 
	The drift process $b$ has $\frac{1}{2}$-H\"older sample path in the mean-square sense, i.e., $\forall t,s\ge0$,
	\begin{equation*}
	E\Big(\sup_{u\in[0,s]}\|b_{t+u}-b_t\|\big|\F^{(0)}_t\Big)\le Ks^{1/2},\,a.s.;
	\end{equation*}
	$c$ is of the form (\ref{c}), there is a sequence of triples $(\tau_m,\mathcal{S}_m,\Gamma_m)$, where $\tau_m$ is a stopping time and $\tau_m\nearrow\infty$;  $\s_m\subset\mathcal{S}^+_d$ is convex, compact such that 
	\begin{equation*}
	t\in[0,\tau_m] \Rightarrow c_t\in\s_m;
	\end{equation*} 
	$\Gamma_m$ is a sequence of bounded $\lambda$-integrable functions on $E$, such that
	\begin{equation*}
	t\in[0,\tau_m] \Longrightarrow \left\{
	\begin{array}{l}
	\|b_t\| + \|\sigma_t\| + \|b^{(c)}_t\| + \|\sigma^{(c)}_t\| \le m\\
	\|\delta(t,x)\|^\nu\wedge1\le\Gamma_m(x),\,\,\nu\in[0,1)\\
	\|\delta^{(c)}(t,x)\|^2\wedge1\le\Gamma_m(x)
	\end{array}\right.
	\end{equation*}
\end{assump}

Let $\big(\Omega^{(0)},\F^{(0)},(\F^{(0)}_t),\mathbb{P}^{(0)}\big)$ be a filtered probability space with respect to which $X$ and $c$ are $(\F^{(0)}_t)$-adapted; let $\big(\Omega^{(1)},\F^{(1)},(\F^{(1)}_t),\mathbb{P}^{(1)}\big)$ be another filtered probability space accommodating $Y$; $\forall t\ge0$, $\forall A\in\F^{(0)}$, let $Q_t(A,\cdot)$ be a conditional probability measure on $\left(\Omega^{(1)},\F^{(1)}\right)$. 

All the stochastic dynamics above can be described on the filtered extension $\left(\Omega,\F,(\mathcal{F}_t),\mathbb{P}\right)$, where
\begin{equation}\label{prob.sp}
\left\{\begin{array}{l}
	\Omega=\Omega^{(0)}\times\Omega^{(1)}\\
	\F=\F^{(0)}\otimes\F^{(1)}\\
	\F_t=\bigcap_{s>t}\big(\F^{(0)}_s\otimes\F_s^{(1)}\big),\hspace{5mm}  \widetilde{\F}_t=\bigcap_{s>t}\big(\F^{(0)}\otimes\F_s^{(1)}\big)\\
	\mathbb{P}\left(A\times \mathrm{d}\omega\right)=\mathbb{P}^{(0)}(A)\cdot\otimes_{t\ge0}Q_t(A,\ds\omega),\,\forall A\in\F^{(0)}
\end{array}\right.
\end{equation}

Let
\begin{equation*}
	\epsilon_t(\omega) = Y_t(\omega) - X_t
\end{equation*}
and the conditional covariance process of noise can be written as
\begin{equation}\label{def.r}
	\gamma_t = \int_{\Omega^{(1)}}\epsilon_t(\omega)\,\epsilon_t(\omega)^\T\,Q_t(\cdot,\mathrm{d}\omega)
\end{equation}
\begin{assump}{A-$\gamma$}[noise]\label{A-r}
$\forall t\in\mathbb{R}^+$,
	\[\int_{\Omega^{(1)}}\epsilon_t(\omega)\,Q_t(\cdot,\ds\omega) = 0\]
$\forall t\ne s$, $\forall A\in\F^{(0)}_{s\wedge t}$
	\[\int_{\Omega^{(1)}\times\Omega^{(1)}}\epsilon_t(\omega)\,\epsilon_s(\omega)^\T\,Q_t(A,\ds\omega)\,Q_s(A,\ds\omega)=0\]
furthermore,
\begin{equation*}
	\gamma_t = \gamma_0+\int_0^tb^{(r)}_s\ds s + \int_0^t\sigma^{(r)}_s\ds W_s + \int_{(0,t]\times E}\delta^{(r)}(s,x)\,\mathfrak{p}(\mathrm{d}s,\ds x)
\end{equation*}
for the same $\tau_m$, $\Gamma_m$ in assumption \ref{A-v} and $\forall q\le8$
\begin{equation*}
	t\in[0,\tau_m] \Longrightarrow \left\{
	\begin{array}{l}
		\|b^{(r)}_t\| + \|\sigma^{(r)}_t\| + \E_{\Omega^{(1)}}\|\epsilon_t(\omega)\|^q \le m\\
		\|\delta^{(r)}(t,x)\|^2\wedge1\le\Gamma_m(x)
	\end{array}\right.
\end{equation*}
\end{assump}

\section{Theoretical preliminaries for proofs}\label{apdx:prelim}
We use the notation and model described in Section 2 and Appendix A in the main article. In the sequel, the constant $K$ changes across lines but remains finite; $K_q$ is a constant depending on $q$; $a_n\asymp b_n$ means both $a_n/b_n$ and $b_n/a_n$ are bounded for large $n$; $\E(\cdot)$ denotes the expectation operator on $(\Omega^{(0)},\F^{(0)})$ or $(\Omega,\F)$; $E(\cdot|\mathcal{H})$ denotes the conditional expectation operator, with $\mathcal{H}$ being $\F^{(0)}_t$, $\F^{(1)}_t$, $\F_t$, $\widetilde{\F}_t$; $E^n_i(\cdot)$ denotes $E(\cdot|\F^n_i)$. Six useful results are stated below.

I. By a localization argument from section 4.4.1 in \cite{jp12}, without loss of generality we can assume $\exists$ a constant K, a bounded $\lambda$-integrable function $\Gamma$ on $E$, a convex compact subset $\mathcal{S}\in\mathcal{S}^+_d$ and $\epsilon>0$, $g\in\mathcal{C}^3(\mathcal{S}^\epsilon)$ where $\mathcal{S}^\epsilon$ denotes the $\epsilon$-enlargement of $\mathcal{S}$ (see (3.8) in the main article), such that
\begin{equation}\label{SA-v}
\left\{\begin{array}{l}
	\|b\| + \|\sigma\| + \|b^{(c)}\| + \|\sigma^{(c)}\| + \|b^{(r)}_t\| + \|\sigma^{(r)}_t\| + \E\|\epsilon_t\|^q \le K,\;\forall q\le8\\
	\|\delta(t,x)\|^\nu\wedge1\le\Gamma(x),\;\nu\in[0,1)\\
	\|\delta^{(c)}(t,x)\|^2\wedge1+\|\delta^{(r)}(t,x)\|^2\wedge1\le\Gamma(x)\\
	c\in\mathcal{S}
\end{array}\right.
\end{equation}

II. Define a continuous It\^o semimartingale with parameters identical to those of (2.1) in the main article
\begin{equation*}
	X'_t = X_0 + \int_0^tb_s\ds s + \int_0^t\sigma_s\ds W_s
\end{equation*}
Let $Y^*= Y - X + X'$. Based on (3.2) in the main article, define
\begin{eqnarray*}
	\widehat{c}^{*n}_i &=& \frac{1}{(k_n-l_n)\Delta_n}\sum_{h=1}^{k_n-l_n+1}\left(\widebar{Y}^{*n}_{i+h}\cdot\widebar{Y}^{*n,\T}_{i+h} - \widehat{Y}^{*n}_{i+h}\right) \\
	\widetilde{c}^{*n}_i &=& \frac{1}{(k_n-l_n)\Delta_n}\sum_{h=1}^{k_n-l_n+1} \widebar{Y}^{*n}_{i+h}\cdot\widebar{Y}^{*n,\T}_{i+h}
\end{eqnarray*}
In the upcoming derivation, $\|\widehat{c}^n_i-\widehat{c}^{*n}_i\|$ is tightly bounded provided $\nu_n$ is properly chosen. The focus then will be shifted from $\widehat{c}^n_i$ to $\widehat{c}^{*n}_i$. Ditto for $\widetilde{c}^n_i$ and $\widetilde{c}^{*n}_i$. The volatility ``estimators'' calculated on continuous sample paths are more tractable.

III. By estimates of It\^o semimartingale increments, $\forall$ finite stopping time $\tau$,
\begin{equation}\label{classic}
\left\{\begin{array}{ll}
	\big\|E\big(X'_{\tau+s}-X'_{\tau}|\F^{(0)}_\tau\big)\big\| + \\
	\hspace{4mm} \big\|E\big(c_{\tau+s}-c_{\tau}|\F^{(0)}_\tau\big)\big\| + \big\|E\big(\gamma_{\tau+s}-\gamma_\tau|\F^{(0)}_\tau\big)\big\|& \le Ks\\
	E\left(\sup_{u\in[0,s]}\left\|X'_{\tau+u}-X'_\tau\right\|^q|\F^{(0)}_\tau\right)& \le Ks^{q/2}\\
	E\left(\sup_{u\in[0,s]}\left\|c_{\tau+u}-c_\tau\right\|^q + \left\|\gamma_{\tau+u}-\gamma_\tau\right\|^q|\F^{(0)}_\tau\right)& \le Ks^{(q/2)\wedge1}
\end{array}\right.
\end{equation}
by Lemma 2.1.7, Corollary 2.1.9 in \cite{jp12},
\begin{equation}\label{jump.bounds}
\left\{\begin{array}{l}
	E\left(\sup_{u\in[0,s]}\|J_{\tau+u}-J_\tau\|^q|\F^{(0)}_\tau\right)\le K_q\,sE\big[\widehat{\delta}(q)_{\tau,s}|\F^{(0)}_\tau\big]\\
	E\left[\sup_{u\in[0,s\wedge1]}\left(\frac{\|J_{\tau+u}-J_\tau\|}{s^w}\wedge1\right)^q|\F^{(0)}_\tau\right]\le
	K\,s^{(1-w\nu)(q/\nu\wedge1)}a(s)
\end{array}\right.
\end{equation}
where $\widehat{\delta}(q)_{t,s}\equiv s^{-1}\int_t^{t+s}\int_E\|\delta(u,x)\|^q\,\lambda(\mathrm{d}x)\ds u$ and $a(s)\to0$ as $s\to0$.

IV. Let $\varphi_n(t)=\sum_{h=1}^{l_n-1}\varphi^n_h\mathds{1}_{((h-1)\Delta_n,h\Delta_n]}(t)$. For a generic process $U$, define
\begin{equation}\label{def.Un}
	U^n_{t,s} = \int_t^{t+s}\varphi_n(u-t)\ds U_u
\end{equation}
This quantity is useful in analyzing $\widebar{U}^n_i$, with $U$ being $X$ or $Y$.

V. For $p\in\mathbb{N}^+$, $l,m=0,1$, by (4.1) in the main article and Riemann summation,
\begin{equation}\label{sum.phi.quadratic}
\sum_{h=i}^{i+pl_n-2}\,\sum_{h'=h+1}^{i+pl_n-1}\phi_l\Big(\frac{h'-h}{l_n}\Big)\phi_m\Big(\frac{h'-h}{l_n}\Big) = l_n^2\left(p\Phi_{lm}-\Psi_{lm}\right) + O(pl_n)
\end{equation}

VI. By Jensen's inequality and Doob's maximal inequality, we have the following lemma:
\begin{lem}\label{Doob.max.ineq}
	Let $Z_i,i=1,\cdots,M$ be random variables, $\mathcal{H}_i=\sigma(Z_1,\cdots,Z_i)$ be the $\sigma$-algebra generated by $Z_1,\cdots,Z_i$, then
	\begin{equation*}
	\E\left(\sup_{m=1,\cdots,M}\bigg\|\sum_{i=1}^m\left[Z_i - E\left(Z_i|\mathcal{H}_i\right)\right]\bigg\|\right)
	\le K\left(\sum_{i=1}^{M}\E\left(\|Z_i\|^2\right)\right)^{1/2}
	\end{equation*}
\end{lem}

\section{Derivation for Theorem 1}\label{apdx:thm1}
\subsection{Properties of spot estimator: I. jumps}\label{apdx:jmp}
By Assumption A-$\nu$, A-$\rr$, (\ref{SA-v}), (\ref{classic}), (\ref{jump.bounds}), and $\widebar{J}^n_i = \psi_n^{-1/2}\widebar{J}^n_{i\Delta_n,(l_n-1)\Delta_n}$ from (\ref{def.Un}), we know that under the tunning (\ref{tuning})
\begin{equation}\label{est.Y.bars.hats}
\left\{\begin{array}{lcl}
	E^n_i\left(\big\|\widebar{Y}^{*n}_i\big\|^q\right) &\le& K_q\Delta_n^{q/2},\; \forall q\le8\\
	E^n_i\left(\big\|\widebar{Y}^n_i\big\|^q\right) &\le& K_q\Delta_n^{(q/2)\wedge(q/4+1/2)},\; \forall q\le8\\
	E^n_i\left(\big\|\widehat{Y}^{*n}_i\big\|^q \vee \big\|\widehat{Y}^n_i\big\|^q\right) &\le& K_q\Delta_n^q,\; \forall q\le4\\
	E^n_i\Big[\Big(\frac{\|\widebar{J}^n_i\|}{\Delta_n^w}\wedge1\Big)^q\Big] &\le& K_q\Delta_n^{[1/2-(w-1/4)\nu]\times[1\wedge(q/\nu)]}\,a_n,\; \forall q\le8
\end{array}\right.
\end{equation}
for some $a_n\to0$. We can write
	\[\left\|\big(\widebar{Y}^n_i\cdot\widebar{Y}^{n,\mathrm{T}}_i\1{\|\widebar{Y}^n_i\|\le\nu_n} - \widehat{Y}^n_i\big) - \big(\widebar{Y}^{*n}_i\cdot\widebar{Y}^{*n,\mathrm{T}}_i - \widehat{Y}^{*n}_i\big)\right\| \le \sum_{r=1}^3\eta^{n,r}_i\]
where
\begin{eqnarray}
	\eta^{n,1}_i &=& \Big\| \widebar{Y}^n_i\cdot\widebar{Y}^{n,\mathrm{T}}_i\1{\|\widebar{Y}^n_i\|\le\nu_n} - \widebar{Y}^{*n}_i\cdot\widebar{Y}^{*n,\mathrm{T}}_i\1{\|\widebar{Y}^{*n}_i\|\le\nu_n} \Big\| \label{def.etas}\\
	\eta^{n,2}_i &=& \big\|\widehat{Y}^n_i - \widehat{Y}^{*n}_i\big\| \nonumber\\
	\eta^{n,3}_i &=& \big\|\widebar{Y}^{*n}_i\big\|^2 \1{\|\widebar{Y}^{*n}_i\|>\nu_n} \nonumber
\end{eqnarray}

Let $u_n=\nu_n/\Delta_n^{1/2}$, $Z^n_i=\|\widebar{Y}^{*n}_i\|/\Delta_n^{1/2}$, $Q^n_i=(\|\widebar{J}^n_i\|/\Delta_n^{1/2})\wedge1$, $V^n_i=(\|\widebar{J}^n_i\|/\Delta_n^{\rho})\wedge1$, we have
	\[\Delta_n\eta^{n,1}_i \le u_n^{-2/(1-2\rho)}(Z^n_i)^{2+2/(1-2\rho)} + (1+Z^n_i) \left[Q^n_i + u_n^2 (V^n_i)^2\right]\]
By successive conditioning and (\ref{jump.bounds}), there is a sequence $a_n\to0$ such that
\begin{equation*}
	E^n_i\big[(\eta^{n,1}_i)^q\big] \le K_q\,\Delta_n^{2\rho q + 1/2-(\rho-1/4)\nu}a_n
\end{equation*}
Analyzing $\eta^{n,2}_i$ with (\ref{classic}), (\ref{jump.bounds}), analyzing $\eta^{n,3}_i$ with Cauchy-Schwarz inequality, Markov's inequality, (\ref{est.Y.bars.hats}), we can get the following lemma:

\begin{lem}\label{est.chat-chat*}
	Assume (\ref{SA-v}), (\ref{tuning}), Assumption A-$\nu$ and A-$\rr$, then $\exists\,a_n\to0$ such that $\forall q\le4$,
	\begin{equation*}
		E^n_i\left(\|\widehat{c}^n_i-\widehat{c}^{*n}_i\|^q\right) \le K_q\left(a_n\Delta_n^{1/2-(\rho-1/4)\nu-(1-2\rho)q} + \Delta_n^{1/2} \right)
	\end{equation*}
\end{lem}

\subsection{Properties of spot estimator: II. continuous part}\label{apdx:contin}
\subsubsection{variables}
``If there is a rifle handing on the wall in act one, it must be fired in the next act. Otherwise it has no business being there'', said the Russian playwright Anton Chekhov. Define
\begin{eqnarray*}
	\widebar{C}^n_i&=&\frac{1}{\psi_n}\sum_{h=1}^{l_n-1}(\varphi^n_h)^2\Delta^n_{i+h}C,\,\, C_t=\int_0^tc_s\ds s\\
	D^n_i&=&\widebar{C}^n_i-c^n_i\Delta_n\\
	\Gamma^n_h&=&\Gamma^n_{h,h},\,\,\Gamma^n_{h,h'}=\frac{1}{\psi_n}\sum_{v=h\vee h'}^{h\wedge h'+l_n-1}(\varphi^n_{v-h+1}-\varphi^n_{n-h})(\varphi^n_{v-h'+1}-\varphi^n_{v-h'})\gamma^n_v\\
	R^n_i&=&\widehat{Y}^{*n}_i - \Gamma^n_i\\
	\zeta^n_i&=&\widebar{Y}^{*n}_i\cdot\widebar{Y}^{*n,\T}_i-\widebar{C}^n_i-\Gamma^n_i
\end{eqnarray*}
given $p\in \mathbb{N}^+$, define
\begin{eqnarray*}
	\zeta(W,p)^n_i &=& \sum_{h=i}^{i+pl_n-1}\big[(\sigma^n_i\widebar{W}^n_h)\cdot(\sigma^n_i\widebar{W}^n_h)^\T-\widebar{C}^n_h\big]\\
	\zeta(X,p)^n_i &=& \sum_{h=i}^{i+pl_n-1}\big(\widebar{X}^n_h\cdot\widebar{X}^{n,\T}_h - \widebar{C}^n_h\big)\\
	\zeta(X,p)'^n_i &=& \sum_{h=i}^{i+pl_n-2}\,\sum_{h'=h+1}^{i+pl_n-1}\widebar{X}^n_h\cdot\widebar{X}^{n,\T}_{h'}\phi_1\Big(\frac{h'-h}{l_n}\Big)\\
	\zeta(p)^n_i &=& \sum_{h=i}^{i+pl_n-1}\zeta^n_h
\end{eqnarray*}

Let $m(n,p)=\left\lfloor\frac{k_n-l_n}{(p+1)l_n}\right\rfloor$, $a(n,p,h)=1+h(p+1)l_n$, $b(n,p,h)=a(n,p,h)+pl_n$, then the estimation error of $\widehat{c}^{*n}_i$ can be decomposed as
\begin{equation}\label{errors.chat*}
	\beta^n_i \coloneqq \widehat{c}^{*n}_i-c^n_i = \xi^{n,0}_i + \xi^{n,1}_i + \xi^{n,2}_i + N(p)^n_i + M(p)^n_i
\end{equation}
where
\begin{eqnarray*}
	\xi^{n,0}_i&=&\frac{1}{k_n-l_n}\sum_{h=1}^{k_n-l_n+1}c^n_{i+h}-c^n_i\\
	\xi^{n,1}_i&=&\frac{1}{(k_n-l_n)\Delta_n}\sum_{h=1}^{k_n-l_n+1}D^n_{i+h}\\
	\xi^{n,2}_i&=&\frac{-1}{(k_n-l_n)\Delta_n}\sum_{h=1}^{k_n-l_n+1}R^n_{i+h}\\
	N(p)^n_i&=&\frac{1}{(k_n-l_n)\Delta_n}\Big(\sum_{h=0}^{m(n,p)-1}\zeta(1)^n_{i+b(n,p,h)}+\sum_{h=m(n,p)(p+1)l_n}^{k_n-l_n}\zeta^n_{i+1+h}\Big)\\
	M(p)^n_i&=&\frac{1}{(k_n-l_n)\Delta_n}\sum_{h=0}^{m(n,p)-1}\zeta(p)^n_{i+a(n,p,h)}
\end{eqnarray*}

\subsubsection{bounds on $\|\xi^{n,r}_i\|$}
By Assumption A-$\nu$, (\ref{SA-v}), (\ref{classic}),
\begin{equation}\label{est.xi(0)n}
\left\{\begin{array}{lcl}
	\left\| E\left(\xi^{n,0}_{i}|\F^{(0),n}_{i}\right) \right\| &\le& Kk_n\Delta_n\\
	E\left( \|\xi^{n,0}_{i}\|^q|\F^{(0),n}_{i}\right) &\le& K_q(k_n\Delta_n)^{(q/2)\wedge1},\, q\ge0
\end{array}\right.
\end{equation}
combined with (\ref{tuning}),
\begin{equation}\label{est.xi(1)n}
\left\{\begin{array}{lcl}
	\left\| E\left(\xi^{n,1}_i|\F^{(0),n}_i\right) \right\| &\le& K\Delta_n^{1/2}\\
	E\left(\|\xi^{n,1}_i\|^q|\F^{(0),n}_i\right) &\le& K_q\Delta_n^{[(q/2)\wedge1]/2},\,q\in\mathbb{N}^+
\end{array}\right.
\end{equation}
By Assumption A-$\rr$,
\begin{equation}\label{est.xi(2)n}
\begin{array}{l}
	\left\| E^n_i\big(\xi^{n,2}_i\big) \right\| \le K\Delta_n^{-1}\\
	E^n_i\big(\|\xi^{n,2}_i\|^q\big) \le
	\left\{\begin{array}{l}
		K\,k_n^{-1/2},\, q=1;\\
		K_q\left(k_n^{-q+1} + k_n^{-q}\Delta_n^{-q/2+1}\right), q=2,3,4 
	\end{array}\right.
\end{array}
\end{equation}

\subsubsection{estimates of $\zeta(X,p)^n_i$ \& $\zeta(X,p)'^n_i$}
$\widebar{C}^n_i=(\psi^n)^{-1}C^n_{i\Delta_n,(l_n-1)\Delta_n}$ in view of (\ref{def.Un}),
hence by (\ref{SA-v})
\begin{equation}\label{est.Cbar}
	\|\widebar{C}^n_i\| \le K\Delta_n
\end{equation}
According to (3.1) in the main article we have $\widebar{X}^n_i=-\psi_n^{-1/2}\sum_{h=0}^{l_n-1}(\varphi^n_{h+1}-\varphi^n_h)(X^n_{i+h}-X^n_i)$, then by (\ref{classic})
\begin{equation}\label{est.Xbar}
E\left(\|\widebar{X}^n_i\|^q|\F^{(0),n}_i\right)\le K_q\,\Delta_n^{q/2}
\end{equation}

Adopt the argument for (5.21) in \cite{j09} in the multivariate setting, we have
\begin{equation}\label{est.Wbar}
\begin{array}{lcl}
	E\big(\widebar{W}^n_h\widebar{W}^{n,\T}_{h'}|\F^{(0),n}_i\big) &=& \frac{l_n\Delta_n}{\psi_n}\phi_0\big(\frac{|h'-h|}{l_n}\big)\I + O_p(l_n^{-1/2}\Delta_n)\\
	E\big(\|\widebar{W}^n_h\|^{2m}|\F^{(0),n}_i\big) &=& \Delta_n^m(2m-1)!! + O_p(l_n^{-1}\Delta_n^m),\,\,m\in\mathbb{N}^+
\end{array}
\end{equation}
Let $U^n_i(p)=\sum_{h=i}^{i+pl_n-1}(\sigma^n_i\widebar{W}^n_h)(\sigma^n_i\widebar{W}^n_h)^\T$, $S^n_i(p)=\sum_{h=i}^{i+pl_n-1}\widebar{C}^n_h$, then
\begin{multline}\label{zeta.W.p.2}
	\zeta(W,p)^{n,jk}_i\zeta(W,p)^{n,lm}_i = U^n_i(p)^{jk}U^n_i(p)^{lm} + S^n_i(p)^{jk}S^n_i(p)^{lm}\\ 
	- U^n_i(p)^{jk}S^n_i(p)^{lm} - U^n_i(p)^{lm}S^n_i(p)^{jk}
\end{multline}
By (\ref{sum.phi.quadratic}), (\ref{est.Cbar}), (\ref{est.Wbar}), (\ref{zeta.W.p.2}), and through similar arguments in section 5.3 of \cite{j09} with a modification for multi-dimension, and exploit the connection between $\zeta(W,p)^n_i$ and $\zeta(X,p)^n_i$, in view of (4.2) in the main article, we have the following lemma:
\begin{lem}\label{est.zetaX}
Assume Assumption A-$\nu$, (\ref{SA-v}), $l_n$ satisfies (\ref{tuning}), then
\begin{eqnarray*}
	E\left(\|\zeta(X,p)^n_i\|^4|\F^{(0),n}_i\right)&\le&Kp^4\Delta_n^2\\
	\left\|E\left[\zeta(X,p)^n_i|\F^{(0),n}_i\right]\right\| &\le& Kp\Delta_n\\
	E\left[\zeta(X,p)'^n_i|\F^{(0),n}_i\right] &=& \frac{\theta^2l_n}{\psi_n}(p\Phi_{01}-\Psi_{01})\,c^n_i + p^2\,O_p(\Delta_n^{1/4})\\
	E\left[\zeta(X,p)^{n,jk}_i\zeta(X,p)^{n,lm}_i|\F^{(0),n}_i\right] &=& \frac{2\theta^4}{\psi_n^2}(p\Phi_{00}-\Psi_{00})\,\Sigma(c^n_i)^{jk,lm} + p^2\,O_p(\Delta_n^{5/4})
\end{eqnarray*}
\end{lem}

\subsubsection{estimates of $\zeta(p)^n_i$}
For $i\le h,h'\le i+pl_n-1$, by (\ref{sum.phi.quadratic})
\begin{eqnarray}\label{est.Gamma}
	\Gamma^n_{h,h'} &=& \frac{1}{\psi_nl_n}\phi_1\Big(\frac{|h'-h|}{l_n}\Big)\rr_i+O_p(\Delta_n^{5/4})\\
	\sum_{h=i}^{i+pl_n-2}\,\sum_{h'=h+1}^{i+pl_n-1}\Gamma^{n,jk}_{h,h'}\Gamma^{n,lm}_{h,h'}&=&\frac{1}{\psi_n^2}\left(p\Phi_{11}-\Psi_{11}\right)\gamma^{n,jk}_i\gamma^{n,lm}_i + p^2\,O_p(\Delta_n^{5/4}) \nonumber
\end{eqnarray}

Let $\xi^{n,j_1\cdots j_q}_{h_1\cdots h_q}=\prod_{v=1}^q\big(\widebar{Y}^{n,j_v}_{h_v}-\widebar{X}^{n,j_v}_{h_v}\big)$. By Assumption A-$\gamma$,
\begin{multline}\label{est.Ybar-Xbar}
	E\left[\big(\widebar{Y}^{n,j}_h-\widebar{X}^{n,j}_h\big)^q\big(\widebar{Y}^{n,k}_{h'}-\widebar{X}^{n,k}_{h'}\big)^r|\widetilde{\F}^n_{h\wedge h'-1}\right] =
	\left\{\begin{array}{ll}
		0                    & q+r=1\\
		\Gamma^{n,jk}_{h,h'} & q=r=1\\
		O_p\big(l_n^{-7/2}\big)\1{|h-h'|\le l_n} & q+r=3\\
		O_p(l_n^{-8})        & q+r=8
	\end{array}\right.\\
	E\left[\xi^{n,jklm}_{hhh'h'}|\widetilde{\F}^n_{h\wedge h'}\right] = \Gamma^{n,jk}_h\Gamma^{n,lm}_{h'} + \Gamma^{n,jl}_{h,h'}\Gamma^{n,km}_{h,h'} +  \Gamma^{n,jm}_{h,h'}\Gamma^{n,kl}_{h,h'} + O_p\left(l_n^{-5}\right)
\end{multline}
then by (\ref{est.Xbar}), (\ref{est.Ybar-Xbar})
\begin{multline*}
	E(\zeta^n_h|\widetilde{\F}^n_h) = \widebar{X}^n_h{\widebar{X}^n_h}^\T - \widebar{C}^n_h,\,\,\,
	\|E^n_h\left(\zeta^n_h\right)\| \le K\Delta_n^{3/2},\,\,\,
	E^n_h\left(\|\zeta^n_h\|^4\right) \le K\Delta_n^4 \nonumber\\
	E\big(\zeta^{n,jk}_h\zeta^{n,lm}_{h'}|\widetilde{\F}^n_{h\wedge h'}\big) = \sum_{r=1}^3\vartheta(r)^{n,jk,lm}_{h,h'}\\
	+ \big(\widebar{X}^{n,j}_h\widebar{X}^{n,k}_h - \widebar{C}^{n,jk}_h\big)\big(\widebar{X}^{n,l}_{h'}\widebar{X}^{n,m}_{h'} - \widebar{C}^{n,lm}_{h'}\big)
\end{multline*}
where
\begin{eqnarray*}
	\vartheta(1)^{n,jk,lm}_{h,h'}&=&\widebar{X}^{n,j}_h\widebar{X}^{n,l}_{h'}\xi^{n,km}_{hh'}+\widebar{X}^{n,j}_h\widebar{X}^{n,m}_{h'}\xi^{n,kl}_{hh'}\\
	&&\hspace{30mm} +\widebar{X}^{n,k}_h\widebar{X}^{n,l}_{h'}\xi^{n,jm}_{hh'}+\widebar{X}^{n,k}_h\widebar{X}^{n,m}_{h'}\xi^{n,jl}_{hh'}\\
	\vartheta(2)^{n,jk,lm}_{h,h'}&=&\xi^{n,jklm}_{hhh'h'} - \xi^{n,jk}_{hh}\Gamma^{n,lm}_{h'} - \xi^{n,lm}_{h'h'}\Gamma^{n,jk}_h + \Gamma^{n,jk}_h\Gamma^{n,lm}_{h'}\\
	\vartheta(3)^{n,jk,lm}_{h,h'}&=&\widebar{X}^{n,j}_h\xi^{n,klm}_{hh'h'}+\widebar{X}^{n,k}_h\xi^{n,jlm}_{hh'h'}+\widebar{X}^{n,l}_{h'}\xi^{n,jkm}_{hhh'}+\widebar{X}^{n,m}_{h'}\xi^{n,jkl}_{hhh'}
\end{eqnarray*}

Let $\Upsilon^{n,jk,lm}_{h,h'} = \Theta(\widebar{X}^n_h\widebar{X}^{n,\T}_{h'},\Gamma^n_{h,h'})^{jk,lm}$ in light of (4.2) in the main article, then
\begin{multline*}
	E\left(\zeta^{n,jk}_h\zeta^{n,lm}_{h'}|\widetilde{\F}^n_{h\wedge h'}\right) = \big(\widebar{X}^{n,j}_h\widebar{X}^{n,k}_h - \widebar{C}^{n,jk}_h\big)\big(\widebar{X}^{n,l}_{h'}\widebar{X}^{n,m}_{h'} - \widebar{C}^{n,lm}_{h'}\big)\\
	+ \Upsilon^{n,jk,lm}_{h,h'} + \Sigma(\Gamma^n_{h,h'})^{jk,lm} + O_p\big(l_n^{-5} + (\|\widebar{X}^n_h\|+\|\widebar{X}^n_{h'}\|)l_n^{-7/2}\big)
\end{multline*}

hence 
\begin{multline*}
	\zeta(p)^{n,jk}_i\zeta(p)^{n,lm}_i = \zeta(X,p)^{n,jk}_i\zeta(X,p)^{n,lm}_i\\
	+ \sum_{h=i}^{i+pl_n-2}\sum_{h'=h+1}^{i+pl_n-1}\Big[\Upsilon^{n,jk,lm}_{h,h'} + \Upsilon^{n,lm,jk}_{h,h'} + \Sigma(\Gamma^n_{h,h'})^{jk,lm} + \Sigma(\Gamma^n_{h,h'})^{lm,jk}\Big]\\
	+ \sum_{h=i}^{i+pl_n-1}\Big[\Upsilon^{n,jk,lm}_{h,h'} + \Sigma(\Gamma^n_{h,h'})^{jk,lm}\Big] + p^2\,O_p(\Delta_n^{5/4})
\end{multline*}
then by (\ref{est.Xbar}), (\ref{est.Gamma})
\begin{multline*}
	E\left[\zeta(p)^{n,jk}_i\zeta(p)^{n,lm}_i|\widetilde{\F}^n_i\right] = \zeta(X,p)^{n,jk}_i\zeta(X,p)^{n,lm}_i +\\
	\frac{2}{\psi_nl_n}\Theta(\zeta(X,p)'^n_i,\rr^n_i)^{jk,lm} + \frac{2}{\psi_n^2}(p\Phi_{11}-\Psi_{11})\,\Sigma(\rr^n_i)^{jk,lm} + p^2\,O_p(\Delta_n^{5/4})
\end{multline*}
According to these results and Lemma \ref{est.zetaX}, one can get the following lemma:
\begin{lem}\label{est.zetaY}
Assume Assumption A-$\nu$, A-$\gamma$, (\ref{SA-v}), $l_n$ satisfies (\ref{tuning}), then
\begin{eqnarray*}
	E[\zeta(p)^n_i|\widetilde{\F}^n_i] &=& \zeta(X,p)^n_i \\
	\|E^n_i[\zeta(p)^n_i]\| &\le& Kp\Delta_n\\
	E^n_i(\|\zeta(p)^n_i\|^q) &\le& K_q\,p^{\lfloor q/2\rfloor\vee1}\Delta_n^{q/2},\hspace{2mm} q=1,2,3,4
\end{eqnarray*}
moreover
\begin{equation*}
	\left|E^n_i\left[\zeta(p)^{n,jk}_i\zeta(p)^{n,lm}_i\right] - (p+1)\theta\Delta_n\,\Xi(c^n_i,\gamma^n_i; p)^{jk,lm}\right| \le K\Delta_n^{5/4}
\end{equation*}
where
\begin{equation}\label{def.Xi(p)}
	\Xi(x,z;p) = \frac{p\Phi_{00}-\Psi_{00}}{(p+1)\Phi_{00}}\,\Sigma(x) + \frac{p\Phi_{01}-\Psi_{01}}{(p+1)\Phi_{01}}\,\Theta(x,z) + \frac{p\Phi_{11}-\Psi_{11}}{(p+1)\Phi_{11}}\,\Sigma(z)
\end{equation}
\end{lem}

Let $p\asymp\Delta_n^{-1/12}$, based on Lemma \ref{est.zetaY},
\begin{eqnarray}\label{est.M(p).N(p)}
	\left\|E^n_i\left[M(p)^n_i\right] \right\| &\le& K\Delta_n^{1/2} \nonumber\\
	\left\|E^n_i\left[N(p)^n_i\right] \right\| &\le& K p^{-1}\Delta_n^{1/2} \nonumber\\
	E^n_i\left[\|M(p)^n_i\|^q\right] &\le&
	\left\{\begin{array}{ll}
		K_q\big(k_n\Delta_n^{1/2}\big)^{-q/2}, & q=1,2,4\\
		K\big(k_n\Delta_n^{1/2}\big)^{-2}, & q=3
	\end{array}\right.\\
	E^n_i\left[\|N(p)^n_i\|^q\right] &\le&
	\left\{\begin{array}{ll}
		K p\big(k_n\Delta_n^{1/2}\big)^{-1}, & q=1\\
		K_q p^{-q/2}\big(k_n\Delta_n^{1/2}\big)^{-q/2}, & q=2,4\\
		K p^{-1}\big(k_n\Delta_n^{1/2}\big)^{-2}, & q=3
	\end{array}\right. \nonumber
\end{eqnarray}

\subsubsection{estimates of $\beta^n_i$}
We need to define more variables:
\begin{equation}\label{def.betavars}
\begin{array}{ll}
	\zeta(p)^n_{i,h} = \zeta(p)^n_{i+a(n,p,h)} & A(p)^n_{i+v} = \sum_{h=v}^{v+pl_n-1}(c^n_{i+h}-c^n_i)\Delta_n\\
	D(p)^n_{i+v} = \sum_{h=v}^{v+pl_n-1}D^n_{i+h} & R(p)^n_{i+v} = \sum_{h=v}^{v+pl_n-1}R^n_{i+h}
\end{array}
\end{equation}
we have
\begin{table}[H]
\centering
\caption{Estimates of ingredients under (\ref{tuning})}\label{ingredient.estimates}
\begin{tabular}{c|ll}
	scaling properties & $E(\|\cdot\|^2|\F^n_i)$ & $\|E(\cdot|\F^n_i)\|$\\
	\hline
	$R(p)^n_{i+v}$     & $p\Delta_n^{3/2}$       & $p\Delta_n^{3/2}$\\
	$D(p)^n_{i+v}$     & $p\Delta_n^{3/2}$       & $p\Delta_n$ \\
	$A(p)^n_{i+v}$     & $p^2\Delta_n^2\big(p\Delta_n^{-1/2}+v\big)$ & $p\Delta_n^{3/2}\big(p\Delta_n^{-1/2}+v\big)$\\
	$\zeta(p)^n_{i,h}$ & $p\Delta_n$             & $p\Delta_n$\\
\end{tabular}
\end{table}
Define 
\begin{multline*}
	\alpha(p)^n_{i,h} = - R(p+1)^n_{i+a(n,p,h)} + D(p+1)^n_{i+a(n,p,h)} 
	 + A(p+1)^n_{i+a(n,p,h)} + \zeta(p+1)^n_{i,h}
\end{multline*}
By Table \ref{ingredient.estimates} and Cauchy-Schwarz inequality,
\begin{multline}\label{a(p)2-zeta(p)2}
	E^n_i\left(\left|\alpha(p)^{n,jk}_{i,h}\alpha(p)^{n,lm}_{i,h}-\zeta(p+1)^{n,jk}_{i,h}\zeta(p+1)^{n,lm}_{i,h}\right|\right)\\
\le K\big(p^2\Delta_n^{5/4}+p^{3/2}\Delta_n^{3/2}v^{1/2}\big)
\end{multline}

Given $j,k,l,m=1,\cdots,d$, in the light of (\ref{errors.chat*}) and (\ref{def.betavars}), by Table \ref{ingredient.estimates} we have
\begin{equation*}
	\left| E^n_i\big(\beta^{n,jk}_i\beta^{n,lm}_i\big) - (k_n\Delta_n^{1/2})^{-1}\Xi(c^n_i,\rr_i)^{jk,lm} \right| = \sum_{r=1}^5\mu^{n,r}_i + p\,O_p((k_n^2\Delta_n)^{-1})
\end{equation*}
where
\begin{eqnarray*}
	\mu^{n,1}_i &=&
	\frac{1}{(k_n-l_n)^2\Delta_n^2}\sum_{h=0}^{m(n,p)-1}E^n_i\left(\left|\alpha(p)^{n,jk}_{i,h}\alpha(p)^{n,lm}_{i,h}
	- \zeta(p+1)^{n,jk}_{i,h}\zeta(p+1)^{n,lm}_{i,h}\right|\right)\\
	\mu^{n,2}_i &=&
	\frac{1}{(k_n-l_n)^2\Delta_n^2}\sum_{h=0}^{m(n,p)-2}\sum_{h'=h+1}^{m(n,p)-1}\left|E^n_i\left[\alpha(p)^{n,jk}_{i,h}\alpha(p)^{n,lm}_{i,h'}
	+\alpha(p)^{n,lm}_{i,h}\alpha(p)^{n,jk}_{i,h'}\right]\right|\\
	\mu^{n,3}_i &=& \frac{1}{(k_n-l_n)^2\Delta_n^2}\sum_{h=0}^{m(n,p)-1}E^n_i\left(\left|\zeta(p+1)^{n,jk}_{i,h}\zeta(p+1)^{n,lm}_{i,h}\right.\right.\\
	&&\hspace{36mm} \left.\left. -(p+2)\theta\Delta_n\,\Xi\big(c^n_{i+a(n,p,h)},\rr^n_{i+a(n,p,h)}; p+1\big)^{jk,lm}\right|\right)\\
	\mu^{n,4}_i &=& \frac{(p+2)\theta}{(k_n-l_n)^2\Delta_n}\sum_{h=0}^{m(n,p)-1}\left|E^n_i\left[\Xi\big(c^n_{i+a(n,p,h)},\rr^n_{i+a(n,p,h)}; p+1\big)^{jk,lm}\right.\right. \\
	&&\hspace{76mm} \left.\left. - \Xi(c^n_i,\rr^n_i; p+1)^{jk,lm}\right]\right|
\end{eqnarray*}
\begin{eqnarray*}
	\mu^{n,5}_i &=& \frac{1}{(k_n-l_n)\Delta_n^{1/2}}\left|\frac{(p+2)\theta}{(k_n-l_n)\Delta_n^{1/2}}\left\lfloor\frac{k_n}{(p+1)l_n}\right\rfloor-\frac{(k_n-l_n)}{k_n}\right|\cdot\left|\Xi(c^n_i,\rr^n_i; p+1)^{jk,lm}\right|\\
	&&\hspace{46mm} + \frac{1}{k_n\Delta_n^{1/2}}\left|\Xi(c^n_i, \rr^n_i; p+1)^{jk,lm} - \Xi(c^n_i,\rr^n_i)^{jk,lm}\right|
\end{eqnarray*}
Use Table \ref{ingredient.estimates} and (\ref{a(p)2-zeta(p)2}) to get bounds on $\mu^{n,r}_i,\,r=1,2,3,4,5$; combine (\ref{errors.chat*}), (\ref{est.xi(0)n}), (\ref{est.xi(1)n}), (\ref{est.xi(2)n}), (\ref{est.M(p).N(p)}), we get the following lemma:

\begin{lem}\label{est.beta} 
Assume (\ref{SA-v}), (\ref{tuning}) and Assumption A-$\nu$, A-$\gamma$, given $p\in\mathbb{N}^+$,
\begin{eqnarray*}
	\|E^n_i(\beta^n_i)\| &\le& K k_n\Delta_n\\
	E^n_i(\|\beta^n_i\|^q) &\le& 
	\left\{\begin{array}{ll}
		K_q\Big[(k_n\Delta_n)^{(q/2)\wedge1} + (k_n\Delta_n^{1/2})^{-q/2}\Big],& q = 1,2,4\\
		Kk_n\Delta_n, & q = 3
	\end{array}\right.
\end{eqnarray*}
additionally
\begin{equation*}
	\left|E^n_i\big(\beta^{n,jk}_i\beta^{n,lm}_i\big) - (k_n\Delta_n)^{-1/2}\Xi(c^n_i,\gamma^n_i)^{jk,lm}\right| \le K\big[k_n\Delta_n + p^{-1}(k_n\Delta_n^{1/2})^{-1}\big]
\end{equation*}
\end{lem}

\subsection{Structure}
Define
\begin{eqnarray*}
	\lambda(x,z) &=& \sum^d_{j,k,l,m=1}\partial^2_{jk,lm}g(x)\times\Xi(x,z)^{jk,lm}\\
	\eta^n_i &=& \lambda(\widehat{c}^n_i,\widehat{\gamma}^n_i)-\lambda(\widehat{c}^{*n}_i,\widehat{\gamma}^n_i)
\end{eqnarray*}
As $\Delta_n^{-1/4}|a^n_t-1|<k_n\Delta_n^{3/4}\to0$, letting $a^n_t=1$ 
doesn't not affect the asymptotic analysis. By Cram\'er-Wold theorem, we can suppose $g$ is $\R$-valued. We have
\begin{equation}\label{decomp}
	\Delta_n^{-1/4}\big[\widehat{S}(g)^n-S(g)\big] = \widebar{S}^{n,0} + \widebar{S}^{n,1} + \widebar{S}(p)^{n,2} + \widebar{S}^{n,3} + \widebar{S}(p)^{n,4}
\end{equation}
where
\begin{eqnarray*}
	\widebar{S}^{n,0}_t&=&\Delta_n^{-1/4}\Bigg[\sum_{i=0}^{N^n_t-1}\int_{ik_n\Delta_n}^{(i+1)k_n\Delta_n}g(c^n_{ik_n})-g(c_s)\ds s-\int_{N^n_tk_n\Delta_n}^tg(c_s)\ds s\Bigg]\\
	\widebar{S}^{n,1}_t&=&k_n\Delta_n^{3/4}\sum_{i=0}^{N^n_t-1}\Big[g(\widehat{c}^n_{ik_n})-g(\widehat{c}^{*n}_{ik_n}) - (2k_n\Delta_n^{1/2})^{-1}\eta^n_{ik_n}\Big]\\
	\widebar{S}(p)^{n,2}_t&=&k_n\Delta_n^{3/4}\sum_{i=0}^{N^n_t-1}\sum^d_{j,k=1}\partial_{jk}g(c^n_{ik_n})\Bigg[\sum_{r=0}^2\xi^{n,r,jk}_{ik_n}+N(p)^{n,jk}_{ik_n}\Bigg]
\end{eqnarray*}
\begin{eqnarray*}
	\widebar{S}^{n,3}_t&=&k_n\Delta_n^{3/4}\sum_{i=0}^{N^n_t-1}\Big[g(\widehat{c}^{*n}_{ik_n})-g(c^n_{ik_n})-\sum^d_{j,k=1}\partial_{jk}g(c^n_{ik_n}) \beta^{n,jk}_{ik_n}\\
	&&\hspace{66mm} - (2k_n\Delta_n^{1/2})^{-1}\lambda(\widehat{c}^{*n}_i,\widehat{\rr}^n_i)\Big]\\
	\widebar{S}(p)^{n,4}_t&=&k_n\Delta_n^{3/4}\sum_{i=0}^{N^n_t-1}\sum^d_{j,k=1}\partial_{jk}g(c^n_{ik_n})\times M(p)^{n,jk}_{ik_n}
\end{eqnarray*}

\subsection{Asymptotic negligibility}\label{apdx:asymneg}
First of all, we need to get bounds on $\partial^r g(\widehat{c}^n_i),\,r\le3$, where $\partial^rg$ denotes the $r$-th order partial derivatives. Let $\widebar{c}^n_i = (k_n-l_n)^{-1}\sum_{h=1}^{k_n-l_n+1}c^n_i$ and $I^n_t=\{0,1,\cdots,N^n_t-1\}$, note that $|I^n_t|\asymp (k_n\Delta_n)^{-1}$, according to Lemma \ref{est.chat-chat*}, there is a sequence $a_n\to0$ such that
\begin{equation*}
	\E\Big(\sup_{i\in I_n}\|\widehat{c}^n_i-\widehat{c}^{*n}_i\|\Big) \le K\left(a_n\Delta_n^{\kappa-1/2-(\rho-1/4)\nu - (1-2\rho)} + \Delta_n^{\kappa-1/2}\right)
\end{equation*}
Note $\widehat{c}^{*n}_i-\widebar{c}^{n}_i = \xi^{n,1}_i + \xi^{n,2}_i + N(p)^n_i + M(p)^n_i$, by (\ref{est.xi(1)n}), (\ref{est.xi(2)n}), (\ref{est.M(p).N(p)}) and $\kappa<3/4$, $E^n_i\big(\|\widehat{c}^{*n}_i-\widebar{c}^{n}_i\|^4\big)\le K\Delta_n^{2\kappa-1}$, so
\begin{equation*}
	\E\Big(\sup_{i\in I_n}\|\widehat{c}^{*n}_i-\widebar{c}^n_i\|^4\Big) \le K\Delta_n^{3\kappa-2}
\end{equation*}
hence by (\ref{tuning}) and Markov's inequality
\begin{equation}\label{chat-c.uniform}
	\sup_{i\in I_n} \|\widehat{c}^n_i-\widebar{c}^n_i\| = o_p(1)
\end{equation}
According to (\ref{SA-v}) and convexity, $\widebar{c}^n_i\in\mathcal{S}$. By (\ref{chat-c.uniform}), $\widehat{c}^n_i\in\mathcal{S}^\epsilon$ when $n$ is sufficiently large. Therefore by (3.8) in the main article, for asymptotic analysis we can assume
\begin{equation}\label{g.cond.strong}
	\|\partial^rg(\widehat{c}^n_i)\| \le K, \hspace{5mm} \forall r=0,1,2,3,\, \forall i\in I^n_t
\end{equation}

Through an almost identical argument for Lemma 4.4 in \cite{jr13},
\begin{equation}\label{v0}
	\widebar{S}^{n,0}\overset{u.c.p.}{\longrightarrow}0
\end{equation}

Define function $g_n$ on $\mathcal{S}^+_d\times\mathcal{S}^+_d$ as $g_n(x,z) = g(x) - (2k_n\Delta_n^{1/2})^{-1}\lambda(x,z)$. Based on (\ref{g.cond.strong}),
\begin{multline*}
	\|g_n(x,z)-g_n(y,z)\| \le K\|x-y\|\\
	+ K(k_n\Delta_n^{1/2})^{-1}\|x-y\|\left(\|x\|^2+\|z\|^2 + \|x-y\|^2+\|z\|\|x-y\|\right)
\end{multline*}
so $\|g_n(x,z)-g_n(y,z)\| \le K\|x-y\|$ when $n$ is sufficiently large. By Lemma \ref{est.chat-chat*}
\begin{multline*}
	\E\left(\sup_{u\in[0,t]}\big\|\widebar{S}^{n,1}_u\big\|\right) \le k_n\Delta_n^{3/4}\sum_{i=0}^{N^n_t-1}\left\|g_n(\widehat{c}^n_{ik_n},\widehat{\rr}^n_{ik_n})-g_n(\widehat{c}^{*n}_{ik_n},\widehat{\rr}^n_{ik_n})\right\|\\
	\le Kt\left(a_n\Delta_n^{1/4-(\rho-1/4)\nu-(1-2\rho)} + \Delta_n^{1/4}\right)
\end{multline*}
Since $\rho>\frac{3-\nu}{4(2-\nu)}$, $1/4-(\rho-1/4)\nu-(1-2\rho)>0$, we have the following lemma:
\begin{lem}\label{lemma.v1}
	Assume Assumption A-$\nu$, A-$\rr$, (\ref{tuning}), (\ref{g.cond.strong}) then
	$$\widebar{S}^{n,1}\overset{u.c.p.}{\longrightarrow}0$$
\end{lem}

Given $e^n_i\in\R^{d\times d}$, consider the process
\begin{equation*}
	\widebar{S}^{*n}_t = k_n\Delta_n^{3/4} \sum_{i=0}^{N^n_t-1} \sum^d_{j,k=1}\partial_{jk}g(c^n_{ik_n})\times e^{n,jk}_{ik_n}
\end{equation*}
suppose $e^n_i$ satisfies
\begin{equation}\label{xi.cond}
	\left\{\begin{array}{lcl}
	\left\| E^n_i\left(e^n_{i}\right) \right\| &\le& K\Delta_n^{1/4} a_n \\
	E^n_i\left( \|e^n_{i}\|^2\right) &\le& K(k_n\Delta_n^{1/2})^{-1} b_n
\end{array}\right.
\end{equation}
where $a_n,b_n\to0$. Since $\partial g$ is bounded by (\ref{SA-v}),
\begin{multline*}
	\E\left(\sup_{u\in[0,t]} \big\|\widebar{S}^{*n}_u\big\|\right) \le Kk_n\Delta_n^{3/4}\sum_{i=0}^{N^n_t-1}\E\left(\left\|E^n_{ik_n}\left(e^n_{ik_n}\right)\right\|\right) \\ 
	+ Kk_n\Delta_n^{3/4}\E\left(\sup_{s\in[0,t]}\left\|\sum_{i=0}^{N^n_s-1}\left[ e^n_{ik_n} - E^n_{ik_n}\left(e^n_{ik_n}\right) \right]\right\|\right)
\end{multline*}
by Lemma \ref{Doob.max.ineq},
\begin{equation*}
	\E\left(\sup_{s\in[0,t]}\left\|\sum_{i=0}^{N^n_s-1}\left[ e^n_{ik_n} - E^n_{ik_n}\left(e^n_{ik_n}\right) \right]\right\|\right)
	\le K\left(\sum_{i=0}^{N^n_t-1}\E\left(\|e^n_{ik_n}\|^2\right)\right)^{1/2}
\end{equation*}
note that $k_n\Delta_nN^n_t\asymp t$, we have
\begin{equation*}
	\E\left(\sup_{u\in[0,t]} \big\|\widebar{S}^{*n}_u\big\|\right) \le K\left(ta_n + \sqrt{tb_n}\right)\to0
\end{equation*}
To show the asymptotic negligibility of $\widebar{S}(p)^{n,2}$, we need to show $e_i$ satisfies (\ref{xi.cond}) in each of the following 4 cases:
\begin{itemize}
	\item[\textcircled{1}] when $e^n_i=\xi^{n,0}_i$, 
	by (\ref{est.xi(0)n}), $a_n=k_n\Delta_n^{3/4},\,b_n=(k_n\Delta_n^{3/4})^2$;
	\item[\textcircled{2}] when $e^n_i=\xi^{n,1}_i$,
	by (\ref{est.xi(1)n}), $a_n=\Delta_n^{1/4}$, $b_n=k_n\Delta_n$;
	\item[\textcircled{3}] when $e^n_i=\xi^{n,2}_i$,
	by (\ref{est.xi(2)n}), $a_n=\Delta_n^{3/4}$, $b_n=\Delta_n^{1/2}$;
	\item[\textcircled{4}] when $e^n_i=N(p)^n_i$, by (\ref{est.M(p).N(p)}), $a_n=p^{-1}\Delta_n^{1/4}$, $b_n=p^{-1}$.
\end{itemize}
Hence we have the following lemma:
\begin{lem}\label{lemma.v2}
	Assume Assumption A-$\nu$, A-$\gamma$, (\ref{tuning}), (\ref{g.cond.strong}), and let $p\asymp\Delta_n^{-1/12}$, then
	$$\widebar{S}(p)^{n,2}\overset{u.c.p.}{\longrightarrow}0$$
\end{lem}

Let $\chi^n_i = \widehat{\rr}^n_i - \rr^n_i$. By (\ref{classic}), the choice of $m_n$ and Jensen's inequality
\begin{equation}\label{est.chi}
	E^n_i(\|\chi^n_i\|^q) \le K_q\Delta_n^{q/4},\, q=1,2
\end{equation}
Let $\iota^n_i=\lambda(\widehat{c}^{*n}_i,\widehat{\gamma}^n_i)-\lambda(c^n_i,\gamma^n_i)$, then by (4.2) and (4.3) in the main article
\begin{equation*}
	\|\iota^n_i\| \le K\left(\|\beta^n_i\| + \|\chi^n_i\| + \|\beta^n_i\|^2 + \|\beta^n_i\|\|\chi^n_i\| + \|\chi^n_i\|^2\right)
\end{equation*}
hence by Lemma \ref{est.beta}, (\ref{tuning}), (\ref{est.chi}),
\begin{equation}\label{est.eta}
	E^n_i(\|\iota^n_i\|) \le K k_n^{-1/2}\Delta_n^{-1/4}
\end{equation}

We can rewrite $\widebar{S}^{n,3}$ as
\[\widebar{S}^{n,3}=G^n+H^n\]
where
\begin{eqnarray*}
	G^n_t&=&k_n\Delta_n^{3/4}\sum_{i=0}^{N^n_t-1}\left[s^n_{ik_n}+\frac{1}{2k_n\Delta_n^{1/2}}\iota^n_{ik_n}+E^n_{ik_n}\left(v^n_{ik_n}\right)\right] \\
	H^n_t&=&k_n\Delta_n^{3/4}\sum_{i=0}^{N^n_t-1}\left[v^n_{ik_n}-E^n_{ik_n}\left(v^n_{ik_n}\right)\right]
\end{eqnarray*}
\begin{eqnarray*}
	s^n_i &=& g(c^n_i+\beta^n_i)-g(c^n_i)-\sum^d_{j,k=1}\partial_{jk}g(c^n_i)\beta_i^{n,jk} - \frac{1}{2}\sum^d_{j,k,l,m=1}\partial^2_{jk,lm}g(c^n_i)\beta_i^{n,jk}\beta_i^{n,lm}\\
	%
	%
	v^n_i&=&\frac{1}{2}\sum^d_{j,k,l,m=1}\partial^2_{jk,lm}g(c^n_i)\left[\beta_i^{n,jk}\beta_i^{n,lm} - (k_n\Delta_n^{1/2})^{-1}\Xi(c^n_i,\rr^n_i)^{jk,lm}\right]
\end{eqnarray*}

By Lemma \ref{est.beta}, (\ref{g.cond.strong}), (\ref{est.eta}), if we let $p\asymp\Delta_n^{-12}$, 
\begin{equation}\label{est.G}
	\E\left(\sup_{s\in[0,t]}\|G^n_s\|\right) \le Kt\left[k_n\Delta_n^{3/4} + (k_n\Delta_n^{2/3})^{-1}\right]
\end{equation}
and
\begin{multline*}
	E^n_i(\|v^n_i\|^2) \le K\sum^d_{j,k,l,m=1} E^n_i\left(\left|\beta_i^{n,jk}\beta_i^{n,lm} - \frac{1}{k_n\Delta_n^{1/2}}\Xi(c^n_i,\rr^n_i)^{jk,lm}\right|^2\right)\\
	\le K\left[k_n\Delta_n + (k_n\Delta_n^{1/2})^{-2}\right]
\end{multline*}
then Lemma \ref{Doob.max.ineq} implies
\begin{multline}\label{est.H}
	\E\left(\sup_{s\in[0,t]}\|H^n_s\|\right)\le Kk_n\Delta_n^{3/4}\left(\sum_{i=0}^{N^n_t-1}\E(\|v^n_{ik_n}\|^2)\right)^{1/2}\\
	\le K\sqrt{t}\left[k_n\Delta_n^{3/4} + (k_n\Delta_n^{1/2})^{-1/2}\right] 
\end{multline}

According to (\ref{est.G}) and (\ref{est.H}), we have the following lemma:
\begin{lem}\label{lemma.v3}
	Assume Assumption A-$\nu$, A-$\gamma$, (\ref{tuning}), (\ref{g.cond.strong}), then
	$$\widebar{S}^{n,3}\overset{u.c.p.}{\longrightarrow}0$$
\end{lem}

\subsection{Stable convergence in law to a continuous It\^o semimartingale}
Recall (\ref{def.betavars}), let
\begin{equation}\label{def.Lambda}
	\Lambda(p)^n_{i,h} = \partial g(c^n_{ik_n})\odot\zeta(p)^n_{ik_n,h}	
\end{equation}
where $\odot$ denotes Hadamard product, i.e., $\Lambda(p)^{n,jk}_{i,h}=\partial_{jk}g(c^n_{ik_n})\odot\zeta(p)^{n,jk}_{ik_n,h}$. We can write 
\begin{equation*}
	\widebar{S}(p)^{n,4}_t=\frac{k_n}{k_n-l_n}\sum^d_{j,k=1}\Delta_n^{-1/4}\sum^{N^n_t-1}_{i=0}\sum_{h=0}^{m(n,p)-1} \Lambda(p)^{n,jk}_{i,h}
\end{equation*}

Let $\mathcal{H}(p)^n_{i,h}=\F^n_{ik_n+a(n,p,h)}$, by Lemma \ref{est.zetaY},
\begin{equation*}
	\Delta_n^{-1/2}\sum^{N^n_t-1}_{i=0}\sum_{h=0}^{m(n,p)-1}\left\|E\big[\zeta(p)^n_{ik_n,h}|\mathcal{H}(p)^n_{i,h}\big]\right\|^2 \le Ktp\Delta_n
\end{equation*}

Let $N$ be a bounded martingale orthogonal to $W$ or $N=W^l$ for some $l=1,\cdots,d'$, and $\Delta N(p)^n_{i,h}=N^n_{ik_n+b(n,p,h)}-N^n_{ik_n+a(n,p,h)}$. The following four statements about convergence in probability for any indices $j,k,l,m$ can verify the conditions of theorem IX.7.28 in \cite{js03}:
\begin{eqnarray}
	\Delta_n^{-1/4}\sum_{i=0}^{N^n_t-1}\sum_{h=0}^{m(n,p)-1}\left\|E\left[\Lambda(p)^n_{i,h}|\mathcal{H}(p)^n_{i,h}\right]\right\|&\overset{\mathbb{P}}{\longrightarrow}&0 \label{S4.c1}\\
	\Delta_n^{-1}\sum_{i=0}^{N^n_t-1}\sum_{h=0}^{m(n,p)-1}E\left[\left\|\Lambda(p)^n_{i,h}\right\|^4|\mathcal{H}(p)^n_{i,h}\right]&\overset{\mathbb{P}}{\longrightarrow}&0 \label{S4.c2}\\
	\Delta_n^{-1/4}\sum_{i=0}^{N^n_t-1}\sum_{h=0}^{m(n,p)-1}\left\|E\big[\Lambda(p)^n_{i,h}\Delta N(p)^n_{i,h}|\mathcal{H}(p)^n_{i,h}\big]\right\| &\overset{\mathbb{P}}{\longrightarrow}&0 \label{S4.c3}
\end{eqnarray}
\begin{multline}\label{S4.c4}
	\Delta_n^{-1/2}\sum_{i=0}^{N^n_t-1}\sum_{h=0}^{m(n,p)-1} E\left[\Lambda(p)^{n,jk}_{i,h}\,\Lambda(p)^{n,lm}_{i,h}|\mathcal{H}(p)^n_{i,h}\right]\\
	\overset{\mathbb{P}}{\longrightarrow} \int_0^t\partial_{jk}g(c_s)\partial_{lm}g(c_s)^\T\,\Xi(c_s, \gamma_s; p)^{jk,lm}\,\mathrm{d}s
\end{multline}

Under (\ref{g.cond.strong}), one can verify (\ref{S4.c1}), (\ref{S4.c2}) by the second and third claims of Lemma \ref{est.zetaY}, respectively. The same argument as that for (5.58) in \cite{j09} leads to (\ref{S4.c3}). By the last claim of Lemma \ref{est.zetaY}, the left-hand side of (\ref{S4.c4}) equals
\begin{multline*}
	\sum_{i=0}^{N^n_t-1}\sum_{h=0}^{m(n,p)-1}\partial_{jk}g(c^n_{ik_n})\partial_{lm}g(c^n_{ik_n})^\T\times\\ \Xi\big(c^n_{ik_n+a(n,p,h)}, \rr^n_{ik_n+a(n,p,h)}; p\big)^{jk,lm}(p+1)l_n\Delta_n + tp\,O_p(\Delta_n^{1/4})
\end{multline*}
then (\ref{S4.c4}) is verified by Riemann summation. By Theorem IX.7.28 in \cite{js03} we have the following lemma:

\begin{lem}\label{lemma.v4}
	Assume Assumption A-$\nu$, A-$\gamma$, (\ref{tuning}), (\ref{g.cond.strong}), then for $\forall p\in\mathbb{N}^+$,
	\[\widebar{S}(p)^{n,4}\overset{\mathcal{L}-s(f)}{\longrightarrow}Z(p)\]
	where $Z(p)$ is a process defined on an extension of the space  $\left(\Omega,\mathcal{F},(\mathcal{F}_t),\mathbb{P}\right)$,
	such that conditioning on $\mathcal{F}$ it is a  mean-0 continuous It\^o martingale with variance
	\[\widetilde{E}[Z(p)Z(p)^\T|\F]=\int_0^t\sum^d_{j,k,l,m=1}\partial_{jk}g(c_s)\partial_{lm}g(c_s)^\T\,\Xi(c_s, \gamma_s; p)^{jk,lm}\ds s\]
	where $\widetilde{E}$ is the conditional expectation operator on the extended probability space and $\Xi(x,z;p)$ is defined in (\ref{def.Xi(p)}).
\end{lem}

By (\ref{v0}), Lemma \ref{lemma.v1}, \ref{lemma.v2}, \ref{lemma.v3}, \ref{lemma.v4}, and $\Xi(x,z;p)\to\Xi(x,z)$ as $p\to\infty$, we arrive at the asymptotic result in Theorem 1.

\section{Derivation for Theorem 2}\label{apdx:thm2}
As a counterpart to (\ref{est.Y.bars.hats}), we have under (\ref{tuning.psd})
\begin{equation}\label{est.Y.bars.psd}
\left\{\begin{array}{lcl}
	E^n_i\left(\big\|\widebar{Y}^{*n}_i\big\|^q\right) &\le& K_q\Delta_n^{q/2},\; \forall q\le8\\
	E^n_i\left(\big\|\widebar{Y}^n_i\big\|^q\right) &\le& K_q\Delta_n^{(q/2)\wedge[1/2-\delta+(1/2+\delta)q/2]},\; \forall q\le8\\
	E^n_i\Big[\Big(\frac{\|\widebar{J}^n_i\|}{\Delta_n^w}\wedge1\Big)^q\Big] &\le& K_q\Delta_n^{[1/2-\delta-(w-1/4-\delta/2)\nu]\times[1\wedge(q/\nu)]}\,a_n,\; \forall q\le8
\end{array}\right.
\end{equation}
for some $a_n\to0$. It follows that
\begin{equation*}
	\big\|\widebar{Y}^n_i\cdot\widebar{Y}^{n,\mathrm{T}}_i\1{\|\widebar{Y}^n_i\|\le\nu_n} - \widebar{Y}^{*n}_i\cdot\widebar{Y}^{*n,\mathrm{T}}_i\big\| \le \eta^{n,1}_i + \eta^{n,3}_i
\end{equation*}
where $\eta^{n,1}_i$ and $\eta^{n,3}_i$ are defined in (\ref{def.etas}).

Through a similar argument to that in Section \ref{apdx:jmp}, for $1\le q\le4$ there is a sequence $a_n\to0$ such that under (\ref{tuning.psd}),
\begin{eqnarray*}
	E^n_i\big[(\eta^{n,1}_i)^q\big] &\le& K_q\big[\Delta_n^{2q} + a_n\Delta_n^{2q\rho+1/2-\delta-(\rho-1/4-\delta/2)\nu}\big] \\
	E^n_i\big[(\eta^{n,3}_i)^q\big] &\le& K_{q,q'}\Delta_n^{q+q'(1-2\rho)/2},\; q'>1/(1-2\rho)
\end{eqnarray*}
Note $\|\widetilde{c}^n_i-\widetilde{c}^{*n}_i\| \le [(k_n-l_n)\Delta_n]^{-1}\sum_{h=1}^{k_n-l_n+1}\big[\eta^{n,1}_{i+h}+\eta^{n,3}_{i+h}\big]$, then
\begin{equation}\label{est.chat-chat*.psd}
	E^n_i\left(\|\widetilde{c}^n_i-\widetilde{c}^{*n}_i\|^q\right) \le K_q\left(a_n\Delta_n^{1/2-\delta-(\rho-1/4-\delta/2)\nu-(1-2\rho)q} + \Delta_n^q\right)
\end{equation}

Define the following quantities
\begin{eqnarray*}
	\xi'^{n,2}_i &=& \frac{1}{(k_n-l_n)\Delta_n}\sum_{h=1}^{k_n-l_n+1}\Gamma^n_{i+h} \\
	\beta'^n_i   &=& \xi^{n,0}_i + \xi^{n,1}_i + \xi'^{n,2}_i + N(p)^n_i + M(p)^n_i
\end{eqnarray*}
where $\xi^{n,0}_i,\, \xi^{n,1}_i,\, N(p)^n_i,\, M(p)^n_i$ are defined in (\ref{errors.chat*}). Note $\widetilde{c}^{*n}_i-c^n_i=\beta'^n_i$.

Under the tuning (\ref{tuning.psd}),
\begin{equation}\label{est.xi(1)n.psd}
\left\{\begin{array}{lcl}
	\left\| E\left(\xi^{n,1}_i|\F^{(0),n}_i\right) \right\| &\le& K\Delta_n^{1/2-\delta}\\
	E\left(\|\xi^{n,1}_i\|^q|\F^{(0),n}_i\right) &\le& K_q\,\Delta_n^{[(q/2)\wedge1](1/2-\delta)},\; q\in\mathbb{N}^+
\end{array}\right.
\end{equation}
By (3.1) and Assumption A-$\rr$ in the main article,
\begin{equation}\label{est.xi(2)n.psd}
\begin{array}{l}
	\left\| E^n_i\big(\xi'^{n,2}_i\big) \right\| \le K\Delta_n^{2\delta}\\
	E^n_i\big(\|\xi'^{n,2}_i\|^q\big) \le K_q\,\Delta_n^{2\delta q}
\end{array}
\end{equation}

Based on the same derivations in Section \ref{apdx:contin} that lead to Lemma \ref{est.zetaY}, we know
\begin{eqnarray}\label{est.N(p).psd}
	\left\|E^n_i\left[N(p)^n_i\right] \right\| &\le& K\big(p^{-1}\Delta_n^{1/2-\delta} + pk_n^{-1}\Delta_n^{-2\delta}\big)\nonumber\\
	E^n_i\left[\|N(p)^n_i\|^q\right] &\le&
	\left\{\begin{array}{ll}
		K_q\,p^{-q/2}\big(k_n\Delta_n^{1/2+\delta}\big)^{-q/2}, & q=1,2,4\\
		K p^{-1}\big(k_n\Delta_n^{1/2+\delta}\big)^{-2}, & q=3
	\end{array}\right.
\end{eqnarray}
and
\begin{eqnarray}\label{est.M(p).psd}
	\left\|E^n_i\left[M(p)^n_i\right] \right\| &\le& K\Delta_n^{1/2-\delta} \nonumber\\
	E^n_i\left[\|M(p)^n_i\|^q\right] &\le&
	\left\{\begin{array}{ll}
		K_q\big(k_n\Delta_n^{1/2+\delta}\big)^{-q/2}, & q=1,2,4\\
		K\big(k_n\Delta_n^{1/2+\delta}\big)^{-2}, & q=3
	\end{array}\right.
\end{eqnarray}

Since $\widetilde{c}^{*n}_i-\widebar{c}^{n}_i = \xi^{n,1}_i + \xi'^{n,2}_i + N(p)^n_i + M(p)^n_i$, based on (\ref{est.xi(1)n.psd}), (\ref{est.xi(2)n.psd}), (\ref{est.N(p).psd}), (\ref{est.M(p).psd}),
\begin{equation*}
\E\Big(\sup_{i\in I_n}\|\widetilde{c}^{*n}_i-\widebar{c}^n_i\|^4\Big) \le K\Delta_n^{3\kappa-2-2\delta}
\end{equation*}
then use $\|\widetilde{c}^n_i-\widebar{c}^{n}_i\|\le\|\widetilde{c}^n_i-\widetilde{c}^{*n}_i\|+\|\widetilde{c}^{*n}_i-\widebar{c}^{n}_i\|$ and (\ref{est.chat-chat*.psd}) with $q=1$, under (\ref{tuning.psd})
\begin{equation*}
\sup_{i\in I_n} \|\widetilde{c}^n_i-\widebar{c}^n_i\| = o_p(1)
\end{equation*}
Therefore, just as explained at the beginning of Section \ref{apdx:asymneg}, when the p.s.d. plug-in $\widetilde{c}^n_i$ and the tuning (\ref{tuning.psd}) are used, we can still assume (\ref{g.cond.strong}).

We know by combining (\ref{est.xi(0)n}), (\ref{est.xi(1)n.psd}), (\ref{est.xi(2)n.psd}), (\ref{est.N(p).psd}) and (\ref{est.M(p).psd}), 
\begin{eqnarray}\label{est.beta.psd}
	\left\|E^n_i\left(\beta'^n_i\right) \right\| &\le& K\big(k_n\Delta_n + \Delta_n^{2\delta}\big) \nonumber\\
	E^n_i\left(\|\beta'^n_i\|^q\right) &\le&
\left\{\begin{array}{ll}
	K_q \big[ (k_n\Delta_n)^{(q/2)\wedge1} + \Delta_n^{[(q/2)\wedge1](1/2-\delta)} & \\
	\hspace{30mm} + \Delta^{2\delta q} + \big(k_n\Delta_n^{1/2+\delta}\big)^{-q/2} \big], & q=1,2,4 \\
	K \big[ k_n\Delta_n + \Delta_n^{1/2-\delta} + \big(k_n\Delta_n^{1/2+\delta}\big)^{-2} \big], & q=3
\end{array}\right.
\end{eqnarray}

Now we study the quadratic form of $\beta'^n_i$. Recall the definitions in (\ref{def.betavars}) and the estimates in Table \ref{ingredient.estimates}. Let $R(p)'^n_{i+v} = \sum_{h=v}^{v+pl_n-1}\Gamma^n_{i+h}$. Analogously when $l_n\asymp\Delta_n^{-(1/2+\delta)}$ we have the following estimates
\begin{table}[H]
	\centering
	\caption{Estimates of ingredients under (\ref{tuning.psd})}\label{ingredient.estimates.psd}
	\begin{tabular}{c|ll}
		scaling properties & $E(\|\cdot\|^2|\F^n_i)$   & $\|E(\cdot|\F^n_i)\|$\\
		\hline
		$R(p)'^n_{i+v}$    & $p^2\Delta_n^{1+2\delta}$ & $p\Delta_n^{1/2+\delta}$\\
		$D(p)^n_{i+v}$     & $p\Delta_n^{3/2-3\delta}$ & $p\Delta_n^{1-2\delta}$ \\
		$A(p)^n_{i+v}$     & $p^2\Delta_n^{2-2\delta}\big(p\Delta_n^{-1/2-\delta}+v\big)$ & $p\Delta_n^{3/2-\delta}\big(p\Delta_n^{-1/2-\delta}+v\big)$\\
		$\zeta(p)^n_{i,h}$ & $p\Delta_n^{1-2\delta}$   & $p^{3/2}\Delta_n^{1-2\delta}$\\
	\end{tabular}
\end{table}
Let
\begin{multline*}
	\alpha(p)'^n_{i,h} = R(p+1)'^n_{i+a(n,p,h)} + D(p+1)^n_{i+a(n,p,h)} + A(p+1)^n_{i+a(n,p,h)} + \zeta(p+1)^n_{i,h}
\end{multline*}
Given $j,k,l,m=1,\cdots,d$, by Table \ref{ingredient.estimates.psd} we have
\begin{equation*}
\left| E^n_i\big(\beta'^{n,jk}_i\beta'^{n,lm}_i\big) - (k_n\Delta_n^{1/2+\delta})^{-1}\Sigma(c^n_i)^{jk,lm} \right| = \sum_{r=1}^5\mu'^{n,r}_i + O_p\big( p(k_n^2\Delta_n^{1+2\delta})^{-1} \big)
\end{equation*}
where
\begin{eqnarray*}
	\mu'^{n,1}_i &=&
	\frac{1}{(k_n-l_n)^2\Delta_n^2}\sum_{h=0}^{m(n,p)-1}E^n_i\left(\left|\alpha(p)'^{n,jk}_{i,h}\alpha(p)'^{n,lm}_{i,h} - \zeta(p+1)^{n,jk}_{i,h}\zeta(p+1)^{n,lm}_{i,h}\right|\right) \\
	\mu'^{n,2}_i &=&
	\frac{1}{(k_n-l_n)^2\Delta_n^2}\sum_{h=0}^{m(n,p)-2}\sum_{h'=h+1}^{m(n,p)-1}\left|E^n_i\left[\alpha(p)'^{n,jk}_{i,h}\alpha(p)'^{n,lm}_{i,h'} + \alpha(p)'^{n,lm}_{i,h}\alpha(p)'^{n,jk}_{i,h'}\right]\right| \\
	\mu'^{n,3}_i &=& \frac{1}{(k_n-l_n)^2\Delta_n^2}\sum_{h=0}^{m(n,p)-1}E^n_i\left(\left|\zeta(p+1)^{n,jk}_{i,h}\zeta(p+1)^{n,lm}_{i,h}\right.\right.\\
	&&\hspace{30mm} \left.\left. - \Delta_n^{1/2-\delta}\frac{(p+1)\Phi_{00}-\Psi_{00}}{(p+2)\Phi_{00}}\,\Sigma\big(c^n_{i+a(n,p,h)}\big)^{jk,lm}(p+2)l_n\Delta_n\right|\right)
\end{eqnarray*}
\begin{eqnarray*}
	\mu'^{n,4}_i &=& \frac{[(p+1)\Phi_{00}-\Psi_{00}]\,l_n}{\Phi_{00}(k_n-l_n)^2\Delta_n^{1/2+\delta}}\sum_{h=0}^{m(n,p)-1}\left|E^n_i\left[\Sigma\big(c^n_{i+a(n,p,h)}\big)^{jk,lm} - \Sigma(c^n_i)^{jk,lm}\right]\right| \\
	\mu'^{n,5}_i &=& \frac{(p+1)\Phi_{00}-\Psi_{00}}{(p+1)\Phi_{00}}\left|\frac{(p+1)l_n}{(k_n-l_n)^2\Delta_n^{1/2+\delta}}\left\lfloor\frac{k_n}{(p+1)l_n}\right\rfloor-\frac{1}{k_n\Delta_n^{1/2+\delta}}\right|\cdot\left|\Sigma(c^n_i)^{jk,lm}\right|\\
	&&\hspace{44mm} + \frac{1}{k_n\Delta_n^{1/2+\delta}}\left|\frac{(p+1)\Phi_{00}-\Psi_{00}}{(p+1)\Phi_{00}}-1\right| \cdot\left|\Sigma(c^n_i)^{jk,lm}\right|
\end{eqnarray*}
Using Table \ref{ingredient.estimates.psd} to analyze $\mu'^{n,r}_i$ for $r=1,\cdots,5$ one by one, we get
\begin{equation}\label{est.beta2.psd}
	\left| E^n_i\big(\beta'^{n,jk}_i\beta'^{n,lm}_i\big) - (k_n\Delta_n^{1/2+\delta})^{-1}\Sigma(c^n_i)^{jk,lm} \right| \le K\big[k_n\Delta_n + (pk_n\Delta_n^{1/2+\delta})^{-1}\big]
\end{equation}

Define
\begin{eqnarray*}
	\lambda'(x) &=& \sum^d_{j,k,l,m=1}\partial^2_{jk,lm}g(x)\times\Sigma(x)^{jk,lm} \\
	g'_n(x)     &=& g(x) - (2k_n\Delta_n^{1/2+\delta})^{-1}\lambda'(x)
\end{eqnarray*}
We have a decomposition analogous to (\ref{decomp}):
\begin{equation}\label{decomp.psd}
	\Delta_n^{-1/4+\delta/2}\big[\widetilde{S}(g)^n-S(g)\big] = \check{S}^{n,0} + \check{S}^{n,1} + \check{S}(p)^{n,2} + \check{S}^{n,3} + \check{S}(p)^{n,4}
\end{equation}
where
\begin{eqnarray*}
	\check{S}^{n,0}_t&=&\Delta_n^{-1/4+\delta/2}\Bigg[\sum_{i=0}^{N^n_t-1}\int_{ik_n\Delta_n}^{(i+1)k_n\Delta_n}g(c^n_{ik_n})-g(c_s)\ds s-\int_{N^n_tk_n\Delta_n}^tg(c_s)\ds s\Bigg]\\
	\check{S}^{n,1}_t&=&k_n\Delta_n^{3/4+\delta/2}\sum_{i=0}^{N^n_t-1}\Big[ g'_n(\widetilde{c}^n_{ik_n}) - g'_n(\widetilde{c}^{*n}_{ik_n}) \Big]\\
	\check{S}(p)^{n,2}_t&=&k_n\Delta_n^{3/4+\delta/2}\sum_{i=0}^{N^n_t-1}\sum^d_{j,k=1}\partial_{jk}g(c^n_{ik_n})\left[\xi^{n,0}_i + \xi^{n,1}_i + \xi'^{n,2}_i+N(p)^{n,jk}_{ik_n}\right]
\end{eqnarray*}
\begin{eqnarray*}
	\check{S}^{n,3}_t&=&k_n\Delta_n^{3/4+\delta/2}\sum_{i=0}^{N^n_t-1}\Big[g(\widetilde{c}^{*n}_{ik_n})-g(c^n_{ik_n})-\sum^d_{j,k=1}\partial_{jk}g(c^n_{ik_n}) \beta'^{n,jk}_{ik_n}\\
	&&\hspace{66mm} - (2k_n\Delta_n^{1/2+\delta})^{-1}\lambda'(\widetilde{c}^{*n}_i)\Big]\\
	\check{S}(p)^{n,4}_t&=&k_n\Delta_n^{3/4+\delta/2}\sum_{i=0}^{N^n_t-1}\sum^d_{j,k=1}\partial_{jk}g(c^n_{ik_n})\times M(p)^{n,jk}_{ik_n}
\end{eqnarray*}

As in the proof of Theorem 1, by the argument for Lemma 4.4 in \cite{jr13},
\begin{equation}\label{v0.psd}
	\check{S}^{n,0}\overset{u.c.p.}{\longrightarrow}0
\end{equation}

According to (\ref{g.cond.strong}),
\begin{equation*}
	\|g'_n(x)-g'_n(y)\| \le K\|x-y\| + K(k_n\Delta_n^{1/2+\delta})^{-1}\|x-y\|\left(\|x\|^2 + \|x-y\|^2\right)
\end{equation*}
and for large $n$, $\|g'_n(x)-g'_n(y)\| \le K\|x-y\|$. By (\ref{est.chat-chat*.psd})
\begin{multline*}
	\E\left(\sup_{u\in[0,t]}\big\|\check{S}^{n,1}_u\big\|\right) \le Kk_n\Delta_n^{3/4+\delta/2}\sum_{i=0}^{N^n_t-1}\left\|\widetilde{c}^n_{ik_n} - \widetilde{c}^{*n}_{ik_n}\right\|\\
	\le Kt\left[a_n\Delta_n^{(2-\nu)\rho + (1+2\delta)\nu/4 - (3/4+\delta/2)} + \Delta_n^{3/4+\delta/2}\right]
\end{multline*}
Because of (\ref{tuning.psd}), $(2-\nu)\rho + (1+2\delta)\nu/4 - (3/4+\delta/2) > 0$, we have
\begin{equation}\label{v1.psd}
	\check{S}^{n,1}\overset{u.c.p.}{\longrightarrow}0
\end{equation}

Given $e^n_i\in\R^{d\times d}$, consider the process
\begin{equation*}
	\check{S}^{*n}_t = k_n\Delta_n^{3/4+\delta/2} \sum_{i=0}^{N^n_t-1} \sum^d_{j,k=1}\partial_{jk}g(c^n_{ik_n})\times e^{n,jk}_{ik_n}
\end{equation*}
Assume the following holds
\begin{equation}\label{xi.cond.psd}
\left\{\begin{array}{lcl}
	\left\| E^n_i\left(e^n_{i}\right) \right\| &\le& K\Delta_n^{1/4-\delta/2} a_n \\
	E^n_i\left( \|e^n_{i}\|^2\right) &\le& K(k_n\Delta_n^{1/2+\delta})^{-1} b_n
\end{array}\right.
\end{equation}
where $a_n,b_n\to0$. Since $\partial g$ is bounded by (\ref{SA-v}),
\begin{multline*}
	\E\left(\sup_{u\in[0,t]} \big\|\check{S}^{*n}_u\big\|\right) \le Kk_n\Delta_n^{3/4+\delta/2}\sum_{i=0}^{N^n_t-1}\E\left(\left\|E^n_{ik_n}\left(e^n_{ik_n}\right)\right\|\right) \\ 
	+ Kk_n\Delta_n^{3/4+\delta/2}\E\left(\sup_{s\in[0,t]}\left\|\sum_{i=0}^{N^n_s-1}\left[ e^n_{ik_n} - E^n_{ik_n}\left(e^n_{ik_n}\right) \right]\right\|\right)
\end{multline*}
note that $k_n\Delta_nN^n_t\asymp t$, by Lemma \ref{Doob.max.ineq}, $\E\left(\sup_{u\in[0,t]} \big\|\check{S}^{*n}_u\big\|\right) \le K\left(ta_n + \sqrt{tb_n}\right)\to0$.

We consider the following 4 cases,
\begin{itemize}
\item[\textcircled{1}] when $e^n_i=\xi^{n,0}_i$: by (\ref{est.xi(0)n}), (\ref{xi.cond.psd}) is satisfied with $a_n=k_n\Delta_n^{3/4+\delta/2},\,b_n=(k_n\Delta_n^{3/4+\delta/2})^2$;
\item[\textcircled{2}] when $e^n_i=\xi^{n,1}_i$: by (\ref{est.xi(1)n.psd}), (\ref{xi.cond.psd}) is satisfied with $a_n=\Delta_n^{1/4-\delta/2}$, $b_n=k_n\Delta_n$;
\item[\textcircled{3}] when $e^n_i=\xi'^{n,2}_i$: by (\ref{est.xi(2)n.psd}), (\ref{xi.cond.psd}) is satisfied with $a_n=\Delta_n^{5/(2\delta)-1/4}$, $b_n=k_n\Delta_n^{1/2+5\delta}$;
\item[\textcircled{4}] when $e^n_i=N(p)^n_i$: by (\ref{est.N(p).psd}), (\ref{xi.cond.psd}) is satisfied with $a_n=p^{-1}\Delta_n^{1/4-\delta/2}$, $b_n=p^{-1}$.
\end{itemize}
Therefore, provided we abide by the tuning choice (\ref{tuning.psd}) and let $p\asymp\Delta_n^{-1/12+\delta/6}$, we have
\begin{equation}\label{v2.psd}
	\check{S}(p)^{n,2}\overset{u.c.p.}{\longrightarrow}0
\end{equation}

Let $\iota'^n_i = \lambda'(\widehat{c}^{*n}_i) - \lambda'(c^n_i)$. By the definition (4.2) in the main article,
\begin{equation*}
	\|\iota'^n_i\| \le K\left(\|\beta'^n_i\| + \|\beta'^n_i\|^2\right)
\end{equation*}
then by (\ref{tuning.psd}), (\ref{est.beta.psd}),
\begin{equation}\label{est.eta.psd}
	E^n_i(\|\iota'^n_i\|) \le K k_n^{-1/2} \Delta_n^{-1/4-\delta/2}
\end{equation}

We can rewrite $\check{S}^{n,3}$ as
	\[\check{S}^{n,3} = G'^n + H'^n\]
where
\begin{eqnarray*}
	G'^n_t &=& k_n\Delta_n^{3/4+\delta/2} \sum_{i=0}^{N^n_t-1}\left[s'^n_{ik_n} + \frac{1}{2k_n\Delta_n^{1/2+\delta}}\iota'^n_{ik_n} + E^n_{ik_n}\left(v'^n_{ik_n}\right)\right]\\
	H'^n_t &=& k_n\Delta_n^{3/4+\delta/2} \sum_{i=0}^{N^n_t-1}\left[v'^n_{ik_n}-E^n_{ik_n}\left(v'^n_{ik_n}\right)\right]
\end{eqnarray*}
\begin{eqnarray*}
	s'^n_i &=& g(c^n_i+\beta'^n_i)-g(c^n_i)-\sum^d_{j,k=1}\partial_{jk}g(c^n_i)\beta'^{n,jk}_i - \frac{1}{2}\sum^d_{j,k,l,m=1}\partial^2_{jk,lm}g(c^n_i)\beta'^{n,jk}_i\beta'^{n,lm}_i\\
	v'^n_i&=&\frac{1}{2}\sum^d_{j,k,l,m=1}\partial^2_{jk,lm}g(c^n_i)\left[\beta'^{n,jk}_i\beta'^{n,lm}_i - (k_n\Delta_n^{1/2+\delta})^{-1}\Sigma(c^n_i)^{jk,lm}\right]
\end{eqnarray*}

By (\ref{g.cond.strong}), (\ref{est.beta.psd}), (\ref{est.eta.psd}), if we let $p\asymp\Delta_n^{-12+\delta/6}$, 
\begin{equation}\label{est.G.psd}
	\E\left(\sup_{s\in[0,t]}\|G'^n_s\|\right) \le Kt\left[k_n\Delta_n^{3/4+\delta/2} + (k_n\Delta_n^{2(1+\delta)/3})^{-1} + \Delta_n^{25\delta/12-1/24}\right]
\end{equation}
and
\begin{multline*}
	E^n_i(\|v'^n_i\|^2) \le K E^n_i\left[\|\beta'^n_i\|^4 + (k_n\Delta_n^{1/2+\delta})^{-1}\|\beta'^n_i\|^2 + (k_n\Delta_n^{1/2+\delta})^{-2}\right] \\
	\le K\left[k_n\Delta_n + (k_n\Delta_n^{1/2+\delta})^{-2}\right]
\end{multline*}
then Lemma \ref{Doob.max.ineq} implies
\begin{equation}\label{est.H.psd}
	\E\left(\sup_{s\in[0,t]}\|H'^n_s\|\right)\le K\sqrt{t}\left[k_n\Delta_n^{3/4+\delta/2} + (k_n\Delta_n^{1/2+\delta})^{-1/2}\right] 
\end{equation}

Because of (\ref{tuning.psd}), (\ref{est.G.psd}), (\ref{est.H.psd}), we have
\begin{equation}\label{v3.psd}
	\check{S}^{n,3}\overset{u.c.p.}{\longrightarrow}0
\end{equation}

Recall the definition (\ref{def.Lambda}), rewrite $\check{S}(p)^{n,4}_t$ as
\begin{equation*}
	\check{S}(p)^{n,4}_t = \frac{k_n}{k_n-l_n}\sum^d_{j,k=1}\Delta_n^{-1/4+\delta/2}\sum^{N^n_t-1}_{i=0}\sum_{h=0}^{m(n,p)-1}\Lambda(p)^{n,jk}_{i,h}
\end{equation*}

By the last line of Table \ref{ingredient.estimates.psd} and Jensen's inequality,
\begin{eqnarray}
	\Delta_n^{-1/4+\delta/2}\sum_{i=0}^{N^n_t-1}\sum_{h=0}^{m(n,p)-1}\left\|E\left[\Lambda(p)^n_{i,h}|\mathcal{H}(p)^n_{i,h}\right]\right\| &\le& Kp^{1/2}\Delta_n^{1/4-\delta} \label{S4.c1.psd}\\
	\Delta_n^{-1+2\delta}\sum_{i=0}^{N^n_t-1}\sum_{h=0}^{m(n,p)-1}E\left[\left\|\Lambda(p)^n_{i,h}\right\|^4|\mathcal{H}(p)^n_{i,h}\right] &\le& Kp\Delta_n^{1/2-\delta} \label{S4.c2.psd}
\end{eqnarray}
The same argument resulting in (5.58) in \cite{j09} gives us
\begin{equation}\label{S4.c3.psd}
	\Delta_n^{-1/4+\delta/2}\sum_{i=0}^{N^n_t-1}\sum_{h=0}^{m(n,p)-1}\left\|E\big[\Lambda(p)^n_{i,h}\Delta N(p)^n_{i,h}|\mathcal{H}(p)^n_{i,h}\big]\right\| \overset{\mathbb{P}}{\longrightarrow} 0
\end{equation}
where $\Delta N(p)^n_{i,h}=N^n_{ik_n+b(n,p,h)}-N^n_{ik_n+a(n,p,h)}$ and $N$ is a bounded martingale orthogonal to $W$ or $N=W^l$ for some $l=1,\cdots,d'$.

Calculation similar to that of $\mu'^{n,3}_i$ shows
\begin{multline}\label{S4.c4.psd}
	\Delta_n^{-1/2+\delta}\sum_{i=0}^{N^n_t-1}\sum_{h=0}^{m(n,p)-1} E\left[\Lambda(p)^{n,jk}_{i,h}\,\Lambda(p)^{n,lm}_{i,h}|\mathcal{H}(p)^n_{i,h}\right] = \\
	\sum_{i=0}^{N^n_t-1}\partial_{jk}g(c^n_{ik_n})\partial_{lm}g(c^n_{ik_n}) \times \\
	\sum_{h=0}^{m(n,p)-1}\frac{p\Phi_{00}-\Psi_{00}}{(p+1)\Phi_{00}}\,\Sigma\big(c^n_{ik_n+a(n,p,h)}\big)^{jk,lm}(p+1)l_n\Delta_n + tp^{3/2}\,O_p(\Delta_n^{1/4-\delta/2}) \\
	\overset{\mathbb{P}}{\longrightarrow} \int_0^t\partial_{jk}g(c_s)\partial_{lm}g(c_s)^\T\,\Sigma(c_s; p)^{jk,lm}\,\mathrm{d}s
\end{multline}

To establish stable convergence, we then only need to combine (\ref{S4.c1.psd}), (\ref{S4.c2.psd}), (\ref{S4.c3.psd}), (\ref{S4.c4.psd}) and apply Theorem IX.7.28 in \cite{js03}. We get for $\forall p\in\mathbb{N}^+$,
\begin{equation}\label{v4.psd}
	\check{S}(p)^{n,4}\overset{\mathcal{L}-s(f)}{\longrightarrow}Z(p)
\end{equation}
where $Z(p)$ is a process defined on an extension of the space $\left(\Omega,\mathcal{F},(\mathcal{F}_t),\mathbb{P}\right)$,
such that conditioning on $\mathcal{F}$ it is a  mean-0 continuous It\^o martingale with variance
	\[\widetilde{E}[Z(p)Z(p)^\T|\F]=\int_0^t\sum^d_{j,k,l,m=1}\partial_{jk}g(c_s)\partial_{lm}g(c_s)^\T\,\frac{p\Phi_{00}-\Psi_{00}}{(p+1)\Phi_{00}}\,\Sigma(c_s)^{jk,lm}\ds s\]
where $\widetilde{E}$ is the conditional expectation operator on the extended probability space.

By (\ref{v0}), (\ref{v1.psd}), (\ref{v2.psd}), (\ref{v3.psd}), (\ref{v4.psd}), as $p\to\infty$, we arrive at the asymptotic result of Theorem 2 in the main article.

\section{Derivation for Proposition 3 and 4}\label{apdx:pca}
In the following lemmas and derivations, we suppress the time index in the subscripts to reduce notational clutter. From now on, we write $c=c_t$, $\lambda^r = \lambda^r_t$, $q^r = q^r_t$ is the eigenvector associated with eigenvalue $\lambda^r$, $q^r_k = q^r_{t,k}$ is the $k$-th entry in $q^r$, and the eigenvalue factorization is written as
\begin{equation}\label{eigen.sim}
c = \left[q^1,\cdots,q^d \right]
\left[\begin{array}{ccc}
	\lambda^1 & & \\
	& \ddots & \\
	& & \lambda^d
\end{array}\right]
\left[\begin{array}{c}
q^{1,\T} \\ \vdots \\ q^{d,\T}
\end{array}\right]
\end{equation}
where $\lambda_1\ge \cdots \ge \lambda_d$.

The following lemma is a restatement of Lemma 2 in \cite{ax19a}.
\begin{lem}\label{lem.derivatives.simple}
When $\lambda^r$ is a simple eigenvalue in (\ref{eigen.sim}), the functions $\lambda^r(c)$ and $q^r(c)$ are twice continuously differentiable. We have the following result about gradients:
\begin{eqnarray*}
	\partial_{jk}\lambda^r(c) &=& q^r_j\, q^r_k \\
	\partial_{jk}q^r(c) &=& \big(\lambda^r\,\I_d - c\big)^\dagger \times \partial_{jk}c \times q^r
\end{eqnarray*}
where the superscript notation $\dagger$ refers to Moore-Penrose inverse, equivalently,
\begin{equation*}
	\partial_{jk}q^r(c) = \sum_{v\ne r} \frac{q^v_j\, q^r_k}{\lambda^r-\lambda^v} \times q^v
\end{equation*}
Moreover, the Hessians are
\begin{equation*}
	\partial^2_{jk,lm}\lambda^r(c) = \sum_{v\ne r} \frac{q^v_j\, q^v_l\, q^r_k\, q^r_m +  q^v_k\, q^v_l\, q^r_j\, q^r_m}{\lambda^r - \lambda^v}
\end{equation*}
and
\begin{multline*}
	\partial^2_{jk,lm}q^r(c) = \sum_{v\ne r} \frac{q^v_l\, q^v_m\, q^v_j\, q^r_k - \, q^v_j\, q^r_k\, q^r_l\, q^r_m}{(\lambda^r-\lambda^v)^2}\times q^v
	+ \sum_{v\ne r}\sum_{b\ne v} \frac{q^b_l\, q^v_m\, q^v_j\, q^r_k}{(\lambda^r-\lambda^v)(\lambda^v-\lambda^b)}\times q^b \\
	+ \sum_{v\ne r}\sum_{b\ne v} \frac{q^b_j\, q^b_l\, q^v_m\, q^r_k}{(\lambda^r-\lambda^v)(\lambda^v-\lambda^b)}\times q^v
	+ \sum_{v\ne r}\sum_{b\ne r} \frac{q^b_k\, q^b_l\, q^r_m\, q^v_j}{(\lambda^r-\lambda^v)(\lambda^r-\lambda^b)}\times q^v
	\end{multline*}
\end{lem}

The following lemmas is a direct consequence Lemma 4 in \cite{ax19a}.
\begin{lem}\label{lem.derivatives.repeated}
Suppose the eigenvalues of $c$ form $K$ clusters in the sense that
\begin{equation*}
	\lambda^{r_0+1} =\cdots= \lambda^{r_1} > \lambda^{r_1+1} =\cdots= \lambda^{r_2} > \cdots\cdots > \lambda^{r_{K-1}+1} =\cdots= \lambda^{r_K}
\end{equation*}
where $r_0=0,\, r_K=d$. We define the following functions on $\R^d$
\begin{equation*}
	f_h(x) = \frac{1}{r_h-r_{h-1}} \sum_{j=r_{h-1}+1}^{r_h} x^j,\; h=1,\cdots,K
\end{equation*}
Define function $\lambda(c)=(\lambda^1, \cdots, \lambda^d)$, and functions $F_h=f_h\circ\lambda,\; h=1,\cdots,K$, where $\circ$ denotes function composition. Then we have
\begin{eqnarray*}
	\partial_{jk} F_h(c) &=& \frac{1}{r_h-r_{h-1}} \sum_{r=r_{h-1}+1}^{r_h} q^r_j\, q^r_k \\
	\partial^2_{jk,lm} F_h(c) &=& \frac{1}{r_h-r_{h-1}} \sum_{r=r_{h-1}+1}^{r_h} \Big[\big(\lambda^r\I_d-c\big)^{\dagger,km}\, q^r_j\, q^r_l + \big(\lambda^r\I_d-c\big)^{\dagger,jl}\, q^r_k\,q^r_m \Big]
\end{eqnarray*}
where the superscript notation $\dagger$ stands for Moore-Penrose inverse, equivalently,
\begin{equation*}
	\partial^2_{jk,lm} F_h(c) = \frac{1}{r_h-r_{h-1}} \sum_{r=r_{h-1}+1}^{r_h} \sum_{\lambda^v\ne\lambda^r} \frac{q^r_j\, q^r_l\, q^v_k\, q^v_m + q^v_j\, q^v_l\, q^r_k\, q^r_m}{\lambda^r - \lambda^v}
\end{equation*}
\end{lem}

\begin{remk}\label{remk.spectral.differentiability}
	If $r_h = r_{h-1}+1$, $F_h(c) = \lambda^{r_h}(c)$. From Lemma \ref{lem.derivatives.simple}, \ref{lem.derivatives.repeated} we see that $\lambda^r(c)$ is first-order differentiable regardless of $\lambda^r$ being a simple or repeated eigenvalue of $c$. However, when $\lambda^r$ is a repeated eigenvalue, $\lambda^r(c)$ is not second-order differentiable; as for the eigenvector, $q^r(c)$ is not differentiable at all.
\end{remk}

By some calculation, we can see that Proposition 3 and 4 are consequences of Theorem 2 in the main article and Lemma \ref{lem.derivatives.simple}, \ref{lem.derivatives.repeated}.

\section{Simulations}
\subsection{Simulation of the rate-optimal CLT}
We simulate the following noisy observations of scalar jump-diffusion with stochastic volatility via the Euler scheme:
\begin{equation*}
\left\{\begin{array}{lcl}
	Y^n_i        &=&X^n_i + \e^n_i \\
	\mathrm{d}X_t&=&.03\ds t + \sqrt{c_t}\ds W_t + J^X_t\ds N^X_t\\
	\mathrm{d}c_t&=&6(.16-c_t)\ds t + .5\sqrt{c_t}\ds B_t + \sqrt{c_{t-}}J^c_t\ds N^c_t
\end{array}\right.
\end{equation*}
where $\e^n_i\overset{\text{i.i.d.}}{\sim}N(0,.005^2)$, $\E[(W_{t+\Delta}-W_t)(B_{t+\Delta}-B_t)]=-.6\Delta$, $J^X_t\sim N(-.01,.02^2)$, $N^X_{t+\Delta}-N^X_t\sim\mathrm{Poisson}(36\Delta)$, $\log(J^c_t)\sim N(-5,.8)$, $N^c_{t+\Delta}-N^c_t\sim\mathrm{Poisson}(12\Delta)$. 

Each simulation employs $23400\times21$ data points with $\Delta_n=1$ second. We choose the following tuning parameters:
\begin{table}[H]
	\begin{tabular}{l|lll}
		functionals    & $l_n$ & $k_n$ & $\nu_n$ \\
		\hline
		$g(c)=c^2$     & $\lfloor\Delta_n^{-.5}\rfloor$ & $\lfloor\Delta_n^{-.69}\rfloor$ & $1.6\overline{\sigma}^2\Delta_n^{.47}$ \\
		$g(c)=\log(c)$ & $\lfloor\Delta_n^{-.5}\rfloor$ & $\lfloor\Delta_n^{-.7}\rfloor$  & $1.5\overline{\sigma}^2\Delta_n^{.47}$
	\end{tabular}
\end{table}
where $\overline{\sigma}^2$ is an estimate of the average volatility by bipower variation \cite{pv09a}.

We compute estimation errors in the simulation, and normalize the errors by the corresponding asymptotic variances. The results are shown in Figure \ref{MC}.
\begin{figure}[H]
	\centering
	\caption{Simulation of estimators of scalar volatility functionals}\label{MC}
	\includegraphics[width=.5\textwidth]{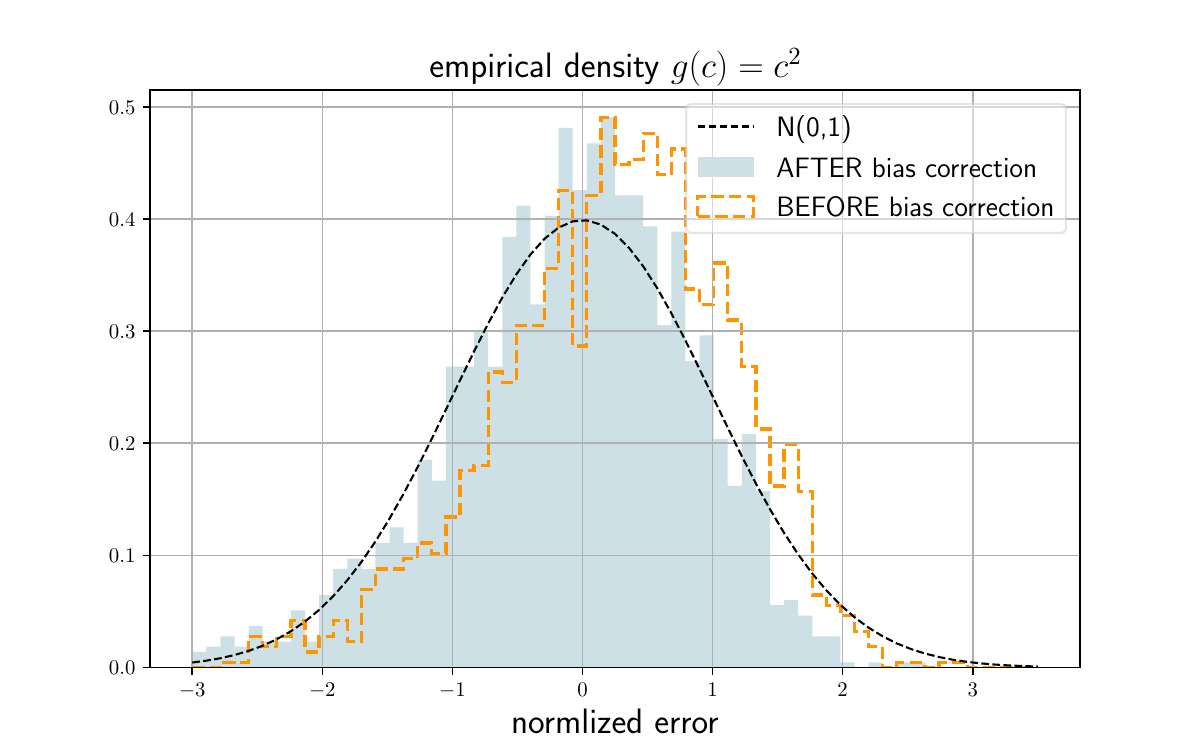}\includegraphics[width=.5\textwidth]{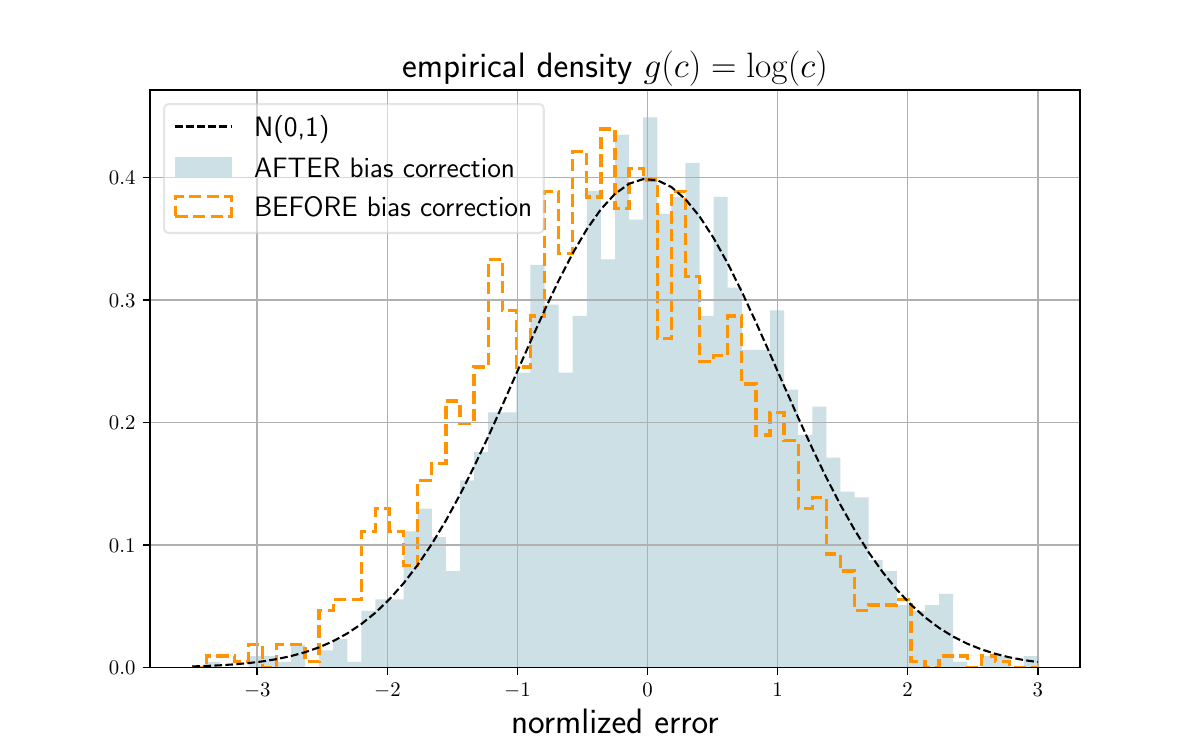}
\end{figure}

\subsection{Simulation of PCA of a factor model}
We simulate the following factor model:
\begin{eqnarray*}\label{def.factor.model}
	Y_t &=& X_t + e_t \nonumber \\
	X_t &=& X_0 + \int_0^t\beta_s\ds F_s + Z_t,
\end{eqnarray*}
where $\beta_s\in\R^{d\times r}$ is the factor loading matrix with $r\le d$, $F_s\in\R^r$ is the latent factor, $Z_t\in\R^{d}$ is idiosyncratic component satisfying $Z\indep F$, and $e_t$ is a multidimensional noise process exhibits no autocorrelation but is possibly cross-sectionally correlated.

In the component-wise differential form, we can write $X_t$ as
\begin{eqnarray*}
	\mathrm{d}X^j_t &=& \sum_{k=1}^r \beta^{jk}_t\ds F^k_t + \mathrm{d}Z^j_t \\
	\mathrm{d}\beta^{j1}_t &=& \kappa^{j1}\big(\theta^{j1}-\beta^{j1}_s\big)\ds s + \xi^{j1}\sqrt{\beta^{j1}_t}\ds\breve{W}^{j1}_t \\
	\mathrm{d}\beta^{jk}_t &=& \kappa^{jk}\big(\theta^{jk}-\beta^{jk}_s\big)\ds s + \xi^{jk}\ds\breve{W}^{jk}_t,\; k\ge2 \\
	\mathrm{d}F^k_t &=& \mu^k\ds t + \sqrt{\Pi^k_t}\ds W^k_t + J^{F,k}_t\ds N^{F,k}_t \\
	\mathrm{d}Z^j_t &=& \chi_t\ds B^j_t + J^{Z,j}_t\ds N^{Z,j}_t,
\end{eqnarray*}
and
\begin{eqnarray*}
	\mathrm{d}\Pi^k_t &=& \widetilde{\kappa}^k\big(\widetilde{\theta}^k-\Pi^k_t\big)\ds t + \eta^k\sqrt{\Pi^k_t}\ds\widetilde{W}^k_t + J^{\Pi,k}_t\ds N^{F,k}_t \\
	\mathrm{d}\chi_t^2 &=& \kappa\big(\theta-\chi_t^2\big)\ds t + \eta\chi_t\ds\widetilde{B}_t,
\end{eqnarray*}
where $W$ and $\widetilde{W}$ are $\R^r$-valued standard Brownian motions, $\breve{W}$ is a $\R^{d\times r}$-valued Brownian motion with independent elements, $B$ is a $\R^d$-valued standard Brownian motion, $\widetilde{B}$ is a scalar standard Brownian motion. Except $W$ and $\widetilde{W}$, all Brownian motions in this model are mutually independent.

In this simulation model, the factor loading, factor volatility and idiosyncratic volatility are modeled by Ornstein-Uhlenbeck processes so that they exhibit mean-reverting dynamics. Especially, the first factor is presumably the market factor, hence the associated factor loading $\beta^{j1}_t\, j=1,\cdots,d$ are positive almost surely.

Moreover, for $k=1,\cdots,r$, $\E\big[(W^k_{t+\Delta}-W^k_t)(\widetilde{W}^k_{t+\Delta}-\widetilde{W}^k_t)\big] = \rho^k\Delta$ is the leverage effect in the $k$-th factor; $N^{F,k}_{t+\Delta}-N^{F,k}_t\sim\mathrm{Poisson}(\lambda^{F,k}\Delta)$ is the count of co-jumps of the factor and factor volatility, $J^{F,k}_t\overset{\mathrm{i.i.d.}}{\sim}\mathrm{Laplace}(\zeta^{F,k})$ is the size of factor jump, $J^{\Pi  ,k}_t\overset{\mathrm{i.i.d.}}{\sim}\mathrm{Exponential}(\zeta^{\Pi,k})$ is the size of volatility jump; for $j=1,\cdots,d$, $N^{Z,j}_{t+\Delta}-N^{Z,j}_t\sim\mathrm{Poisson}(\lambda^{Z,j}\Delta)$ is the count of idiosyncratic jumps, $J^{Z,j}_t\overset{\mathrm{i.i.d.}}{\sim}\mathrm{Laplace}(\zeta^{Z,j})$ is the size of idiosyncratic jump.

We set $d=30,\,r=3$, and generate $22800\times21$ data points with $\Delta_n=1$ second in each simulation. We choose the following tuning parameters in PCA:
\begin{table}[H]
\begin{tabular}{l|ccc}
	functionals    & $l_n$ & $k_n$ & $\nu_n$ \\
	\hline
	$g(c)=\lambda^1(c)$ & $\lfloor.23\Delta_n^{-.62}\rfloor$ & $\lfloor.57\Delta_n^{-.75}\rfloor$ & $1.5\overline{c}\Delta_n^{.47}$ \\
	$g(c)=\lambda^2(c)$ & $\lfloor.038\Delta_n^{-.62}\rfloor$ & $\lfloor.19\Delta_n^{-.75}\rfloor$  & $1.5\overline{c}\Delta_n^{.47}$ \\
	$g(c)=q^{1,1}(c)$ & $\lfloor.23\Delta_n^{-.62}\rfloor$ & $\lfloor.57\Delta_n^{-.75}\rfloor$  & $1.5\overline{c}\Delta_n^{.47}$ \\
	$g(c)=q^{1,2}(c)$ & $\lfloor.23\Delta_n^{-.62}\rfloor$ & $\lfloor.57\Delta_n^{-.75}\rfloor$  & $1.5\overline{c}\Delta_n^{.47}$
\end{tabular}
\end{table}
where $\overline{c}$ is the vector of time-averaged estimates of volatility matrix diagonal elements computed by bipower variations. We conduct element-wise jump truncation by the rule $|\widebar{Y}^{r,n}_i|>1.5\overline{c}^r\Delta_n^{.47}$.

Figure \ref{MC.eigenval} and Figure \ref{MC.eigenvec} show empirical densities of studentized estimation errors.

\begin{figure}[H]
	\centering
	\caption{Simulation of eigenvalue estimators}\label{MC.eigenval}
	\includegraphics[width=.5\textwidth]{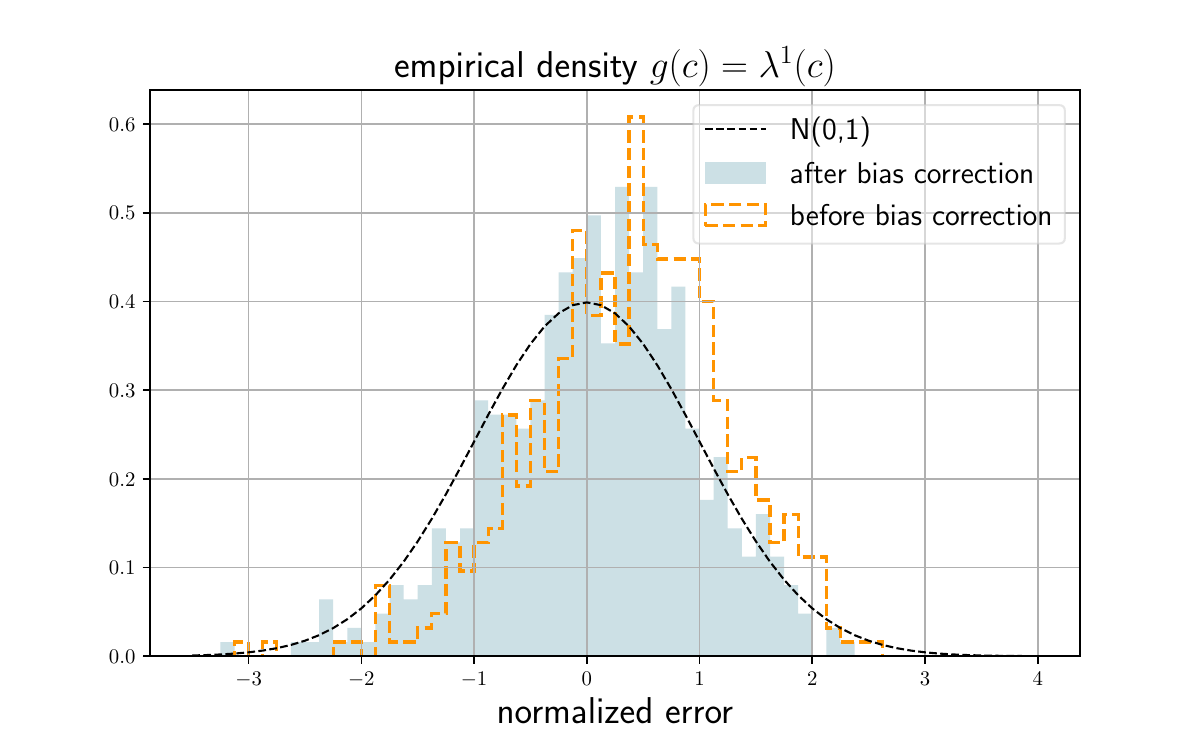}\includegraphics[width=.5\textwidth]{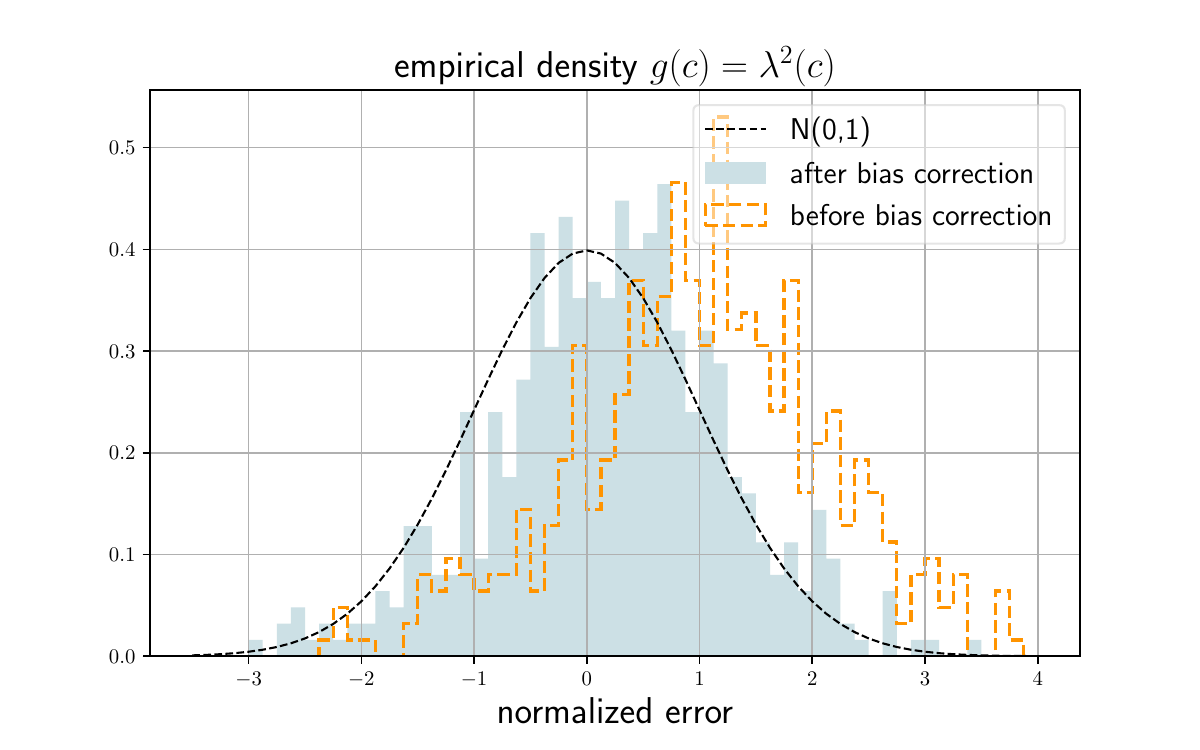}
\end{figure}

\begin{figure}[H]
	\centering
	\caption{Simulation of eigenvector estimators}\label{MC.eigenvec}
	\includegraphics[width=.5\textwidth]{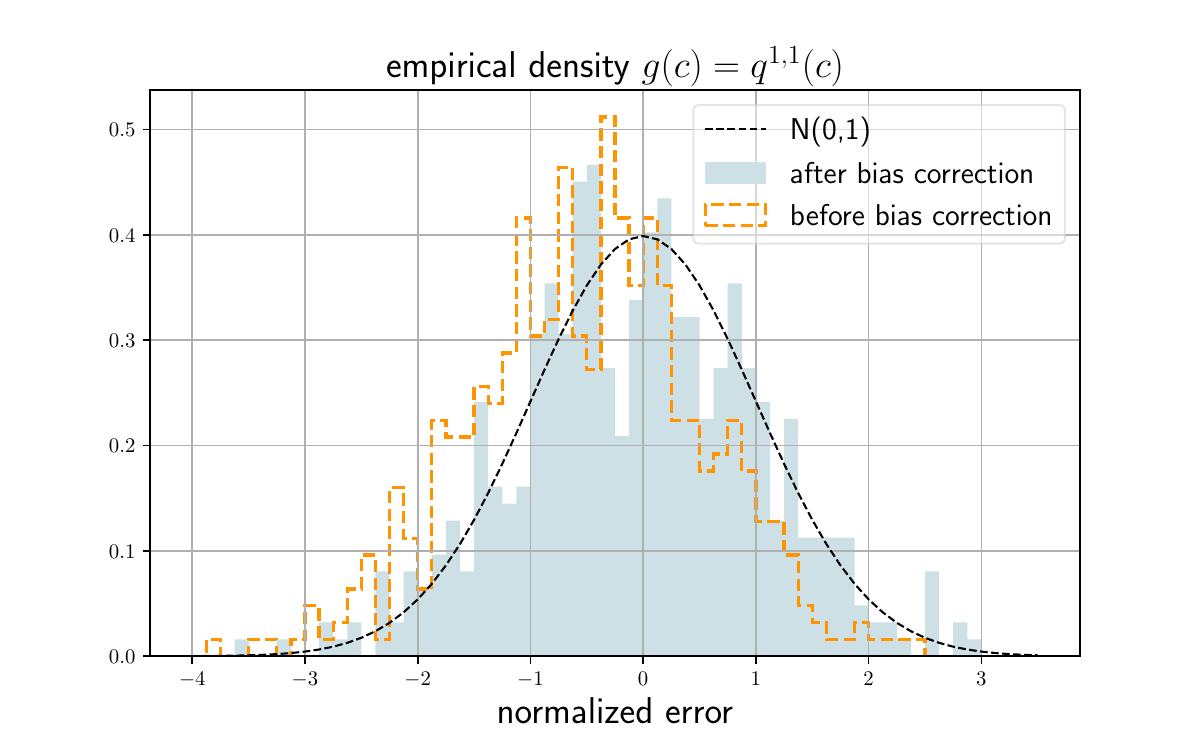}\includegraphics[width=.5\textwidth]{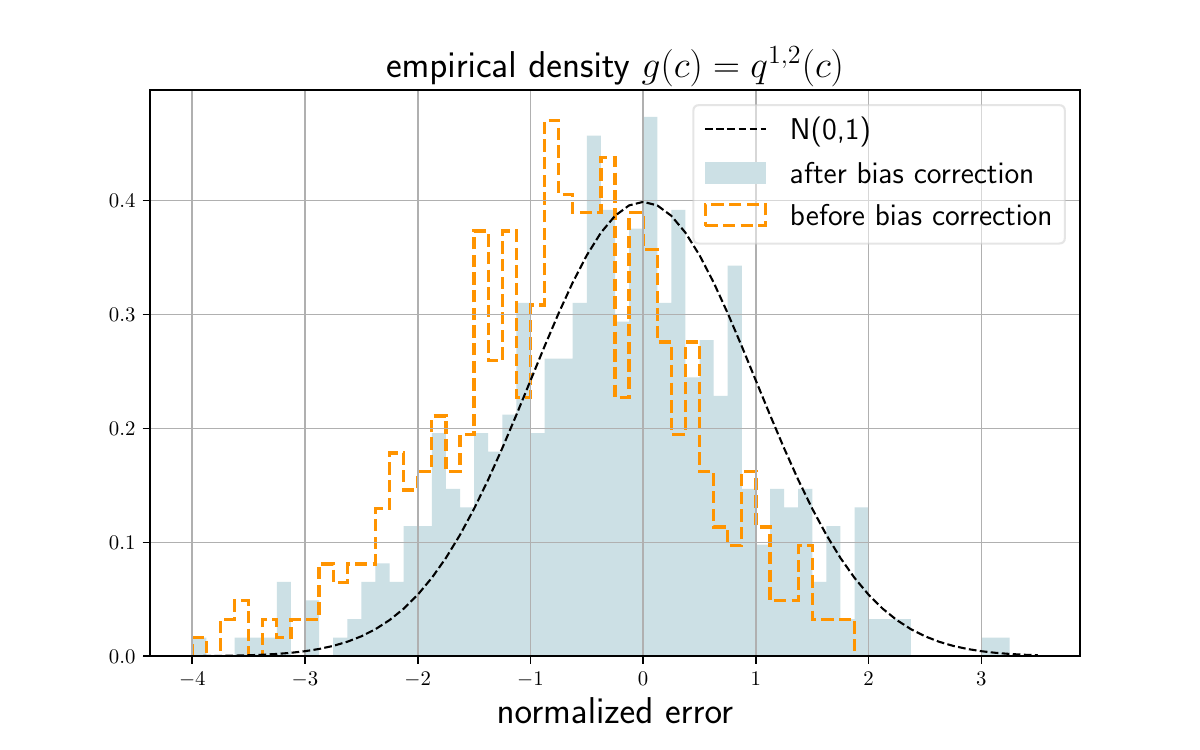}
\end{figure}


\begin{thebibliography}{52}
	\bibitem{tt12a}
\textsc{Todorov, V., Tauchen, G.} (2012).
The realized Laplace transform of volatility.
\textit{Econometrica} 80 (3), 1105-1127.

\bibitem{lx16}
\textsc{Li, J., Xiu, D.} (2016).
Generalized method of integrated moments for high-frequency data.
\textit{Econometrica} 84 (4), 1613-1633.

\bibitem{ltt17}
\textsc{Li, J., Todorov, V., Tauchen, G.} (2017).
Adaptive estimation of continuous-time regression models using high-frequency data.
\textit{Journal of Econometrics} 200, 36-47.

\bibitem{ax19a}
\textsc{A\"it-Sahalia, Y., Xiu, D.} (2019).
Principal component analysis of high-frequency data.
\textit{Journal of the American Statistical Association} 114 (525), 287-303.

\bibitem{llx19}
\textsc{Li, J., Liu, Y., Xiu, D.} (2019).
Efficient estimation of integrated volatility functionals via multiscale Jackknife.
\textit{The Annals of Statistics} 47 (1), 156-176.

\bibitem{y21}
\textsc{Yang, X.} (2021).
Semiparametric estimation in continuous-time: asymptotics for integrated volatility functionals with small and large bandwidths.
\textit{Journal of Business \& Economic Statistics} 39 (3), 793-806.

\bibitem{zma05}
\textsc{Zhang, L., Mykland, P., A\"it-Sahalia, Y.} (2005).
A tale to two scales: determining integrated volatility with noisy high-frequency data.
\textit{Journal of the American Statistical Association} 100 (472), 1394-1411.


\bibitem{z11}
\textsc{Zhang, L.} (2011).
Estimating covariation: Epps effect, microstructure noise.
\textit{Journal of Econometrics} 160, 33-47.

\bibitem{bhls08}
\textsc{Barndorff-Nielsen, O.E., Hansen, P.R., Lunde, A., Shephard, N.} (2008).
Designing realized kernels to measure the ex post variation of equity prices in the presence of noise.
\textit{Econometrica} 76 (6), 1481-1536.

\bibitem{bhls11}
\textsc{Barndorff-Nielsen, O.E., Hansen, P.R., Lunde, A., Shephard, N.} (2011).
Multivariate realised kernels: consistent positive semi-definite estimators of the covariation of equity prices with noise and non-synchronous trading.
\textit{Journal of Econometrics} 162, 149-169.

\bibitem{pv09b}
\textsc{Podolskij, M., Vetter, M.} (2009).
Estimation of volatility functionals in the simultaneous presence of microstructure noise and jumps.
\textit{Bernoulli} 15 (3), 634-658.

\bibitem{j09}
\textsc{Jacod, J., Li, Y., Mykland, P., Podolskij, M., Vetter, M.} (2009).
Microstructure noise in the continuous case: the preaveraging approach.
\textit{Stochastic processes and their applications}  119, 2249-2276.

\bibitem{jpv10}
\textsc{Jacod, J., Podolskij, M., Vetter, M.} (2010).
Limit theorems for moving averages of discretized processes plus noise.
\textit{The Annals of Statistics} 38 (3), 1478-1545.

\bibitem{ckp10}
\textsc{Christensen, K., Kinnebrock, S., Podolskij, M.} (2010).
Pre-averaging estimators of the ex-post covariance matrix in noisy diffusion models with non-synchronous data.
\textit{Journal of Econometrics} 159, 116-133.

\bibitem{x10}
\textsc{Xiu, D.} (2010).
Quasi-maximum likelihood estimation of volatility with high frequency data.
\textit{Journal of Econometrics} 159, 235-250.

\bibitem{afx10}
\textsc{A\"it-Sahalia, Y., Fan, J., Xiu, D.} (2010).
High-frequency covariance estimates with noisy and asynchronous financial data.
\textit{Journal of the American Statistical Association} 105 (492), 1504-1517.

\bibitem{sx17}
\textsc{Shephard, N., Xiu, D.} (2017).
Econometric analysis of multivariate realised QML: estimation of the covariation of equity prices under asynchronous trading.
\textit{Journal of Econometrics} 201, 19-42.

\bibitem{bhmr14}
\textsc{Bibinger, M., Hautsch, N., Malec, P., Rei}\ss{}, \textsc{M.} (2014).
Estimating the quadratic covariation matrix from noisy observations: local method of moments and efficiency.
\textit{The Annals of Statistics} 42 (4), 1312-1346.

\bibitem{bhmr19}
\textsc{Bibinger, M., Hautsch, N., Malec, P., Rei}\ss{}, \textsc{M.} (2019).
Estimating the spot covariation of asset prices - statistical theory and empirical evidence.
\textit{Journal of Business \& Economic Statistics} 37 (3), 419-435.

\bibitem{mz09}
\textsc{Mykland, P., Zhang, L.} (2009).
Inference for continuous semimartingales observed at high frequency.
\textit{Econometrica} 77 (5), 1403-1445.

\bibitem{jr13}
\textsc{Jacod, J., Rosenbaum, M.} (2013).
Quarticity and other functionals of volatility: efficient estimation.
\textit{The Annals of Statistics} 41 (3), 1462-1484.

\bibitem{c19}
\textsc{Chen, R.} (2019).
The Fourier transform method for volatility functional inference by asynchronous observations.
\hyperlink{https://arxiv.org/abs/1911.02205}{arXiv.1911.02205}.

\bibitem{tt11}
\textsc{Todorov, V., Tauchen, G., Grynkiv, I.} (2011).
Realized Laplace transforms for estimation of jump diffusive volatility models.
\textit{Journal of Econometrics} 164, 367-381.

\bibitem{tt12b}
\textsc{Todorov, V., Tauchen, G.} (2012).
Inverse realized Laplace transforms for nonparametric volatility density estimation in jump-diffusions.
\textit{Journal of the American Statistical Association} 107 (498), 622-635.

\bibitem{ltt16}
\textsc{Li, J., Todorov, V., Tauchen, G.} (2016).
Inference theory for volatility functional dependencies.
\textit{Journal of Econometrics} 193, 17-34.

\bibitem{y20}
\textsc{Yang, X.} (2020).
Time-invariant restrictions of volatility functionals: efficient estimation and specification tests.
\textit{Journal of Econometrics} 215, 486-516.

\bibitem{cmz20}
\textsc{Chen, D., Mykland, P., Zhang, L.} (2020).
The five trolls under the bridge: principal component analysis with asynchronous and noisy high frequency data.
\textit{Journal of the American Statistical Association} 115 (532), 1960-1977.

\bibitem{jr15}
\textsc{Jacod, J., Rosenbaum, M.} (2015).
Estimation of volatility functionals: the case of a $\sqrt{n}$ window.
In Friz, P. et al. \textit{Large Deviations and Asymptotic Methods in Finance}, Springer.


\bibitem{l13}
\textsc{Li, J.} (2013).
Robust estimation and inference for jumps in noisy high frequency data: a local-to-continuity theory for the pre-averaging method.
\textit{Econometrica} 81 (4), 1673-1693.

\bibitem{mz16}
\textsc{Mykland, P., Zhang, L.} (2016).
Between data cleaning and inference: pre-averaging and robust estimators of the efficient price.
\textit{Journal of Econometrics} 194, 242-262.

\bibitem{ax19b}
\textsc{A\"it-Sahalia, Y., Xiu, D.} (2019).
A Hausman test for the presence of market microstructure noise in high frequency data.
\textit{Journal of Econometrics} 211, 176-205.

\bibitem{p20}
\textsc{Pelger, M.} (2020).
Understanding systematic risk: a high-frequency approach.
\textit{Journal of Finance} Volume 75, Issue 4.

\bibitem{jp12}
\textsc{Jacod, J., Protter, P.} (2012).
\textit{Discretization of Processes}.
Springer-Verlag Berlin Heidelberg.



\bibitem{js03}
\textsc{Jacod, J., Shiryaev, A. N.} (2003).
\textit{Limit Theorems for Stochastic Processes}, 2ed.
Springer-Verlag Berlin Heidelberg.

\bibitem{pv09a}
\textsc{Podolskij, M., Vetter, M.} (2009).
Bipower-type estimation in a noisy diffusion setting.
\textit{Stochastic Processes and their Applications} 119, 2803-2831.

\end{thebibliography}

\end{document}